\documentclass{article}
\usepackage[whole]{bxcjkjatype}
\usepackage{graphicx}
\usepackage{braket}
\usepackage{amsmath}
\usepackage[margin=1in]{geometry}
\usepackage{amsfonts}
\usepackage[utf8]{inputenc}
\usepackage[T1]{fontenc}
\usepackage{verbatim}
\usepackage[sort&compress,numbers]{natbib}
\usepackage{url}
\usepackage{hyperref}
\hypersetup{colorlinks,linkcolor=blue,urlcolor=blue,citecolor=blue,hypertexnames=false}
\usepackage{graphicx}
\usepackage[T1]{fontenc}
\usepackage{amsthm}
\usepackage{amssymb}
\newtheorem{definition}{Definition}[subsection]
\newtheorem{proposition}{Proposition}[subsection]
\newtheorem{example}{Example}[subsection]
\newtheorem{claim}{Claim}[subsection]
\newtheorem{theorem}{Theorem}[section]

\newtheorem{lemma}{Lemma}[subsection]

\newtheorem{remark}{Remark}[subsection]
\usepackage{vwcol}
\usepackage{tikz}
\usepackage{tikz-cd}
\usepackage{mathtools}
\usepackage{tasks}
\usepackage{appendix}
\usepackage{algorithm}
\usepackage{algpseudocode}
\usepackage{makecell}
\usepackage{wrapfig}
\usepackage{xspace}
\usepackage{makecell}
\usepackage{array}
\usepackage{multirow}
\usepackage{arydshln}
\usetikzlibrary{calc}
\usetikzlibrary{decorations.pathreplacing,calligraphy}

\DeclareMathOperator{\supp}{supp}
\DeclareMathOperator{\tr}{tr}

\title{Constant-Overhead Magic State Distillation}
\author{Adam Wills\thanks{Center for Theoretical Physics, Massachusetts Institute of Technology, Cambridge, MA; Hon Hai (Foxconn) Research Institute, Taipei, Taiwan. Email: \texttt{a\_wills@mit.edu}.}\and Min-Hsiu Hsieh\thanks{Hon Hai (Foxconn) Research Institute, Taipei, Taiwan. Email: \texttt{min-hsiu.hsieh@foxconn.com}.}\and Hayata Yamasaki\thanks{Department of Physics, Graduate School of Science, The University of Tokyo, 7-3-1 Hongo, Bunkyo-ku, Tokyo, 113-0033, Japan. Email: \texttt{hayata.yamasaki@gmail.com}.}}
\date{}
\begin{document}

\maketitle

\begin{abstract}
Magic state distillation is a crucial yet resource-intensive process in fault-tolerant quantum computation. The protocol's overhead, defined as the number of input magic states required per output magic state with an error rate below $\epsilon$, typically grows as $\mathcal{O}(\log^\gamma(1/\epsilon))$. Achieving smaller overheads, i.e., smaller exponents $\gamma$, is highly desirable; however, all existing protocols require polylogarithmically growing overheads with some $\gamma > 0$, and identifying the smallest achievable exponent $\gamma$ for distilling magic states of qubits has remained challenging. To address this issue, we develop magic state distillation protocols for qubits with efficient, polynomial-time decoding that achieve an $\mathcal{O}(1)$ overhead, meaning the optimal exponent $\gamma = 0$; this improves over the previous best of $\gamma \approx 0.678$ due to Hastings and Haah. In our construction, we employ algebraic geometry codes to explicitly present asymptotically good quantum codes for $2^{10}$-dimensional qudits that support transversally implementable logical gates in the third level of the Clifford hierarchy. The use of asymptotically good codes with non-vanishing rate and relative distance leads to the constant overhead. These codes can be realised by representing each $2^{10}$-dimensional qudit as a set of $10$ qubits, using stabiliser operations on qubits. The $10$-qubit magic states distilled with these codes can be converted to and from conventional magic states for the controlled-controlled-$Z$ ($CCZ$) and $T$ gates on qubits with only a constant overhead loss, making it possible to achieve constant-overhead distillation of such standard magic states for qubits. These results resolve the fundamental open problem in quantum information theory concerning the construction of magic state distillation protocols with the optimal exponent.
\end{abstract}

\section{Introduction}

It is well understood that some method of protecting quantum information from the effects of noise and decoherence is necessary to achieve scalable quantum computing. This is the mandate of the field of quantum error correction. Equally challenging is achieving the ability to not only reliably store quantum information in the presence of noise, but compute with the information in a manner that is tolerant to faults; this is the aim of studying fault-tolerant quantum computing~\cite{gottesman2010,gottesman2016surviving}.
In the near-thirty years since the discovery of the threshold theorem~\cite{10.1145/258533.258579,doi:10.1137/S0097539799359385,shor1996fault,10.5555/2011665.2011666,10.1007/11786986_6,yamasaki2024time}, many schemes and techniques for quantum error correction and fault-tolerant quantum computing have been proposed, many featuring various advantages and disadvantages in certain areas. Some such techniques have started to be successfully experimentally realised.

One important notion in quantum error correction and fault-tolerant quantum computing is that of a transversal gate. If quantum information is protected by some error-correcting code, for example, a stabiliser code \cite{gottesman1997stabilizer}, one acts with a one qubit gate $U$ on all the physical qubits of the code, which has the effect of executing the gate $U$ on all the encoded qubits. This requires the quantum error-correcting code to have a particular structure; we say that ``it supports a transversal gate $U$''. The point here is that acting transversally is an inherently fault-tolerant way to compute on the encoded information because it prevents the spread of errors between qubits.

Unfortunately, the Eastin-Knill theorem \cite{eastin2009restrictions} tells us that no quantum code can support a set of transversal gates that enables universal quantum computation; some further technique is required. Moreover, in many natural quantum error-correcting codes, the transversal gates that are supported are Clifford gates, whose action (on stabiliser states) is known to be efficiently classically simulatable.

One particularly successful technique that can augment transversal gates to achieve universal quantum computation is known as magic state distillation, initially introduced by Bravyi and Kitaev in \cite{bravyi2005universal}. Let us briefly outline this. Suppose we have some diagonal $\tau$-qubit non-Clifford gate, $U$, and let us define the state
\begin{equation}
    \ket{M} = U\left(\ket{+}^{\otimes \tau}\right)
\end{equation}
where $\ket{+} = \frac{1}{\sqrt{2}}(\ket{0}+\ket{1})$ is the usual $+1$-eigenstate of the $X$-operator.
Common examples of $U$ include the $T$ gate, in which case $\tau=1$, the controlled-$S$ gate, $CS$, in which case $\tau=2$, or the controlled-controlled-$Z$ gate, $CCZ$, in which case $\tau=3$. All of these gates are non-Clifford, and all such gates enable universal quantum computation when combined with Clifford operations.
We write magic states for $T$ gates and $CCZ$ gates as $\ket{T}$ and $\ket{CCZ}$, respectively.

Notice that because $U$ is diagonal, the state $\ket{M}$ essentially lists the non-zero elements of $U$ in its amplitudes. This fact enables an interesting construction, known as gate teleportation. In this construction, one state $\ket{M}$ may be consumed to execute the gate $U$ on an unknown $\tau$-qubit quantum state $\ket{\psi}$, using Clifford operations only. This requires, however, $U$ to be a gate in the third level of the Clifford hierarchy: a concept introduced by Gottesman and Chuang \cite{gottesman1999demonstrating} that we will discuss more later. We mention, for now, that $T$, $CS$, and $CCZ$ are all in the third level.

Using gate teleportation, we see, therefore, that a supply of states $\ket{M}$ can promote the Clifford operations to give universal quantum computation. The issue that magic state distillation tackles, then, is that of producing high-quality ``magic states'' $\ket{M}$. What is meant by this is that in reality, we do not hold the pure state $\ket{M}$, but rather some mixed state $\rho$, which we hope to be very ``close'' to the magic state $\ket{M}$, in the sense that we want its error rate $p$ to be small, where
\begin{equation}
    \bra{M}\rho\ket{M} \geq 1-p.
\end{equation}
As the name suggests, the task accomplished by magic state distillation is to produce a smaller number of magic states of a higher quality (lower error rate) from a larger number of magic states of a lower quality (higher error rate).

Magic state distillation has been studied extensively since the original proposal by Bravyi and Kitaev \cite{bravyi2005universal} and is still under active investigation. One issue of particular interest is the overhead of magic state distillation --- essentially the ratio of input to output magic states to a protocol. More precisely, if one has a target error rate $\epsilon$, the overhead is the number of input magic states of constant error rate (below some constant threshold) that are required per output magic state of error rate $\epsilon$ to be produced. Denoted $C(\epsilon)$, in general the overhead may be expressed as
\begin{equation}
    C(\epsilon) = \mathcal{O}(\log^\gamma(1/\epsilon)).
\end{equation}
Magic state distillation is considered one of the most prohibitive elements in terms of resource cost of fault-tolerant quantum computation proposals, and so reducing the exponent $\gamma$ is a major issue if we wish to achieve realistic fault-tolerant quantum computation. For context, the initial value of $\gamma$ obtained by Bravyi and Kitaev was $\gamma = \log_3(15) \approx 2.46$.

The common assumption that is made when constructing magic state distillation protocols is to assume noiseless Clifford gates as if the protocol described is protected by some outer error-correcting code supporting transversal Clifford gates. This was the assumption initially introduced by Bravyi and Kitaev \cite{bravyi2005universal}. Under this assumption, several works presenting protocols with lower values of $\gamma$ have been published. In 2012, Bravyi and Haah presented the framework of triorthogonality \cite{bravyi2012magic}, a powerful framework for developing magic state distillation protocols for the $T$ gate, but later generalised to the $CS$ and $CCZ$ gates with the generalised triorthogonality framework of Haah and Hastings \cite{haah2018codes}. Triorthogonality starts with the definition of a triorthogonal matrix: a binary matrix in which any pair of rows has even overlap, and any triple of rows has even overlap. Bravyi and Haah discuss how these matrices may be used to define a ``qubit triorthogonal code'': a quantum code that supports a non-Clifford transversal gate (in fact the $T$ gate in their case). They show that such codes may be used to construct magic state distillation protocols for which
\begin{equation}
\label{eq:overhead_of_existing_protocol}
    \gamma = \frac{\log(n/k)}{\log d}
\end{equation}
and where $[[n,k,d]]$ denote the parameters of the quantum code\footnote{The notation $[[n,k,d]]$ is a canonical notation for a quantum code with $n$ physical qubits, $k$ logical qubits (we say the code has dimension $k$), and distance $d$.}. In particular, this is done with a ``multi-level'' magic state distillation protocol, in which a single magic state distillation subroutine is concatenated with itself, meaning that the output states of one level are used as inputs to the next level. In their paper, Bravyi and Haah achieve $\gamma$ arbitrarily close to $\log_2(3) \approx 1.58$, improving upon the best-known values at the time: $\gamma = \log_2(5) \approx 2.32$ for qubits due to Meier, Eastin, and Knill \cite{meier2012magic}, and $\gamma = 2$ for qudits due to Campbell, Anwar, and Browne \cite{campbell2012magic}.

Bravyi and Haah also conjectured in their paper that the exponent $\gamma$ could not go below 1 for any distillation protocol, and subsequently, a protocol with $\gamma \to 1$ was discovered by Jones \cite{jones2013multilevel}. Their conjecture was, however, falsified by Hastings and Haah in the breakthrough work \cite{hastings2018distillation}, achieving $\gamma \approx 0.678$, although these protocols are not considered practical due to their requirement for exorbitantly large numbers of input magic states. In particular, a protocol with $\gamma < 1$ is only discovered on $n=2^{58}$-qubit inputs in that work. Protocols with $\gamma> 0$ that can be made arbitrarily close to $0$ were discovered shortly after by Krishna and Tillich~\cite{krishna2019towards}, although these come with their own drawback. To achieve $\gamma \to 0$, these protocols require prime-dimensional qudits whose dimension diverges to infinity. In particular, for any fixed-dimensional qudit, these protocols yield some constant $\gamma > 0$. Moreover, the actual dimensions of the qudits required to achieve low values of $\gamma$ were quite large; achieving $\gamma < 1$ requires $41$-dimensional qudits, and improving on Hastings and Haah's value of $\gamma \approx 0.678$ requires $97$-dimensional qudits.
It is unknown how to use the protocol for prime-dimensional qudits  in~\cite{krishna2019towards} to distil magic states of qubits.

Despite their separate impracticalities, both~\cite{hastings2018distillation} and~\cite{krishna2019towards} represent significant theoretical breakthroughs. In particular, they both gave inspiration for the present work. As a key example, Krishna and Tillich in \cite{krishna2019towards} generalised the notion of binary triorthogonal matrices to a definition of triorthogonal matrices over the field $\mathbb{F}_p$, where $p$ is a prime, where again the definition is made in view of defining a quantum code (on $p$-dimensional qudits) that supports a non-Clifford transversal gate.

\subsection{Overview of Our Results}

It has been an important open question as to whether protocols may be designed with $\gamma$ made arbitrarily close to 0 for qubits. The main contributions of this paper are not only to answer this in the affirmative but also to determine that constant-overhead magic state distillation is possible, i.e. $\gamma = 0$, for qubits (see also Table~\ref{tab:gamma_values}).

\begin{theorem}\label{mainTheorem}
    Constant-overhead magic state distillation for qubits is possible. In particular, there exist constants $\Omega,$ $p_\mathrm{th}$ such that, for any target error rate $\epsilon$, there is a large enough $n$ such that a supply of $n$ $\ket{CCZ}$ magic states of error rate below $p_\mathrm{th}$ may be used to produce $k$ $\ket{CCZ}$ magic states of error rate below $\epsilon$, under the local stochastic error model, where $n/k \leq \Omega$.
\end{theorem}
\renewcommand{\arraystretch}{1.4}
\begin{table}[h]
    \centering
    \begin{tabular}{c|c||c|c}
         Work&Year&Overhead Exponent $\gamma$&Qudit Dimension $d$  \\\hline\hline
         Hastings and Haah \cite{hastings2018distillation}& 2017&$\gamma \approx 0.678$&$d=2$\\\hline
         Krishna and Tillich \cite{krishna2019towards}&2018&$\gamma \to 0$&$d \to \infty$\\\hline
         Present Work & 2024 & $\gamma = 0$&$d=2$
    \end{tabular}
    \caption{The overhead for magic state distillation is $\mathcal{O}(\log^\gamma(1/\epsilon))$, for a target error rate $\epsilon$, as discussed. This table presents the result of this work in comparison to the previous state of the art.}
    \label{tab:gamma_values}
\end{table}
\renewcommand{\arraystretch}{1}
Our main construction is geared around the distillation of the $\ket{CCZ}$ state, which enables the execution of the (non-Clifford) $CCZ$ gate via stabiliser operations only, using a gate teleportation protocol. In addition, one may also achieve constant-overhead magic state distillation in the same sense as Theorem \ref{mainTheorem} (possibly with different constants) for the $\ket{T}$ gate by using known protocols to convert $\ket{CCZ}$ states to $\ket{T}$ states~\cite{beverland2020lower,PhysRevA.87.022328,PhysRevA.87.042302,Gidney2019efficientmagicstate}.

At a high level, there are two ingredients that allow us to prove Theorem \ref{mainTheorem}. The first is a qubit code which is asymptotically good (meaning it has essentially optimal density of information storage and error-correcting capability: $k,d = \Theta(n)$) that also supports a non-Clifford transversal gate. The second is a triorthogonality framework for prime-power qudits (in fact, for $2^{10}$-dimensional qudits) which generalises the framework of Bravyi and Haah \cite{bravyi2012magic} for qubits and that of Krishna and Tillich \cite{krishna2019towards} for prime-dimensional qudits. We also show the translation between a $2^{10}$-dimensional qudit and $10$ qubits. 
The proof of Theorem \ref{mainTheorem} is essentially the work of the entire paper, but the conclusion comes in Section \ref{errorAnalysis} in the form of Theorem \ref{thm:error_rate}.

Because our construction of a quantum code may be of interest in other contexts, we will enunciate it as its own theorem here.
This theorem will be established via the results of Sections \ref{mainCodeSection} and \ref{mainConstruction}, in particular, Theorem~\ref{technicalCodeTheorem}.

\begin{theorem}\label{quantumCodeTheorem}
    We have an explicit construction of a family of $[[n,k,d]]$ quantum codes for $2^{10}$-dimensional qudits that is asymptotically good, so $k,d = \Theta(n)$, and supports a non-Clifford transversal gate in the third level of the Clifford hierarchy of $2^{10}$-dimensional qudits. This latter statement means that there is a single-qudit non-Clifford gate $U$ in the third level of the Clifford hierarchy of $2^{10}$-dimensional qudits such that acting with $U^{\otimes n}$ on the physical qudits of a single code block executes the encoded gate $\overline{U^{\otimes k}}$ on the logical qudits of the code block.
    Each $U$ is also in the third level of the Clifford hierarchy of qubits, implemented by a finite sequence of Clifford+$CCZ$ gates on qubits, and can be used for gate teleportation to implement a $CCZ$ gate on qubits.
    Moreover, this code has a polynomial-time decoder for correcting $Z$ errors, which are relevant for magic state distillation.
\end{theorem}

To discuss our whole construction in a little more depth, let us return to \cite{hastings2018distillation} and \cite{krishna2019towards}. In the former paper, the (binary) triorthogonal matrices, originally introduced in~\cite{bravyi2012magic}, are constructed using Reed-Muller codes. In the latter, the triorthogonal matrices over the field $\mathbb{F}_p$, where $p$ is a prime, are defined and constructed using Reed-Solomon codes. In both cases, it is the properties of these codes as polynomial evaluation codes that lead to the triorthogonality conditions being satisfied. The problem in the first case is that Reed-Muller codes do not have good parameters simultaneously on themselves and on their duals\footnote{The other problem with Reed-Muller codes is that, as a code family, their size grows very quickly, which leads to the long lengths of the codes achieving $\gamma < 1$.}, which leads to only $\gamma \approx 0.678$. On the other hand, Reed-Solomon codes do have good parameters on themselves and their duals, and so $\gamma \to 0$ is possible, although a natural Reed-Solomon code family is defined over fields whose size is growing. This, in turn, corresponds to the qudits of the triorthogonal codes, and thus the magic state distillation protocols, growing impractically large.

One naturally asks, therefore, whether there is a further class of classical code with good parameters on themselves and on their duals, which may be defined over a fixed-order field.
Our idea is that the algebraic geometry codes, dating back to the work of Goppa \cite{goppa1982algebraico}, fit this description.
The algebraic geometry codes have been considered in the context of quantum codes without arguing concrete applications to fault-tolerant quantum computing --- first in \cite{ashikhmin2001asymptotically}, but here we aim to use these codes for magic state distillation.
However, the problem to overcome with algebraic geometry codes is that they are naturally defined on relatively large, albeit fixed, fields\footnote{Recall that a field of order $r$, where $r$ is an integer, exists if and only if $r$ is a prime power, where prime powers are often denoted by the character $q$. There is one field of order $q$ for each prime power $q$, up to isomorphism, and it may be denoted $\mathbb{F}_q$.}. This is discussed in more depth in Section \ref{AGPrelims}. The problem of algebraic geometry codes naturally defining codes on larger fields may be easily overcome in classical coding theory by some ``binarisation'' process. For example, if one has a classical code family over $\mathbb{F}_q$, where $q=2^s$, then it is possible to express each $\mathbb{F}_q$-symbol as a list of $s$ bits. Performing this transformation on every $\mathbb{F}_q$-symbol in every codeword defines a binary code, and since $s$ is constant, the parameters of the resulting code family will be essentially the same as the code family over $\mathbb{F}_q$.
By contrast, the problem in the present case is that it is indeed possible to use algebraic geometry codes to construct codes over $\mathbb{F}_q$, for some fixed $q=2^s$ which satisfy (a generalisation of) the triorthogonality properties of Bravyi and Haah \cite{bravyi2012magic}, but it is not clear whether it is possible to maintain the triorthogonality under a binarisation process to obtain a standard binary triorthogonal code.

We nevertheless develop another approach to circumvent this problem. Let us first ask what is meant by a $q$-dimensional qudit, where $q=2^s$ is a power of two. One finds that there is a standard approach outlined in Chapter 8 of \cite{gottesman2016surviving}.
A quantum state of a single qudit of dimension $q$ is in the Hilbert space $\mathbb{C}^q$, where the standard choice of computational basis is labelled by field elements: $\{\ket{\eta}:\eta \in \mathbb{F}_q\}$. On this, there is a standard choice of Pauli group, Clifford group, and Clifford hierarchy. It may then be shown that all these sets of gates are essentially ``the same'' as their counterparts for a collection of $s$ qubits.

The way this is done is presented in detail in Section \ref{primePowerQuditsPrelims}, but for now, we say the following. There is a natural correspondence (in fact there are many, but we fix one) between elements of $\mathbb{F}_{2^s}$ and bit strings of length $s$ (i.e. elements of $\mathbb{F}_2^s$),
\begin{equation}
    \mathbb{F}_{2^s} \ni \beta \longleftrightarrow \mathbf{b} \in \mathbb{F}_2^s.
\end{equation}
Accordingly, this gives a correspondence between computational basis states of $\mathbb{C}^q$ and $(\mathbb{C}^2)^{\otimes s}$:
\begin{equation}
    \ket{\beta} \longleftrightarrow \ket{\mathbf{b}}.
\end{equation}
Then, a quantum state of one $q$-dimensional qudit, which is a linear combination of $\{\ket{\beta}:\beta \in \mathbb{F}_q\}$, corresponds to a quantum state of $s$ qubits, as the same linear combination of the corresponding $\{\ket{\mathbf{b}}: \mathbf{b} \in \mathbb{F}_2^s\}$. There is a similarly natural isomorphism between the sets of unitaries acting on the two spaces. If one has a unitary acting on $\mathbb{C}^q$, denoted as its matrix in the computational basis of $\mathbb{C}^q$, $\{\ket{\beta}:\beta \in \mathbb{F}_q\}$, the corresponding unitary acting on $(\mathbb{C}^2)^{\otimes s}$ has exactly the same matrix, with respect to the corresponding computational basis on that space, $\{\ket{\mathbf{b}}: \mathbf{b} \in \mathbb{F}_2^s\}$. It may then be shown that this isomorphism specialises to an isomorphism of the Pauli group, the Clifford group, and of the group of diagonal non-Clifford gates in the third level of the Clifford hierarchy.

Knowing this, we want to define a generalisation of triorthogonal matrices (and from there, triorthogonal codes) for qudits of dimension $q=2^s$ and then construct such a triorthogonal matrix using algebraic geometry codes over $\mathbb{F}_q$. Our construction is built on $q$-dimensional qudits, where $q=2^{10} = 1024$, and the code on $q$-dimensional qudits naturally corresponds to a code on ``sets of $10$ qubits'' (as is discussed in Section \ref{primePowerQuditsPrelims}). 
We ultimately want magic state distillation protocols on qubits, which is the reason why we have chosen the prime power $q$ to be a power of two.
The reason for the particular choice of $q=2^{10}$ will become clear in Section \ref{mainConstruction}.
The definition of a triorthogonal matrix in this paper is as follows (note that the following definition of triorthogonality naturally generalises to qudits of dimension $q=2^s$ for any $s$ while we here present the case of $q=2^{10}=1024$, i.e., $s=10$, relevant to our construction).

\begin{definition}[Triorthogonal Matrix]\label{firstTriorthogDefn}
    For us, a matrix $G \in \mathbb{F}_{q}^{m \times n}$, where $q=2^{10}=1024$, with rows $(g^a)_{a=1}^m$, is called triorthogonal if for all $a,b,c \in \{1, ..., m\}$,
    \begin{equation}
        \sum_{i=1}^n(g^a_i)^4(g^b_i)^2(g^c_i) = \begin{cases}
            1 &\text{ if } 1 \leq a=b=c \leq k\\
            0 &\text{ otherwise}
        \end{cases}
    \end{equation}
    and
    \begin{equation}
        \sum_{i=1}^n \sigma_ig_i^ag_i^b = \begin{cases}
            \tau_a &\text{ if } 1 \leq a=b \leq k\\
            0 &\text{ otherwise}
        \end{cases}
    \end{equation}
    for some integer $k \in \{1, ..., m\}$ and for some $\sigma_i, \tau_a \in \mathbb{F}_{1024}$, where $\sigma_i, \tau_a \neq 0$. In these expressions, all arithmetic takes place over $\mathbb{F}_{1024}$.
\end{definition}

The conditions in Definition \ref{firstTriorthogDefn} are carefully chosen so that the quantum code that we define out of the triorthogonal matrix supports a non-Clifford transversal gate (the resulting code is called a quantum triorthogonal code). This gate will be
\begin{equation}
    U = \sum_{\beta \in \mathbb{F}_{1024}}\exp\left[i\pi\tr(\beta^7)\right]\ket{\beta}\bra{\beta}
\end{equation}
where $\tr:\mathbb{F}_{1024} \to \mathbb{F}_2$ is the trace map, which is defined in Section \ref{finiteFieldPrelims}. This will enable the distillation of the one-qudit magic state
\begin{equation}
    \ket{M} = U\ket{+_q},
\end{equation}
where $\ket{+_q} = \frac{1}{\sqrt{q}}\sum_{\beta \in \mathbb{F}_q}\ket{\beta}$,
via a distillation subroutine, essentially the same as that employed in the conventional triorthogonality framework for qubits~ \cite{bravyi2012magic}. This all may be equally well described as the distillation of the $10$-qubit state $\ket{U_{10}}$, which is the qubit state corresponding to $\ket{M}$ according to the correspondence described above.

Via gate teleportation, this allows the execution of the $10$-qubit non-Clifford gate corresponding to $U$. However, we note that this gate is highly non-standard, and it is much more desirable to have the distillation in terms of some more standard non-Clifford gate. We show, in fact, that $\ket{U_{10}}$ may be converted from and into $\ket{CCZ}$ up to only some constant-factor loss of the overhead.

The quantum code that we construct from our triorthogonal matrix will be asymptotically good, i.e. it has parameters $[[n,k,d]]$ where $k,d = \Theta(n)$ as $n \to \infty$, and the argument of~\cite{bravyi2012magic} naturally applies so that the overhead exponent of $\ket{U_{10}}$ distillation if we concatenated the distillation subroutine multiple times, would be
\begin{equation}
\label{eq:overhead_of_existing_protocol2}
    \gamma = \frac{\log(n/k)}{\log d}.
\end{equation}
Since we have $k,d = \Theta(n)$, $\gamma$ in such a concatenated protocol would be made arbitrarily close to $0$. However, we note that this could be achieved with any constant rate $k/n$ and any $d$ which diverges to infinity as $n \to \infty$. In our case, it turns out that because the parameters are asymptotically good, a stronger claim is possible. Using just one round of the magic state distillation subroutine, i.e., without concatenating it with itself, distillation with $\gamma = 0$ is possible. Essentially, this is true because the error suppression is strong enough (following from $d = \Theta(n)$) that one round of magic state distillation, with a large enough $n$, is enough to make the target error rate as small as desired.
To achieve $\gamma=0$, it is also crucial that our protocol, using a polynomial-time decoder for algebraic geometry codes, performs error correction without post-selection to keep the overhead small\footnote{We note that the idea of using error correction, rather than post-selection, in magic state distillation has appeared before~\cite{krishna2018magic,Hu2022designingquantum}.
This idea was originally due to David Poulin~\cite{KrishnaPersonal}.}.
Because $k = \Theta(n)$, this means that the number of outputted $\ket{U_{10}}$ states grows in proportion to the number of input states; i.e., the overhead is constant $\mathcal{O}(1)$ as $\epsilon\to 0$. When combined with the conversion to the magic states $\ket{CCZ}$ (with only a constant-factor loss), this establishes Theorem \ref{mainTheorem}. This argument on the error analysis will be made precise in Section \ref{errorAnalysis}.

\paragraph*{Note on related works}
Independent works~\cite{golowich2024asymptoticallygoodquantumcodes,nguyen2024goodbinaryquantumcodes} from ours have also presented asymptotically good quantum codes that support transversal non-Clifford gates.
Whereas our protocol for magic state distillation uses $2^{10}$-dimensional qudits for our codes, the codes presented in these works support transversal $CCZ$ gates on qubits, addressing one of the open questions in the first version of this paper. These other works use some similar constructions to ours, and some other techniques, which will likely prove to be complementary to those in the present paper.

However, our contribution is to prove the achievability of constant-overhead magic state distillation in addition to the construction of asymptotically good codes.
The use of asymptotically good codes in the existing protocols for magic state distillation leads to the polylogarithmic overhead $\mathcal{O}(\log^\gamma(1/\epsilon))$ with $\gamma>0$ arbitrarily small, according to the existing analysis summarized in Equations~\eqref{eq:overhead_of_existing_protocol} and~\eqref{eq:overhead_of_existing_protocol2}.
Without presenting different protocols and analyses, the works~\cite{golowich2024asymptoticallygoodquantumcodes,nguyen2024goodbinaryquantumcodes} only lead to $\gamma \to 0$, i.e., magic state distillation with polylogarithmic overhead of arbitrarily low degree, rather than constant-overhead magic state distillation.
By contrast, our error analysis in Section~\ref{errorAnalysis} shows that using the asymptotically good codes in the single-round protocol presented in this work, which uses error correction without post-selection, it is possible to achieve the constant overhead $\gamma=0$, where we can also use the asymptotically good codes in Refs.~\cite{golowich2024asymptoticallygoodquantumcodes,nguyen2024goodbinaryquantumcodes} as well as our asymptotically good codes.

\subsection{Discussion and Future Directions}
\label{sec:discussion}

In this work, we have provided protocols for the distillation of $\ket{CCZ}$ magic states on qubits with constant overhead under the local stochastic error model. This answers in the affirmative an important open question in fault-tolerant quantum computing as to whether the magic state distillation overhead exponent $\gamma$ could be made arbitrarily close to zero for qubits, and in fact, we show a stronger statement --- that $\gamma = 0$ is possible.

We believe that, in addition, this work should open many interesting new research directions, both on the front of more practical issues and also in theory. Most pressingly, it is necessary to examine the constant factors and the finite length performance of the protocols. In the present work, we have restricted our focus to showing that constant-overhead magic state distillation is possible for qubits, and have made no attempt to optimise the construction up to constant factors, or analyse the finite length performance. This is an issue we will tackle in follow-up work to present and optimise exact values of the constants $\Omega$ and $p_\mathrm{th}$.
In case the protocols in their current form appear to be impractical, or even if they are quite practical, a natural follow-up question is then to ask whether other protocols from ours also exist at more reasonable lengths and/or with improved values of $\Omega$ and $p_\mathrm{th}$.

It is worth noting, however, that our protocol makes it possible to perform constant-overhead magic state distillation with the same threshold as the conventional $15$-to-$1$ protocol in~\cite{bravyi2005universal} (and also the same thresholds as any other small-scale protocols, e.g., those in~\cite {meier2012magic,PRXQuantum.3.030319}); to achieve this, we first suppress the error rate of initial magic states below the threshold $p_\mathrm{th}$ of our protocol using these small-scale protocols with a constant number of concatenations and then switch to our protocol to further suppress the error rate of the magic states arbitrarily within the constant overhead.
Thus, the issue involves the trade-off between the overhead and the threshold in the practical regime of physical and target logical error rates.
Such a trade-off needs to be estimated by numerical simulation since the analytical bounds of error rates and thresholds may be loose.
An improvement might come from algebraic geometry codes based on our development, or something entirely different.

As discussed, the protocols we construct naturally distil a $10$-qubit magic state for a $10$-qubit non-Clifford gate, but this magic state may be converted to and from $\ket{CCZ}$ by stabiliser operations; thus, our protocols may distil the $\ket{CCZ}$ state.
Distillation of $\ket{T}$ states requires further conversions~\cite{beverland2020lower,PhysRevA.87.022328,PhysRevA.87.042302,Gidney2019efficientmagicstate}.
Algebraic geometry codes provide the desired framework to construct asymptotically good codes from matrices satisfying polynomial equations like Definition~\ref{firstTriorthogDefn}.
However, in this work, we did not optimise constant factors appearing in these codes and protocols for constant-overhead magic state distillation.
The question one might ask, therefore, is whether one may realise practical constant-overhead magic state distillation, which may require improving the codes and protocols with more optimised constant factors.
\\

\noindent\textbf{Open Question 1:} Can we construct asymptotically good quantum codes supporting transversally implementable non-Clifford gates with more optimised constant factors? In a practical finite-length regime, can the protocols for constant-overhead magic state distillation be advantageous over the conventional polylogarithmic-overhead protocols?
\\

On even more speculative matters, one can ask about magic state distillation protocols using quantum LDPC codes, including transversal gates for such codes, where ``LDPC'' stands for low-density parity-check: a very strong property for a quantum code to have to hamper the spread of errors during syndrome measurement. It was not that long ago that asymptotically good quantum LDPC codes were constructed, first in \cite{panteleev2022asymptotically}, with follow-up works \cite{leverrier2022quantum,dinur2023good}, thus bringing to the close a long search. The study of protocols for fault-tolerant quantum computing with such codes is still in its infancy but promises to be a burgeoning area of research over the next few years. The following question promises to be as difficult as it is interesting.
\\

\noindent\textbf{Open Question 2:} What are the optimal parameters of a quantum LDPC code supporting a transversally implementable non-Clifford gate? Does an asymptotically good quantum LDPC code supporting a transversally implementable non-Clifford gate exist?
\\

As a very theoretical problem, we recall the conversion between the single-qudit magic state $\ket{M}$ (or equivalently the $10$-qubit magic state $\ket{U_{10}}$) and the $\ket{CCZ}$ states of qubits. To perform this conversion, we will go on to describe how the $10$-qubit non-Clifford gate in question may be decomposed into a sequence of $Z$, $CZ$, and $CCZ$ gates. Now, once one has chosen an isomorphism between $2^{10}$-dimensional qudits and $10$ qubits, this decomposition is unique.
However, it does depend on the isomorphism chosen in the first place. In particular, different choices of self-dual bases used for the isomorphism may lead to different gate decompositions, potentially resulting in a smaller number of $CCZ$ gates in this decomposition. A smaller number is desirable, as it leads to a more efficient conversion between $\ket{U_{10}}$ and $\ket{CCZ}$. The following question is, therefore, well-motivated and could be an interesting algebraic problem.
\\

\noindent\textbf{Open Question 3:} How can we analytically determine the isomorphism between $2^{10}$-dimensional qudits and $10$ qubits (specified by a self-dual basis of $\mathbb{F}_{1024}$ over $\mathbb{F}_2$) which leads to the decomposition of our $10$-qubit non-Clifford gate containing the smallest number of $CCZ$ gates? In the case of $\mathbb{F}_{q}$ for $q=2^{s}$ with fixed $s=10$, we may be able to try an exhaustive search of all the self-dual bases, but can we efficiently identify the optimal solution for general $s$?
\\

Lastly, from the perspective of quantum information theory, magic state distillation can be viewed as a type of distillation task to convert noisy states into pure states.
In special cases, it may be possible to identify the optimal rate of distillation; for example, the optimal rate of entanglement distillation under one-way local operations and classical communication (LOCC) can be fully characterised for any bipartite mixed states~\cite{doi:10.1098/rspa.2004.1372}, and the optimal rate of entanglement distillation under two-way local operations and classical communication (LOCC) can be characterised for bipartite pure states~\cite{PhysRevA.53.2046}.
Our results prove that magic state distillation is possible at a nonzero asymptotic rate as in entanglement distillation, and this opens a quest to seek the asymptotically optimal overhead rates and protocols in the task of magic state distillation.
This is a fundamental question relevant to the quantum resource theory of magic~\cite{Veitch_2012,Veitch_2014,PhysRevLett.118.090501}.
\\

\noindent\textbf{Open Question 4:} What are the optimal overhead and the corresponding asymptotically optimal family of protocols of magic state distillation?

\subsection{Outline of the Paper}\label{outline}

The paper is structured as follows. Section \ref{prelimSection} presents the necessary preliminary material. In particular, Section \ref{finiteFieldPrelims} presents the material on finite fields, Section \ref{primePowerQuditsPrelims} describes prime-power-dimensional qudits, Section \ref{AGPrelims} presents the algebraic geometry codes, and Section \ref{sec:task} defines the task of magic state distillation and the error model we use.

Section \ref{mainCodeSection} introduces the framework for constructing $2^s$-dimensional-qudit triorthogonal codes from triorthogonal matrices over finite fields of order $2^s$.
Using this framework, our magic state distillation protocol is presented in Section \ref{MSDProtocolMainSection}.
The quantum code in the framework of Section \ref{mainCodeSection} is explicitly constructed from the triorthogonal matrix that is explicitly constructed in Section \ref{mainConstruction}.
The error analysis for the constant-overhead magic state distillation takes place in Section \ref{errorAnalysis}.

Section \ref{mainCodeSection} contains Section \ref{nonCliffordGate} which analyses the non-Clifford gate that the quantum code will support transversally. Section \ref{quantumCodefromMatrix} will then construct the quantum code and prove results about its properties.

Section \ref{MSDProtocolMainSection} contains Section \ref{MSDSummarySection} which summarises the magic state distillation protocols, Section \ref{twirling} which gives more information on the twirling portion of the protocol, Section \ref{MSDfromCode} which gives more information on the bulk of the protocol (the distillation of the $10$-qubit magic state) and lastly Section \ref{CCZConversion}, which describes the conversion between the $10$-qubit magic state and the $\ket{CCZ}$ state.

Section \ref{mainConstruction} contains Sections \ref{AGCodeConstruction} and \ref{puncturingSection} which provide the majority of the construction of the triorthogonal matrix, whereas Section \ref{concrete} makes the triorthogonal matrix construction concrete by assigning values to each variable of the construction, thus demonstrating that all hypotheses of the construction may be simultaneously satisfied, and that the parameters of the associated quantum code are asymptotically good.

As mentioned, Section \ref{errorAnalysis} provides the error analysis of the magic state distillation protocol, and proves that under our error model (the local stochastic error model), the protocol has a constant error threshold and achieves a constant overhead. 

\section{Preliminaries}\label{prelimSection}

\subsection{Finite Fields}\label{finiteFieldPrelims}

Finite fields play a significant role in this paper and so it will be important to have a full account of them. First, it is well known that if $r$ is a positive integer, there is a field of order $r$ if and only if $r$ is a prime power. Moreover, if $r$ is a prime power, there is a unique field of order $r$ up to isomorphism. Prime powers are often denoted by the letter $q$ and the field of order $q$ is denoted $\mathbb{F}_q$.

When $q$ is a prime (let's say $q=p$), the field $\mathbb{F}_p$ is most easily understood. As a set, $\mathbb{F}_p$ may be identified with the set of integers $\{0, 1, ..., p-1\}$, and arithmetic in this field then corresponds simply to addition and multiplication modulo $p$. Fields $\mathbb{F}_q$ with $q=p^s$ may be constructed from $\mathbb{F}_p$ as we will now describe. One considers any irreducible polynomial over $\mathbb{F}_p$ of degree $s$ and considers a root of this polynomial, say $\alpha$. Then, $\mathbb{F}_q$ is exactly the field generated by adding the element $\alpha$ to $\mathbb{F}_p$. It is worth emphasising that then $\mathbb{F}_q$ contains a copy of $\mathbb{F}_p$.

The case of importance to us is the finite field of $2^{10}=1024$ elements, $\mathbb{F}_{1024}$. We will often write $\mathbb{F}_q$ for this field, where $q=2^s$ and $s=10$. To construct this, $1+x^3+x^{10}$ is a convenient choice of irreducible polynomial over $\mathbb{F}_2$. Then, one may write
\begin{equation}\label{fieldBasis}
    \mathbb{F}_{q} = \left\{\sum_{i=0}^{s-1}a_i\alpha^i:a_i \in \mathbb{F}_2\right\}
\end{equation}
as a set. The arithmetic in $\mathbb{F}_{q}$ is then completely determined by modulo 2 addition and the expression $1+\alpha^3 + \alpha^{10} = 0$, or equivalently, $\alpha^{10} = 1+\alpha^3$. As a concrete example, one can show that $(1+\alpha^6+\alpha^8)(\alpha+\alpha^5) = \alpha^3+\alpha^4+\alpha^5+\alpha^6+\alpha^7+\alpha^9$. Also, the modulo 2 addition quickly gives us the following useful expression:
\begin{equation}\label{char2Binomial}
    (\beta_1+\beta_2)^2 = \beta_1^2+\beta_2^2 \text{ for all } \beta_1, \beta_2 \in \mathbb{F}_{q}.
\end{equation}

One thing that the expression \eqref{fieldBasis} makes clear is that the field $\mathbb{F}_{q}$ may be viewed as a vector space over $\mathbb{F}_2$, simply by addition in the larger field $\mathbb{F}_{q}$ and scalar multiplication by elements in $\mathbb{F}_2$. In fact, $\mathbb{F}_{q}$ is an $s$-dimensional vector space over $\mathbb{F}_2$, and the set $\{1, \alpha, \alpha^2, \ldots,\alpha^{s-1}\}$ is acting as a basis in this context. Many other choices of basis are possible, however, and we will go on to discuss those that are particularly important to us soon.

One of the central objects in this work is the trace map of finite fields, denoted $\tr$. While $\tr$ can be defined between any two finite fields where one contains the other, we will only need to consider $\tr:\mathbb{F}_{q} \to \mathbb{F}_2$. Any reference to $\tr$ or the ``trace map'' throughout the paper will refer to this trace. This may be defined as
\begin{equation}
    \tr:\begin{cases}
        \mathbb{F}_{q} &\to \mathbb{F}_2\\
        \gamma &\mapsto \sum_{i=0}^{s-1}\gamma^{2^i}
    \end{cases}
    \;\;\text{where arithmetic takes place in }\mathbb{F}_q.
\end{equation}
From this expression, it is not immediately clear why $\tr(\gamma)$ should necessarily be in $\mathbb{F}_2$ for all $\gamma \in \mathbb{F}_{q}$, since it is expressed as merely some arithmetic in $\mathbb{F}_{q}$. It can be shown, however, that this is indeed the case. While $\tr$ satisfies many nice properties, we will only need the following:
\begin{proposition}\label{traceProps}
    The trace map satisfies:
    \begin{enumerate}
        \item $\tr:\mathbb{F}_{q} \to \mathbb{F}_2$ is an $\mathbb{F}_2$-linear map;
        \item $\tr(\gamma) = 0$ for half the elements $\gamma \in \mathbb{F}_{q}$ and $\tr(\gamma) = 1$ for the other half;
        \item $\tr(\gamma^2) = \tr(\gamma)$ for all $\gamma \in \mathbb{F}_{q}$.
    \end{enumerate}
\end{proposition}
Explicitly, the first point states that
\begin{align}
    \tr(a\gamma) &= a\tr(\gamma) \text{ for all } a\in \mathbb{F}_2 \text{ and } \gamma \in \mathbb{F}_{q}\\
    \tr(\gamma_1 + \gamma_2) &=\tr(\gamma_1) + \tr(\gamma_2) \text{ for all } \gamma_1, \gamma_2 \in \mathbb{F}_{q}.
\end{align}
It is worth emphasising that the arithmetic within the argument of $\tr(\cdot)$ is taken over $\mathbb{F}_{q}$, whereas the arithmetic outside of the trace (on the right-hand sides) is taking place in $\mathbb{F}_2$. We also mention that the second point of Proposition \ref{traceProps} might be written more formally as ``$\ker(\tr)$ is an $(s-1)$-dimensional vector space over $\mathbb{F}_2$''.

Combining our above discussions of $\mathbb{F}_{q}$ viewed as a vector space, as well as that of the trace map, one may consider the following.
\begin{definition}[Self-Dual Basis]
    Consider a basis $(\alpha_i)_{i=0}^{s-1}$ for $\mathbb{F}_{q}$ over $\mathbb{F}_2$. The basis is called self-dual if
    \begin{equation}
        \tr(\alpha_i\alpha_j) = \delta_{ij} \text{ for all } i,j = 0, ..., s-1.
    \end{equation}
\end{definition}
It is known that self-dual bases exist for $\mathbb{F}_{p^s}$ over $\mathbb{F}_p$ if and only if $p=2$ or both $p$ and $s$ are odd \cite{seroussi1980factorization}, and so in particular they exist for $\mathbb{F}_{1024}$ over $\mathbb{F}_2$. There are, in fact, 6684672 distinct examples of self-dual bases for $\mathbb{F}_{1024}$ over $\mathbb{F}_2$ (see Theorem 3 of \cite{jungnickel1990number}). One such example is given by the elements
\begin{align}
\alpha_0 &= 1+\alpha^2+\alpha^4+\alpha^5+\alpha^7+\alpha^8\label{sdbFirst},\\
\alpha_1 &= \alpha^3+\alpha^6+\alpha^7+\alpha^8+\alpha^9,\\
\alpha_2 &= \alpha+\alpha^2+\alpha^5+\alpha^7 + \alpha^8+\alpha^9,\\
\alpha_3 &= 1+\alpha+\alpha^2+\alpha^3+\alpha^4+\alpha^6+\alpha^7+\alpha^8+\alpha^9,\\
\alpha_4 &= 1+\alpha+\alpha^4+\alpha^5+\alpha^7+\alpha^9,\\
\alpha_5 &= \alpha + \alpha^2+\alpha^3 + \alpha^7,\\
\alpha_6 &= \alpha^2+\alpha^6+\alpha^7,\\
\alpha_7 &= \alpha^2 + \alpha^5 + \alpha^7,\\
\alpha_8 &= 1+\alpha^3 + \alpha^7,\\
\alpha_9 &= 1 + \alpha^4 + \alpha^6 + \alpha^7,\label{sdbLast}
\end{align}
where, again, $\alpha$ is an element in $\mathbb{F}_{1024}$ satisfying $\alpha^{10}=1+\alpha^3$.

\subsection{Prime-Power Qudits}\label{primePowerQuditsPrelims}
Whereas qubits and qudits of prime dimension are quite well-studied, qudits whose dimension is a prime power are less commonly discussed. We describe such qudits now, relying heavily on Gottesman's book \cite{gottesman2016surviving}, in particular Chapter 8, but adding notation and details where we wish to spell things out most explicitly. Ultimately, we will find that a $2^{10}$-dimensional qudit can be thought of as just $10$ qubits.

Let $q=2^s$. While we think of $s=10$, the following treatment in fact holds for any $s$. The state space of a single qudit of dimension $q$ is, unsurprisingly, $\mathbb{C}^q$. The standard choice of notation (which is very convenient for us) for the computational basis of $\mathbb{C}^q$ is $\left\{\ket{\gamma}:\gamma \in \mathbb{F}_{q}\right\}$, i.e. the $q$ computational basis states of $\mathbb{C}^q$ are labelled by the elements of $\mathbb{F}_q$. If we start by considering the set of unitaries acting on $\mathbb{C}^q$, $U(\mathbb{C}^q)$, the most important subgroup for us is the Pauli group (acting on one qudit of dimension $q$), denoted $\mathcal{P}_{1,q}$, whose elements are
\begin{equation}
    \mathcal{P}_{1,q} = \{i^aX^\beta Z^\gamma:a \in \{0, 1, 2, 3\} \text{ and } \beta, \gamma \in \mathbb{F}_q\}.
\end{equation}
The elementary Pauli operators $X^\beta$ and $Z^\gamma$ act on computational basis states as
\begin{align}
    X^\beta\ket{\eta} &= \ket{\eta+\beta}\\
    Z^\gamma\ket{\eta} &= (-1)^{\tr(\gamma\eta)}\ket{\eta}
\end{align}
where $\eta + \beta$ refers to addition in $\mathbb{F}_q$, and $\tr$ is the trace map introduced in the previous section. With these definitions, the fundamental commutation relations become
\begin{equation}
    X^\beta Z^\gamma = (-1)^{\tr(\beta\gamma)}Z^\gamma X^\beta.
\end{equation}
The Clifford group for the qudit of dimension $q$ is
\begin{equation}
    \mathcal{C}_{1,q} = \{U \in U(\mathbb{C}^q): UPU^\dagger \in \mathcal{P}_{1,q} \text{ for all } P \in \mathcal{P}_{1,q}\}.
\end{equation}

With this defined, let us begin to explain the translation between $q$-dimensional qudits and $s$ qubits. First, fix any self-dual basis $(\alpha_i)_{i=0}^{s-1}$ for $\mathbb{F}_{q}$ over $\mathbb{F}_2$.
We write the computational basis states of $s$ qubits as
\begin{equation}
    \left\{\ket{\mathbf{b}} = \bigotimes_{i=0}^{s-1}\ket{b_i}:\mathbf{b} \in \mathbb{F}_2^s\right\} \subseteq (\mathbb{C}^2)^{\otimes s},
\end{equation}
where $\mathbf{b} = (b_0, b_1, \ldots, b_{s-1})$.
We then introduce the following correspondence between computational basis states of $\mathbb{C}^q$ and those of $(\mathbb{C}^2)^{\otimes s}$:
\begin{equation}
    \Ket{\beta = \sum_{i=0}^{s-1}b_i\alpha_i} \leftrightarrow \bigotimes_{i=0}^{s-1}\ket{b_i}.
\end{equation}
In words, for any computational basis state $\ket{\beta}$, $\beta$ may be (uniquely) expressed in terms of the self-dual basis as $\beta = \sum_{i=0}^{s-1}b_i\alpha_i$, where $b_i \in \mathbb{F}_2$, and then the computational basis state $\ket{\beta}$ of $\mathbb{C}^q$ corresponds to the computational basis state $\bigotimes_{i=0}^{s-1}\ket{b_i}$ of $(\mathbb{C}^2)^{\otimes s}$. This correspondence may then be extended linearly to an isomorphism of vector spaces $\psi: \mathbb{C}^{q} \to (\mathbb{C}^2)^{\otimes s}$. 

Next, let us consider the Pauli group for $s$ qubits. This is defined as
\begin{equation}
    \mathcal{P}_{s,2} = \{i^aX^{\mathbf{b}}Z^{\mathbf{c}}:a \in \{0,1,2,3\} \text{ and }\mathbf{b},\mathbf{c} \in \mathbb{F}_2^s\}.
\end{equation}
In this expression, $X$ and $Z$ are the standard elementary Pauli operators for qubits, and we have employed the standard notation
\begin{equation}
    X^{\mathbf{b}} = \bigotimes_{i=0}^{s-1}X^{b_i}
\end{equation}
for a tensor product of a single-qubit operator acting on multiple qubits. This allows us to set up a very natural isomorphism between $\mathcal{P}_{1,q}$ and $\mathcal{P}_{s,2}$, where again we note that $q=2^s$. Indeed, we define
\begin{equation}
    \theta_P:\begin{cases}
        \mathcal{P}_{1,q} &\to \mathcal{P}_{s,2}\\
        i^aX^\beta Z^\gamma &\mapsto i^aX^{\mathbf{b}}Z^{\mathbf{c}} \text{, where } \beta = \sum_{i=0}^{s-1}b_i\alpha_i \text{ and }\gamma = \sum_{i=0}^{s-1}c_i\alpha_i.
    \end{cases}
\end{equation}
In words, we expand out the labels $\beta$ and $\gamma$ of the elementary Paulis $X$ and $Z$ in terms of the fixed self-dual basis, and their coefficients give the corresponding powers of $X$ and $Z$ on each qubit.

It is worth emphasising that the isomorphism of vector spaces $\psi : \mathbb{C}^{q} \to (\mathbb{C}^2)^{\otimes s}$, and the map $\theta_P: \mathcal{P}_{1,q} \to \mathcal{P}_{s,2}$, which will turn out to be an isomorphism of groups, both depend on the choice of self-dual basis for $\mathbb{F}_{q}$ over $\mathbb{F}_2$, of which there are multiple, but it is important that the self-dual basis used when defining $\psi$ and $\theta_P$ is the same. Given the multiple choices for this self-dual basis, however, it is true that there are multiple ways to decompose a $q$-dimensional qudit into $s$ qubits.

The fact that $\theta_P$ forms an isomorphism (of groups) is easily seen --- indeed it is a homomorphism from
\begin{align}
    \theta_P(i^{a_1}X^{\beta_1}Z^{\gamma_1}i^{a_2}X^{\beta_2}Z^{\gamma_2}) &= \theta_P(i^{a_1+a_2}(-1)^{\tr(\beta_2\gamma_1)}X^{\beta_1+\beta_2}Z^{\gamma_1+\gamma_2})\\
    &= i^{a_1+a_2}(-1)^{\tr(\beta_2\gamma_1)}X^{\mathbf{b}_1+\mathbf{b}_2}Z^{\mathbf{c}_1+\mathbf{c}_2}\\
    &= i^{a_1}X^{\mathbf{b}_1}Z^{\mathbf{c}_1}i^{a_2}X^{\mathbf{b}_2}Z^{\mathbf{c}_2}\\
    &= \theta_P(i^{a_1}X^{\beta_1}Z^{\gamma_1})\theta_P(i^{a_2}X^{\beta_2}Z^{\gamma_2}),
\end{align}
where the calculation uses the obvious notation, and going into the third line we have used $\mathbf{b}_2\cdot \mathbf{c}_1 = \sum_{i=0}^{s-1}(\mathbf{b}_2)_i(\mathbf{c}_1)_i = \sum_{i,j=0}^{s-1}(\mathbf{b}_2)_i(\mathbf{c}_1)_j\delta_{ij} = \tr\left(\sum_{i=0}^{s-1}(\mathbf{b}_2)_i\alpha_i\sum_{j=0}^{s-1}(\mathbf{c}_1)_j\alpha_j\right) = \tr(\beta\gamma)$ which implies $X^{\mathbf{b}_2}Z^{\mathbf{c}_1} = (-1)^{\tr(\beta\gamma)}Z^{\mathbf{c}_1}X^{\mathbf{b}_2}$. The fact that $\theta_P$ forms a bijection follows quickly from the fact that, via decomposition in the self-dual basis, elements of $\mathbb{F}_{q}$ are in one-to-one correspondence with bit strings in $\mathbb{F}_2^s$.

Using $\psi$, let us now construct a more general isomorphism between unitaries acting on $\mathbb{C}^q$ and those acting on $(\mathbb{C}^2)^{\otimes s}$. This is defined as
\begin{equation}
    \theta:\begin{cases}
        U\left(\mathbb{C}^{q}\right) &\to U\left((\mathbb{C}^2)^{\otimes s}\right)\\
        U &\mapsto \theta(U)
    \end{cases}
\end{equation}
where $\theta(U)$ acts on $v \in (\mathbb{C}^2)^{\otimes s}$ as
\begin{equation}
    \theta(U)(v) = \psi U \psi^{-1}(v).
\end{equation}
Since $\psi$ maps the computational basis states of $\mathbb{C}^{q}$ to those of $(\mathbb{C}^2)^{\otimes s}$, all this definition does is simply take some $U \in U(\mathbb{C}^{q})$, imagined as a matrix in terms of the computational basis, and translate it into exactly the same matrix in terms of the computational basis of $(\mathbb{C}^2)^{\otimes s}$. It is, therefore, clear that if $U$ is a unitary, then indeed $\theta(U)$ is a unitary. Also, given unitaries $U_1, U_2 \in U\left(\mathbb{C}^{q}\right)$, and given any $v \in (\mathbb{C}^2)^{\otimes s}$, we have
\begin{align}
    \theta(U_1U_2)(v) &= \psi U_1 U_2\psi^{-1}(v)\\
    &=\psi U_1\psi^{-1}\psi U_2\psi^{-1}(v)\\
    &= \theta(U_1)\theta(U_2)(v)
\end{align}
which tells us that $\theta(U_1U_2) = \theta(U_1)\theta(U_2)$. Because $\theta$ is invertible, we find that $\theta$ is an isomorphism of the groups of unitaries. We claim, moreover, that $\theta$ specialises to an isomorphism from $\mathcal{P}_{1,q}$ to $\mathcal{P}_{s,2}$, and indeed that it acts in the same way as $\theta_P$. Considering any $\mathbf{e} \in \mathbb{F}_2^s$ and $\eta = \sum_{i=0}^{s-1}e_i\alpha_i$, we have
\begin{align}
    \theta(i^aX^\beta Z^\gamma)\ket{\mathbf{e}} &= \psi i^aX^\beta Z^\gamma\psi^{-1}\ket{\mathbf{e}}\\
    &=\psi i^a X^\beta Z^\gamma \ket{\eta}\\
    &= i^a(-1)^{\tr(\gamma\eta)}\psi\ket{\eta+\beta}\\
    &= i^a(-1)^{\tr(\gamma\eta)}\ket{\mathbf{e}+\mathbf{b}}\\
    &= i^aX^{\mathbf{b}}Z^{\mathbf{c}}\ket{\mathbf{e}}\\
    &= \theta_P(i^aX^\beta Z^\gamma)\ket{\mathbf{e}}
\end{align}
so that indeed we find $\theta(P) = \theta_P(P)$ for all $P \in \mathcal{P}_{1,q}$, thus confirming that $\theta$ specialises to an isomorphism between Pauli operators, and acts in the same way as $\theta_P$. With the usual definition of the Clifford group on $s$ qubits:
\begin{equation}
    \mathcal{C}_{s,2} = \{U \in U\left((\mathbb{C}^2)^{\otimes s}\right):UPU^\dagger \in \mathcal{P}_{s,2} \text{ for all } P \in \mathcal{P}_{s,2}\},
\end{equation}
we can check that $\theta$ specialises to an isomorphism from $\mathcal{C}_{1,q}$ to $\mathcal{C}_{s,2}$. Indeed, let $U \in \mathcal{C}_{1,q}$ and $P \in \mathcal{P}_{s,2}$. Then we find
\begin{equation}
    \theta(U)P\theta(U)^\dagger = \theta(U\theta^{-1}(P)U^\dagger)
\end{equation}
and because $\theta^{-1}(P) \in \mathcal{P}_{1,q}$, we have $U\theta^{-1}(P)U^\dagger \in \mathcal{P}_{1,q}$, which leads to $\theta(U)P\theta(U)^\dagger \in \mathcal{P}_{s,2}$, from which we conclude that $\theta(U) \in \mathcal{C}_{s,2}$. A similar calculation shows that if $U \in \mathcal{C}_{s,2}$ then $\theta^{-1}(U) \in \mathcal{C}_{1,q}$, showing that indeed $\theta$ specialises to an isomorphism from $\mathcal{C}_{1,q}$ to $\mathcal{C}_{s,2}$.

The next thing to consider is the Clifford hierarchy, originally introduced by Gottesman and Chuang \cite{gottesman1999demonstrating}, and defined recursively in our setting as
\begin{align}
    \mathcal{C}_{1,q}^{(1)} &= \mathcal{P}_{1,q},\\
    \mathcal{C}_{1,q}^{(k)} &= \{U \in U(\mathbb{C}^{q}): UPU^\dagger \in \mathcal{C}_{1,q}^{(k-1)}\}\text{ for }k > 1,
\end{align}
for one qudit of dimension $q$, and similarly for $s$ qubits. Note that, in particular, $\mathcal{C}^{(2)}_{1,q} = \mathcal{C}_{1,q}$ and $\mathcal{C}^{(2)}_{s,2} = \mathcal{C}_{s,2}$. Of great interest in this work will be diagonal gates in the third level of the Clifford hierarchy. While the gates in the $k$-th level of the Clifford hierarchy do not form a group for $k > 2$, the diagonal gates in the $k$-th level always do \cite{zeng2008semi,cui2017diagonal}. By an easy induction using the same ideas as above, and the fact that $\theta$ maps diagonal gates to diagonal gates, it may be seen that $\theta$ restricts to an isomorphism between the diagonal gates in the $k$-th level of the Clifford hierarchy for one $q$-dimensional qudit and $s$ qubits for every $k \geq 1$.

With this formalism set up for the translation between individual $q$-dimensional qudits and sets of $s$ qubits, it remains to consider what it means to build a stabiliser code for qudits of dimension $q$. Again, this is covered in \cite{gottesman2016surviving}. The only codes we will be interested in are CSS codes, which we recall are stabiliser codes in which each generator of the stabiliser may be taken to be either an $X$ Pauli operator or a $Z$ Pauli operator. 

Consider again two classical linear codes $\mathcal{L}_X \subseteq \mathbb{F}_q^n$ and $\mathcal{L}_Z \subseteq \mathbb{F}_q^n$ \footnote{In particular, $\mathcal{L}_X$ and $\mathcal{L}_Z$ are $\mathbb{F}_q$-linear spaces.}, where $n$ is the number of qudits on which our code is defined. To define a meaningful quantum code, the requirement placed on these spaces is that $\mathcal{L}_X \subseteq \mathcal{L}_Z^\perp$, where
\begin{equation}
    \mathcal{L}_Z^\perp = \left\{x \in \mathbb{F}_q^n: \sum_{i=1}^nx_iy_i = 0\text{ for all } y \in \mathcal{L}_Z\right\}
\end{equation}
and arithmetic in this definition occurs over $\mathbb{F}_q$. We then define the code $CSS(X,\mathcal{L}_X;Z,\mathcal{L}_Z) \subseteq (\mathbb{C}^q)^{\otimes n}$ in the usual way:
\begin{equation}
    CSS(X,\mathcal{L}_X;Z,\mathcal{L}_Z) = \left\{\ket{\psi} \in (\mathbb{C}^q)^{\otimes n}: X^{\mathbf{x}}\ket{\psi} = \ket{\psi}\text{ and } Z^{\mathbf{z}}\ket{\psi} = \ket{\psi} \text{ for all } \mathbf{x} \in \mathcal{L}_X \text{ and } \mathbf{z} \in \mathcal{L}_Z\right\},
\end{equation}
which can be defined in words as the simultaneous $+1$-eigenspace of all the $X$-stabiliser elements described by elements of $\mathcal{L}_X$ and $Z$-stabiliser elements described by elements of $\mathcal{L}_Z$. We emphasise again the use of the notation $X^{\mathbf{x}}$ with $\mathbf{x} \in \mathbb{F}_q^n$; $X^{\mathbf{x}} = \bigotimes_{i=1}^nX^{\mathbf{x}_i}$. The condition $\mathcal{L}_X \subseteq \mathcal{L}_Z^\perp$ is what makes this definition make sense. Indeed, given $\mathbf{x} \in \mathcal{L}_X$ and $\mathbf{z} \in \mathcal{L}_Z$, the stabilisers $X^{\mathbf{x}}$ and $Z^{\mathbf{z}}$ commute:
\begin{equation}
    X^{\mathbf{x}}Z^{\mathbf{z}} = (-1)^{\sum_{i=1}^n\tr(\mathbf{x}_i\mathbf{z}_i)}Z^{\mathbf{z}}X^{\mathbf{x}}= (-1)^{\tr(\sum_{i=1}^n\mathbf{x}_i\mathbf{z}_i)}Z^{\mathbf{z}}X^{\mathbf{x}}. = Z^{\mathbf{z}}X^{\mathbf{x}}. 
\end{equation}
We emphasise here that the sum can be brought inside the trace because $\tr$ is a linear map from $\mathbb{F}_q$ to $\mathbb{F}_2$ (see Proposition \ref{traceProps}) and because the sum is in an exponent with base $-1$, meaning that the sum outside the trace may be taken modulo 2.

We note that because $\mathcal{L}_X$ and $\mathcal{L}_Z$ are $\mathbb{F}_q$-linear vector spaces, the quantum code $CSS(X,\mathcal{L}_X;Z,\mathcal{L}_Z)$ is a ``true $\mathbb{F}_q$-stabiliser code'' in the sense of \cite{gottesman2016surviving}. What this means for us is that the code encodes an integer number of $q$-dimensional qudits. In particular, as a complex vector space,
\begin{equation}\label{numEncodedQudits}
    \dim CSS(X,\mathcal{L}_X;Z,\mathcal{L}_Z) = q^{n-\dim\mathcal{L}_X-\dim\mathcal{L}_Z},
\end{equation}
where $\dim\mathcal{L}_X$ and $\dim\mathcal{L}_Z$ denote the dimensions of $\mathcal{L}_X$ and $\mathcal{L}_Z$ as vector spaces over $\mathbb{F}_q$, respectively. From this, we learn that the code encodes $n-\dim\mathcal{L}_X-\dim\mathcal{L}_Z$ qudits of dimension $q$. However, this can all be translated into a statement about qubit codes. Indeed, the code may be equally well thought of as a stabiliser code on $ns$ qubits, which encodes $s(n-\dim\mathcal{L}_X-\dim\mathcal{L}_Z)$ logical qubits. This may be achieved by simply mapping every state in $CSS(X,\mathcal{L}_X;Z,\mathcal{L}_Z)$ under $\psi^{\otimes n}$. Equivalently, one may define the corresponding qubit code by mapping every stabiliser $X^{\mathbf{x}}$ and $Z^{\mathbf{z}}$ under $\theta^{\otimes n}$ to form a qubit stabiliser group. 

Let us go into detail about how to determine a list of stabiliser generators for the qubit code from those of the qudit code. Let us suppose that $\hat{\mathcal{L}}_X\subseteq \mathcal{L}_X$ is an $\mathbb{F}_q$-basis for $\mathcal{L}_X$, and that $\hat{\mathcal{L}}_Z\subseteq \mathcal{L}_Z$ is an $\mathbb{F}_q$-basis for $\mathcal{L}_Z$. With $\{\alpha_i\}_{i=0}^{s-1}$ our self-dual basis for $\mathbb{F}_q$ over $\mathbb{F}_2$, we then let
\begin{align}
    \tilde{\mathcal{L}}_X &\coloneq \{\alpha_i\}_{i=0}^{s-1}\cdot\hat{\mathcal{L}}_X\\
    \tilde{\mathcal{L}}_Z &\coloneq \{\alpha_i\}_{i=0}^{s-1}\cdot\hat{\mathcal{L}}_Z.
\end{align}
We then have a set of $X$-stabiliser generators for the qubit code
\begin{equation}
    \{\theta_P(X^{v}): v \in \tilde{\mathcal{L}}_X\}
\end{equation}
and $Z$-stabiliser generators
\begin{equation}
    \{\theta_P(Z^{v}): v \in \tilde{\mathcal{L}}_Z\}
\end{equation}
where, as always $X^v = \bigotimes_{i=1}^nX^{v_i}$, and similarly for $Z$. Aside from the fact that all these operators are independent qubit operators, we will see why this makes sense as we now discuss what it means to measure stabiliser generators in the case of $q$-dimensional qudits. Given a code state $\ket{\psi} \in CSS(X,\mathcal{L}_X;Z,\mathcal{L}_Z)$, let us suppose it is affected by an error $Z^{v_\mathrm{err}}$ for some $v_\mathrm{err} \in \mathbb{F}_q^n$. Now given some $v_\mathrm{stab} \in \hat{\mathcal{L}}_X$, measuring the stabiliser element $X^{v_\mathrm{stab}}$ leads to an eigenvalue
\begin{equation}
    (-1)^{\tr(v_\mathrm{stab}\cdot v_\mathrm{err})}
\end{equation}
where the dot product is taken over $\mathbb{F}_q^n$. On the other hand, one can also measure the stabiliser element for $\alpha_iv_\mathrm{stab}$ to obtain eigenvalues
\begin{equation}\label{basisMeasurementsSDB}
    (-1)^{\tr(\alpha_iv_\mathrm{stab}\cdot v_\mathrm{err})}.
\end{equation}
We claim, then, that given knowledge of all elements in Equation \eqref{basisMeasurementsSDB} (for $i=0, \ldots, s-1$), we can determine the value of $v_\mathrm{stab}\cdot v_\mathrm{err}$. Indeed, writing $v_\mathrm{stab}\cdot v_\mathrm{err} = \sum_{i=0}^{s-1}\omega_i\alpha_i$ for $\omega_i \in \mathbb{F}_2$, one finds that
\begin{equation}
    \tr(\alpha_i v_\mathrm{stab}\cdot v_\mathrm{err}) = \omega_i
\end{equation}
from which $v_\mathrm{stab} \cdot v_\mathrm{err}$ may be deduced. Ultimately, if we measure the stabilisers of a qudit code (where the qudits are of dimension $q$), we are measuring the eigenvalues of the operators in $\tilde{\mathcal{L}}_X$ and $\tilde{\mathcal{L}}_Z$, which will be in $\pm 1$, but these values in $\pm 1$ can equivalently be thought of as (a shorter list of) elements in $\mathbb{F}_q$.

As for the code's error correction capability, we have that, as always, the set of undetectable logical Pauli errors is formed of the set of Paulis that commute with all stabilisers, while themselves not being stabilisers. This set of undetectable logical Pauli errors can then either be thought of as a set of Pauli operators acting on $n$ qudits of dimension $q$, or equivalently as a set of Pauli operators acting on $ns$ qubits. As sets of operators, these are entirely equivalent, although there is a slight difference when we think about the distance of the code. The distance is always defined as the minimum weight of an undetectable logical Pauli error, although in general this will be different based on whether we count the weight in terms of non-trivial Pauli operators acting on $q$-dimensional qudits, or in terms of non-trivial qubit Pauli operators\footnote{Note also that if one has a fixed code determined by $\mathbb{F}_q$-linear spaces $\mathcal{L}_X \subseteq \mathbb{F}_q^n$ and $\mathcal{L}_Z \subseteq \mathbb{F}_q^n$, then the set of logical error operators, when thought of in terms of qubit operators will change if one changes the isomorphism between the Pauli groups. Recall from earlier that this isomorphism may be changed by changing the self-dual basis of $\mathbb{F}_{q}$ over $\mathbb{F}_2$ that one uses.}. However, as long as one is aware of this difference, this will not be a problem for us. Indeed, because a Pauli operator acting on one $q$-dimensional qudit is non-trivial if and only if the corresponding Pauli operator acting on $s$ qubits is non-trivial, the distance of the code, when counted in terms of qubit Pauli operators, will always be at least the distance of the code when counted in terms of $q$-dimensional qudit Pauli operators.

\subsection{Algebraic Geometry Codes}\label{AGPrelims}

In this work, we employ a class of classical evaluation codes: algebraic geometry codes. It would be impossible to give a full account of the subject of algebraic geometry codes in a reasonable amount of space. Nevertheless, we will give here the essentials for our purposes, referring the reader to one of the many excellent references on the subject for more details. In particular, for a purely algebraic approach, we recommend Stichtenoth \cite{stichtenoth2009algebraic}, or for a more detailed treatment of the interplay between algebraic function fields and varieties, see Niederreiter and Xing \cite{niederreiter2009algebraic}. The following treatment follows that of Stichtenoth closely. Throughout this section, $\mathbb{F}_q$ represents the field of $q$ elements, where $q=p^s$ may be any prime power, although the case $q=2^{10}=1024$ ($p=2$ and $s=10$) is particularly relevant in this work.

\subsubsection{Function Fields, Places and Valuations}

Let us begin with the most fundamental definitions in the algebraic treatment of algebraic geometry codes. Function fields are the large algebraic objects with which algebraic geometry codes are built, and `places' and `valuations' are two ways of describing the ``points'' in function fields at which the codes are defined.

\begin{definition}[Field Extensions and the Degree of a Field Extension]
A field extension $L$ of $K$, denoted $K \subseteq L$ or $L/K$, is nothing more than a field $L$ that contains another, $K$. Given such a field extension, $L$ may be viewed as a vector space over the field $K$. The degree of the extension, written $[L:K]$, is the dimension of this vector space, and the extension is called finite if its degree is finite.
\end{definition}
As a pertinent example, the finite field $\mathbb{F}_{p^s}$ is a field extension of $\mathbb{F}_p$ of degree $s$.

\begin{definition}[Algebraic Function Field over a Finite Field]
    An algebraic function field of one variable $F$ over $\mathbb{F}_q$ is a finite field extension of $\mathbb{F}_q(x)$ (where $\mathbb{F}_q(x)$ denotes the field of rational functions in the variable $x$ with coefficients in the field $\mathbb{F}_q$).
\end{definition}

An algebraic function field of one variable may simply be referred to as an `algebraic function field', or even just as a `function field' when the meaning is clear. We denote an algebraic function field $F$ over $\mathbb{F}_q$ as $F/\mathbb{F}_q$.

\begin{definition}[Valuation Ring] A valuation ring of a function field $F/\mathbb{F}_q$ is a ring $\mathcal{O}\subseteq F$ such that:
\begin{enumerate}
    \item $\mathbb{F}_q \subsetneq \mathcal{O} \subsetneq F$;
    \item for each $x \in F$, $x \in \mathcal{O}$ or $x^{-1} \in \mathcal{O}$.
\end{enumerate}
\end{definition}

As a subring of a field, a valuation ring $\mathcal{O}$ naturally forms a commutative ring. Given a valuation ring $\mathcal{O}$, a subset $I\subset\mathcal{O}$ is called an ideal if $I$ forms an additive subgroup of $\mathcal{O}$, and for every $z\in\mathcal{O}$ and $x\in I$, their product is in $I$, i.e., $zx\in I$.
An ideal $I$ that forms a proper subset of $\mathcal{O}$, $I\neq\mathcal{O}$, is called a proper ideal.
A proper ideal $I$ is called a maximal ideal if there exists no other proper ideal $J$ with $I$ a proper subset of $J$.
For the ring $\mathcal{O}$, a unit is an element having a multiplicative inverse in $\mathcal{O}$, and the set of units is denoted by $\mathcal{O}^\times$.
The valuation ring $\mathcal{O}$ has the following properties.

\begin{proposition}[From Proposition 1.1.5 and Theorem 1.1.6 of \cite{stichtenoth2009algebraic}. See also Theorem 1.1.13 of~\cite{stichtenoth2009algebraic}]
    Let $\mathcal{O}$ be a valuation ring of the function field $F/\mathbb{F}_q$. Then the following hold.
    \begin{enumerate}
        \item The valuation ring $\mathcal{O}$ is a discrete valuation ring. To be explicit, $\mathcal{O}$ contains an element $t$, called a prime element, such that for any $0 \neq z \in \mathcal{O}$, $z$ may be written uniquely as $z=ut^n$ for some $u \in \mathcal{O}^\times$ and $n \in \mathbb{Z}_{\geq 0}$, where $\mathcal{O}^\times$ is the set of units in $\mathcal{O}$.
        \item The valuation ring $\mathcal{O}$ contains a unique maximal ideal which we denote $P$. If $t$ is a prime element of $\mathcal{O}$, we have
        \begin{equation}
            P = \{ut^n : u \in \mathcal{O}^\times, n \in \mathbb{Z}_{> 0}\} \cup \{0\}.
        \end{equation}
        One may, therefore, write $P = t\mathcal{O}$.
        \item Let $t$ be a prime element of $\mathcal{O}$. Then any element $0 \neq x \in F$ may be uniquely represented as $x=ut^n$ for $n \in \mathbb{Z}$ and $u \in \mathcal{O}^\times$.
    \end{enumerate}
\end{proposition}

It follows from the first and third points that $F$ is nothing more than the quotient field (also known as the field of fractions) of $\mathcal{O}$, and this is true for any $\mathcal{O}$.

The maximal ideal $P$ of a valuation ring $\mathcal{O}$ is of great importance and is given a name.

\begin{definition}[Place of a Function Field]

A place $P$ of a function field $F/\mathbb{F}_q$ is the (unique) maximal ideal of some valuation ring in $F$. We let
\begin{align}
    \mathbb{P}_F \coloneq \{P:P \text{ is a place of $F/\mathbb{F}_q$}\}
\end{align}
denote the set of all places of $F/\mathbb{F}_q$.    
\end{definition}

Our next definition appears to be a slight detour although its significance will become clear.

\begin{definition}\label{discreteValuations}
    A discrete valuation of $F/\mathbb{F}_q$ is a function $\nu:F \to \mathbb{Z} \cup \{\infty\}$ such that, for all $x,y \in F$,
    \begin{enumerate}
        \item $\nu(x) = \infty \iff x = 0$;
        \item $\nu(xy) = \nu(x) + \nu(y)$;
        \item $\nu(x+y) \geq \min\{\nu(x),\nu(y)\}$;
        \item there exists an element $z \in F$ such that $\nu(z) = 1$;
        \item $\nu(a) = 0$ for all $0 \neq a \in \mathbb{F}_q$.
    \end{enumerate}
\end{definition}
Here, $\infty$ is a symbol that behaves as $\infty + \infty = \infty + n = n + \infty  = \infty$ and $\infty > n$ for every $n \in \mathbb{Z}$. We now connect discrete valuations to our previous discussion.

\begin{lemma}[From Theorem 1.1.13 of \cite{stichtenoth2009algebraic} and the preceding remark]\label{PlaceValuation}
Let $\mathcal{O}$ be a valuation ring of the function field $F/\mathbb{F}_q$ with a place $P$, and let us choose a prime element $t$ of $\mathcal{O}$. The function $\nu_P: F \to \mathbb{Z} \cup \{\infty\}$ defined by
\begin{equation}
    \nu_P(x)\coloneq \begin{cases}
        n &\text{ if } 0 \neq x = ut^n \text{ for } u \in \mathcal{O}^\times\\
        \infty &\text{ if } x = 0
    \end{cases}
\end{equation}
is a discrete valuation of the function field $F/\mathbb{F}_q$. This map is independent of the choice of prime element $t$ in the valuation ring $\mathcal{O}$.
\end{lemma}

With this, it turns out that places, valuation rings, and discrete valuations of a function field $F/\mathbb{F}_q$ can be thought of as the same thing.

\begin{theorem}[From Theorem 1.1.13 and the remark preceding Definition 1.1.9 of \cite{stichtenoth2009algebraic}]
\label{thm:places_valuation_rings_discrete_valuations}

Let $F/\mathbb{F}_q$ be a function field. There is a natural one-to-one correspondence between its places, valuation rings, and discrete valuations as follows.

\begin{enumerate}
    \item Given a valuation ring $\mathcal{O}$, it contains a unique place $P$: its unique maximal ideal. On the other hand, given any place $P$, it arises from a unique valuation ring $\mathcal{O} = \{x \in F: x^{-1} \notin P\}$.
    \item Given a valuation ring $\mathcal{O}$ and its place $P$, they specify a unique discrete valuation $\nu_P$ as described in Lemma \ref{PlaceValuation}. Conversely, any discrete valuation $\nu$ specifies a unique place and corresponding valuation ring of the function field via
    \begin{align}
        P &= \{x \in F: \nu(x) > 0\},\\
        \mathcal{O} &= \{x \in F: \nu(x) \geq 0\},
    \end{align}
respectively.    
\end{enumerate}
\end{theorem}

It therefore makes sense to denote by $\mathcal{O}_P$ the valuation ring corresponding to the place $P$. Because $P$ is an ideal of $\mathcal{O}_P$, we can consider the quotient ring $\mathcal{O}_P/P$ and, because $P$ is a maximal ideal, this forms a field.

\begin{definition}[Residue Class Field and Residue Class Map]\label{residueClassFieldAndMapDef}
    Given a place $P$ and corresponding valuation ring $\mathcal{O}_P$ of the function field $F/\mathbb{F}_q$, we define the residue class field of $P$ as $F_P = \mathcal{O}_P/P$. The residue class map of $P$ is the quotient map from $\mathcal{O}_P$ into $F_P$:
    \begin{align}
        \mathcal{O}_P \ni x \mapsto x(P) \coloneq x+P \in F_P
    \end{align}
    which is then a ring homomorphism\footnote{For the reader that references \cite{stichtenoth2009algebraic}, we note that Stichtenoth defines the residue class map as a map from $F$ into $F_P$ acting on elements in $\mathcal{O}_P$ as in our definition, but mapping any element not in $\mathcal{O}_P$ to the symbol $\infty$. For our purposes, it is simpler to just define the residue class map on elements in $\mathcal{O}_P$, as is done in the book of Niederreiter and Xing \cite{niederreiter2009algebraic} (see their Definition 1.5.10). There is no difference in what follows.}.
\end{definition}

We can now elucidate the structure of residue class fields. 

\begin{proposition}[From Proposition 1.1.15 of \cite{stichtenoth2009algebraic} and the discussion preceding Definition 1.1.14]\label{residueClassFieldStructure}

Every residue class field $F_P$ contains (an isomorphic copy of the field) $\mathbb{F}_q$. Moreover, $F_P$ is a finite extension of $\mathbb{F}_q$.
    
\end{proposition}

\begin{definition}[Degree of a Place]
\label{def:degreee_place}
    Given a place $P$ of a function field $F/\mathbb{F}_q$, the degree of the place $P$ is
    \begin{align}
    \label{eq:deg_P}
        \deg(P) = [F_P:\mathbb{F}_q],
    \end{align}
    the degree of the field extension $F_P/\mathbb{F}_q$. The place $P$ is called rational if $\deg(P) = 1$.
\end{definition}
Note that Proposition \ref{residueClassFieldStructure} tells us that the degree of any place is finite. We emphasise the important fact that the residue class field of a rational place is nothing more than (an isomorphic copy of) $\mathbb{F}_q$.

\subsubsection{Divisors, Differentials, and the Riemann-Roch Theorem}

This section aims to present the Riemann-Roch theorem, which gives an expression for the dimension of a Riemann-Roch space, which we will also need to define. The Riemann-Roch theorem ultimately gives a useful expression for the dimension of an algebraic geometry code.

\begin{definition}[Divisors and the Divisor Group]
\label{def:divisor}
Given a function field $F/\mathbb{F}_q$, a divisor is a formal sum
\begin{equation}
    D = \sum_{P \in \mathbb{P}_F}n_PP\label{typicalDivisor}
\end{equation}
with coefficients $n_P \in \mathbb{Z}$ and $n_p = 0$ for all but finitely many places $P\in\mathbb{P}_F$. The divisors collectively form the divisor group $\mathrm{Div}(F)$ under the obvious addition, which is the free abelian group generated by the places of $F/\mathbb{F}_q$. Given the divisor $D$ in Equation \eqref{typicalDivisor}, it is useful to define
\begin{align}
    \nu_P(D) \coloneq n_P.
\end{align}
The support of a divisor $\supp(D)$ is the set of places $P \in \mathbb{P}_F$ on which $\nu_P(D) \neq 0$.
\end{definition}

\begin{definition}[Positive Divisors]
\label{def:positive_divisor}
A divisor $D$ is called positive if $\nu_P(D) \geq 0$ for all $P \in \mathbb{P}_F$,
where $\nu_P(D)$ is defined in Definition~\ref{def:divisor}; in this case, we write
\begin{align}
    D \geq 0.
\end{align}
For two divisors $D_1$ and $D_2$, we write $D_1 \geq D_2$ if $D_1 - D_2 \geq 0$ (where 0 represents the zero divisor).
\end{definition}

\begin{definition}[Degree of a Divisor]
\label{def:degree_divisor}
The degree of a divisor $D$ is defined as
\begin{equation}
    \deg(D) = \sum_{P \in \mathbb{P}_F}\nu_P(D)\deg(P),
\end{equation}
where $\deg(P)$ is defined in Definition~\ref{def:degreee_place}.
Note that $\deg : \mathrm{Div}(F) \to \mathbb{Z}$ is a group homomorphism.
\end{definition}

\begin{definition}[Principal Divisors] \label{principalDivisorDefinition}Given $0 \neq x \in F$, the principal divisor of $x$ is
\begin{equation}
    (x) \coloneq \sum_{P \in \mathbb{P}_F}\nu_P(x)P,
\end{equation}
where $\nu_P(x)$ is given in Lemma~\ref{PlaceValuation}.
\end{definition}

The principal divisor $(x)$ does indeed form a divisor in light of the fact that for all $0 \neq x \in F$, the set of places $P$ for which $\nu_P(x) \neq 0$ is finite (see Definition 1.1.18 and Corollary 1.3.4 of \cite{stichtenoth2009algebraic}). At this point, we may make the following highly important definition.

\begin{definition}[Riemann-Roch Space]
\label{def:riemann_space}
    Given a function field $F/\mathbb{F}_q$ and a divisor $D \in \mathrm{Div}(F)$, the Riemann-Roch space corresponding to $D$ is
    \begin{equation}
        \mathcal{L}(D)\coloneq \{x \in F: D + (x) \geq 0\} \cup \{0\},
    \end{equation}
    where the principal divisor $(x)$ is defined in Definition~\ref{principalDivisorDefinition}, and the notion of a positive divisor $ D + (x) \geq 0$ is defined in Definition~\ref{def:positive_divisor}. 
It can be checked that $\mathcal{L}(D)$ forms a vector space over $\mathbb{F}_q$. We denote the dimension of this vector space as $l(D)$.
\end{definition}

The Riemann-Roch theorem gives a useful expression for the quantity $l(D)$. We will require even more definitions before we can state it fully, but we can say a little about $l(D)$ now.

\begin{proposition}[Proposition 1.4.14 in \cite{stichtenoth2009algebraic}]Given a function field $F/\mathbb{F}_q$, there is a constant $\gamma \in \mathbb{Z}$ depending only on $F$ such that for any $D \in \mathrm{Div}(F)$,
\begin{equation}
    \deg(D)-l(D) \leq \gamma.
\end{equation}
\end{proposition}

In light of this, the following definition makes sense.

\begin{definition}[Genus of a Function Field]
    The genus $g$ of a function field $F/\mathbb{F}_q$ is the integer
    \begin{equation}
        g \coloneq \max\{\deg(D)-l(D) + 1:D \in \mathrm{Div}(F)\}.
    \end{equation}
    Note that, due to Corollary~1.4.16 of~\cite{stichtenoth2009algebraic}, we have
    \begin{align}
        g\geq 0.
    \end{align}
\end{definition}

To state the Riemann-Roch theorem, we will need the concept of Weil differentials, and to understand Weil differentials we need the concept of adeles.

\begin{definition}[Adeles and Adele Spaces]
\label{def:adeles}
Given a function field $F/\mathbb{F}_q$, an adele of $F$ is a mapping $\alpha : \mathbb{P}_F \to F$ such that $\alpha(P) \in \mathcal{O}_P$ for all but finitely many places $P$. We write
\begin{align}
    \alpha_P &\coloneq \alpha(P),\\
    \nu_P(\alpha) &\coloneq \nu_P(\alpha_P).
\end{align}
Then,
\begin{equation}
    \mathcal{A}_F \coloneq \{\alpha : \alpha \text{ is an adele of } F/\mathbb{F}_q\}
\end{equation}
is called the adele space of $F$, and it forms a vector space over $\mathbb{F}_q$. Given any $D \in \mathrm{Div}(F)$, we may define an $\mathbb{F}_q$-linear subspace of $\mathcal{A}_F$:
\begin{equation}
    \mathcal{A}_F(D) \coloneq \{\alpha \in \mathcal{A}_F: \nu_P(\alpha) + \nu_P(D) \geq 0 \text{ for all } P \in \mathbb{P}_F\},
\end{equation}
where $\nu_P(D)$ is defined in Definition~\ref{def:divisor}.
\end{definition}
\begin{definition}[Principal Adele]
\label{def:principal_adele}
    Given any $x \in F$, the principal adele of $x$, denoted $\mathrm{princ}(x)$, is the adele which takes value $x$ at every place $P \in \mathbb{P}_F$.
\end{definition}
In a very similar way to how principal divisors are well-defined divisors, principal adeles are well-defined adeles because for any $x \in F$, $\nu_P(x) < 0$ for finitely many places $P \in \mathbb{P}_F$ (see Definition 1.1.18 and Corollary 1.3.4 of \cite{stichtenoth2009algebraic}). It can be checked that the set of principal adeles forms an $\mathbb{F}_q$-linear subspace of $\mathcal{A}_F$ which is isomorphic to $F$. When we refer to $F$ as a subspace of $\mathcal{A}_F$, we will always mean the space of principal adeles.

Using the adele space $\mathcal{A}_F$, we may define Weil differentials.
\begin{definition}[Weil Differentials]\label{differentialDef}
    A Weil differential of $F$ is an $\mathbb{F}_q$-linear map $\omega: \mathcal{A}_F \to \mathbb{F}_q$ such that
    \begin{align}
        \omega|_{\mathcal{A}_F(D)+F} \equiv 0
    \end{align}
    for some $D \in \mathrm{Div}(F)$, meaning that $\omega$ vanishes on the subspace $\mathcal{A}_F(D)+F \subseteq \mathcal{A}_F$. We denote the set of all Weil differentials of $F$ by $\Omega_F$. Given a fixed divisor $D \in \mathrm{Div}(F)$, we write $\Omega_F(D)$ for the set of Weil differentials that vanish on $\mathcal{A}_F(D)+F$.
It can be shown that $\Omega_F$ forms a vector space over $\mathbb{F}_q$, and $\Omega_F(D)$ is a linear subspace for any $D \in \mathrm{Div}(F)$.
\end{definition}

Using the concept of Weil differentials, we introduce the notion of canonical divisors based on the following proposition.
\begin{proposition}[Lemma 1.5.10 of \cite{stichtenoth2009algebraic}]
    Given a function field $F/\mathbb{F}_q$ and $0 \neq \omega \in \Omega_F$, let us consider
    \begin{equation}
        M(\omega) \coloneq \{D \in \mathrm{Div}(F): \omega \text{ vanishes on } \mathcal{A}_F(D)+F\}.
    \end{equation}
    Then, there exists a uniquely determined divisor $(\omega) \in M(\omega)$ such that
    \begin{equation}
        M(\omega) = \{D \in \mathrm{Div}(F): D \leq (\omega)\},
    \end{equation}
    where $D \leq (\omega)$ means that $(\omega)-D\geq 0$ is a positive divisor as in Definition~\ref{def:positive_divisor}.
\end{proposition}

\begin{definition}[Canonical Divisors]
\label{def:canonical_divisor}
    A divisor $W \in \mathrm{Div}(F)$ is a canonical divisor if $W = (\omega)$ for some $0 \neq \omega \in \Omega_F$.
\end{definition}
\begin{remark}[Remark~1.5.12 of~\cite{stichtenoth2009algebraic}]\label{canonicalDivisors}
We may recharacterise the space $\Omega_F(D)$ as
\begin{equation}
    \Omega_F(D) = \{\omega \in \Omega_F: \omega = 0 \text{ or } (\omega) \geq D\}.
\end{equation}
\end{remark}

With this, we may at last state the most important result of this section.

\begin{theorem}[Riemann-Roch Theorem --- Theorem~1.5.15 of~\cite{stichtenoth2009algebraic}]\label{riemannRochTheorem} Let $F/\mathbb{F}_q$ be a function field of genus $g$. Let $A \in \mathrm{Div}(F)$ be any divisor of $F$ and $W$ any canonical divisor. Then
\begin{equation}
    l(A) = \deg(A) + 1 - g + l(W-A),
\end{equation}
where $\deg(A)$ is defined in Definition~\ref{def:degree_divisor}, and $l(A)$ and $l(W-A)$ are defined in Definition~\ref{def:riemann_space}.
Moreover, if $\deg(A) \geq 2g-1$, we have
\begin{align}
    l(A) = \deg(A) + 1 - g.    
\end{align}
\end{theorem}
We remark that the simpler result
\begin{align}
\label{eq:riemannTheorem}
    l(A) \geq \deg(A) + 1 - g
\end{align}
is simply called Riemann's theorem. Also, by setting $A=0$ followed by $A=W$ in the Riemann-Roch theorem, and using $l(0) = 1$, one finds that  for any canonical divisor $W$ 
\begin{align}
    l(W) &= g,\\
\label{eq:degree_canonical_divisor}
    \deg(W) &= 2g-2.
\end{align}

Now that we have seen the Riemann-Roch theorem, we finish this section with a small number of further definitions and results that will be important for us. The first result below is sometimes called the `Duality Theorem'.

\begin{theorem}[Duality Theorem --- Theorem 1.5.14 of \cite{stichtenoth2009algebraic}]\label{dualityThm}
    Let $A \in \mathrm{Div}(F)$ be any divisor and $W$ any canonical divisor. Then the spaces $\mathcal{L}(W-A)$ and $\Omega_F(A)$ are isomorphic as vector spaces over $\mathbb{F}_q$.
\end{theorem}

Finally, we introduce local components of differentials, followed by presenting a proposition showing how to express the value of a differential acting on an adele in terms of its local components.

\begin{definition}[Local Components of Differentials]
\label{def:local_components_of_differentials}
Given a place $P$ and $x \in F$, one defines
\begin{align}
    \iota_P(x) \in \mathcal{A}_F
\end{align}
as the adele whose value is $x$ at $P$ and $0$ elsewhere. If we are given a Weil differential $\omega \in \Omega_F$, its local component at $P$ is the $\mathbb{F}_q$-linear map $\omega_P: F \to \mathbb{F}_q$ given by 
\begin{equation}
    \omega_P(x) \coloneq \omega(\iota_P(x)).
\end{equation}
\end{definition}

\begin{proposition}[Proposition 1.7.2 of \cite{stichtenoth2009algebraic}]\label{localComponentFormula}
    Given a Weil differential $\omega$ and an adele $\alpha$, $\omega_P(\alpha_P) = 0$ for all but finitely many places $P$ and, moreover
    \begin{equation}
        \omega(\alpha) = \sum_{P \in \mathbb{P}_F}\omega_P(\alpha_P).
    \end{equation}
\end{proposition}

\subsubsection{The Construction of Algebraic Geometry Codes, their Parameters and Decoding}\label{AGPrelims3}

Algebraic geometry codes are then constructed by evaluating the elements of some Riemann-Roch space via the residue class map on $n$ rational places.

\begin{definition}[Algebraic Geometry Codes]
\label{def:AG_code}
    Let $P_1, ..., P_n$ be pairwise distinct rational places of $F$ and set $D \coloneq P_1 + ... + P_n$. Let $A$ be a divisor such that $\supp(A) \cap \supp(D) = \emptyset$. Then we define the algebraic geometry code
    \begin{equation}
        C_{\mathcal{L}}(D,A) \coloneq \{(x(P_1), ..., x(P_n)): x \in \mathcal{L}(A)\} \subseteq \mathbb{F}_q^n.
    \end{equation}
\end{definition}
This definition makes sense for the following two reasons. First, $\supp(A) \cap \supp(D) = \emptyset$, and so by Definition~\ref{def:riemann_space} of the Riemann-Roch space, $x \in \mathcal{L}(A) \implies \nu_{P_i}(x) \geq 0$ for all $i=1, ..., n$, i.e., $x \in \mathcal{O}_{P_i}$ due to Theorem~\ref{thm:places_valuation_rings_discrete_valuations}; therefore, the value of the residue class map $x(P_i)$ in Definition~\ref{residueClassFieldAndMapDef} is well defined. Secondly, because $P_i$ are places of degree one (Definition~\ref{def:degreee_place}), the residue class field is $\mathbb{F}_q$, making the codewords indeed vectors in $\mathbb{F}_q^n$. The following theorem establishes the properties of $C_{\mathcal{L}}(D,A)$ that we need.
\begin{theorem}[From Corollary 2.2.3, Theorem 2.2.7 and Theorem 2.2.8 of \cite{stichtenoth2009algebraic}]\label{AGCodesParams}
    Let the function field $F$ have genus $g$. Furthermore, suppose $g \leq a \coloneq \deg(A) < n$. Then $C_{\mathcal{L}}(D,A)$ is a linear code of length $n$ over $\mathbb{F}_q$ with dimension
    \begin{equation}
        k \geq a+1-g.
    \end{equation}
    If, further, $a \geq 2g-1$, we have
    \begin{equation}
        k=a+1-g.
    \end{equation}
    Also under the hypotheses $g \leq a < n$ and $a \geq 2g-1$, the dual code $C_{\mathcal{L}}(D,A)^\perp$ has distance
    \begin{equation}
        d^\perp \geq a-(2g-2).
    \end{equation}
\end{theorem}

To construct algebraic geometry codes, it is therefore necessary to have function fields $F$ of genus $g$ with a large number of rational places compared to their genus. Indeed, to establish a non-zero distance on the dual code $C_{\mathcal{L}}(D,A)^\perp$ (which will be crucial for our purpose, magic state distillation), we need $2g-1 < n$, where $n$ is a number of rational places used to build the code. Given that the Hasse-Weil bound~\cite[Theorem~5.2.3]{stichtenoth2009algebraic} says that the number $N(F)$ of rational places of a function field $F/\mathbb{F}_q$ of genus $g$ always satisfies
\begin{equation}
    |N(F)-q-1| \leq 2gq^{1/2},
\end{equation}
the following definition makes sense.
\begin{definition}[Maximum Number of Rational Places for a given $q$ and $g$]
    For a given prime power $q$ and genus $g$, we denote by $N_q(g)$ the maximum number of rational places that a function field $F/\mathbb{F}_q$ can have if it has genus $g$.
\end{definition}
As we do in this work, it is common to build infinite families of algebraic geometry codes from function fields over a fixed finite field $\mathbb{F}_q$ with growing genus. To capture the number of rational places of such a field, we define the following important quantity.
\begin{definition}[Ihara's Constant]
For a given prime power $q$, Ihara's constant is defined as
\begin{equation}
    A(q) = \limsup_{g \to \infty} \frac{N_q(g)}{g}.
\end{equation}
\end{definition}
The way to interpret this is that for any fixed $q$ and $\delta > 0$, we can find an infinite sequence of genera $g_i\to \infty $ and an infinite family of function fields $F_i/\mathbb{F}_q$ of genus $g_i$ with a number of rational places at least $(A(q)-\delta)g_i$. For prime powers $q$ that are square, the following is true \cite{yasutaka1982some,garcia1995tower,garcia1996asymptotic}:
\begin{equation}
    A(q) = q^{1/2}-1,
\end{equation}
and so in particular
\begin{equation}
    A(1024) = 31.
\end{equation}

Let us conclude with a brief discussion on efficient classical decoding. It will be useful in this construction to be able to efficiently decode the dual codes of the algebraic geometry codes defined above. We denote the dual code of $C_{\mathcal{L}}(D,A)$ as
\begin{equation}
    C_\Omega(D,A) = C_{\mathcal{L}}(D,A)^\perp = \left\{v \in \mathbb{F}_q^n:\sum_{i=1}^nx_iv_i = 0\text{ for all }x \in C_{\mathcal{L}}(D,A)\right\}.
\end{equation}
We set a positive integer $t > 0$, which we call the decoding radius.
If $t$ and $A_1 \in \mathrm{Div}(F)$ satisfy
\begin{align}
    &\supp A_1 \cap \supp D = \emptyset,\label{decodingConditionFirst}\\
    &\deg A_1 < \deg A - (2g-2) - t,\\
    &l(A_1) > t,\label{decodingConditionLast}
\end{align}
then there is an efficient classical algorithm to decode the (classical) code $C_\Omega(D,A)$ from errors of weight at most $t$, as presented by~\cite[Decoding Algorithm~8.5.4]{stichtenoth2009algebraic}.
Note that if the weight of errors exceeds $t$, the decoder may not be able to find the correct codeword, but we assume that the decoder, even in such a case, still outputs some codeword that may have a logical error; that is, we work on a setting without post-selection.
The runtime of this decoder is $O(\mathrm{poly}(n))$, where the dominant step is to solve a system of linear equations as presented in Section~8.5 of~\cite{stichtenoth2009algebraic}.

\subsection{Task of Magic State Distillation and Error Model}
\label{sec:task}
We present the definition of magic state distillation analysed herein.
The goal of magic state distillation is to transform given noisy magic states into a smaller number of less noisy magic states via stabiliser operations only (i.e., operations composed of preparation of $\ket{0}$, $Z$-basis measurements, and Clifford gates).

In this paper, we consider magic state distillation on qubits.
We may have several possible choices of magic states, such as the $T$ state $\ket{T}=T\ket{+}$ and the $CCZ$ state $\ket{CCZ}=CCZ\ket{+}^{\otimes 3}$.
Our primary focus is on $CCZ$ gates (or equivalently, Toffoli gates), and we consider $\ket{CCZ}$ as the target magic state to be distilled, which enables the implementation of the $CCZ$ gate using stabiliser operations only via a gate teleportation procedure.

Then, the magic state distillation deals with the state transformation from $\zeta$ noisy $CCZ$ states into $\xi(\zeta)$ $CCZ$ states of a lower error rate.
We may write $\xi(\zeta)$ as simply $l$ if the dependency on $\zeta$ is obvious in context.
Following the convention of magic state distillation in~\cite{bravyi2005universal,bravyi2012magic}, we assume that stabiliser operations on qubits have no error, and the error source is the preparation operations to prepare initial noisy $CCZ$ states.
Note that this setting is well motivated by typical situations in fault-tolerant quantum computation, where fault-tolerant logical stabiliser operations are available by transversal implementations with their logical error rate arbitrarily suppressed, and logical magic states are prepared from physical magic states by state injection with their initial logical error rate in the same order as physical error rate, which needs to be reduced by magic state distillation.
To deal with the general error model in fault-tolerant quantum computation as in~\cite{gottesman2010,yamasaki2024time}, we assume that the $CCZ$-state preparation operations undergo the local stochastic error model, which allows correlated errors and is more general than the independent and identically distributed (IID) error model in the conventional error analyses of magic state distillation~\cite{bravyi2005universal,bravyi2012magic}.
In particular, we say that the preparation operations undergo the local stochastic error model if the following conditions are satisfied: (i) among the $\zeta$ locations of the $CCZ$-state preparation operations labelled $1,\ldots,\zeta$, a subset $F\subset\{1,\ldots,\zeta\}$ of faulty locations are randomly chosen with probability $p(F)$, and the $\zeta$ preparation operations to prepare noiseless $\ket{CCZ}^{\otimes \zeta}$ are followed by noise at the locations in $F$ represented by an arbitrary completely positive and trace-preserving (CPTP) map $\mathcal{E}_F$ acting on the $3\zeta$ qubits of $F$ (i.e., $\mathcal{E}_F$ may be correlated); (ii) each location $i\in\{1,\ldots,\zeta\}$ has a parameter $p_{i}$, and for any
set $S$ of locations, the probability $\Pr\{F\supset S\}$ of having faults at every location in $S$ (i.e., the probability of the randomly chosen set $F$ including $S$) is at most $\prod_{i\in S}p_{i}$. 
An upper bound of $p_{i}$ over all locations, i.e.,
\begin{align}
    p_\mathrm{ph}\coloneqq\max_i\{p_{i}\}, 
\end{align}
is called the physical error rate, so that we have
\begin{align}
\label{eq:physical_error_rate}
    \Pr\{F\supset S\}\leq p_\mathrm{ph}^{|S|},
\end{align}
where $|S|$ is the cardinality of the set $S$.
Note that in the IID error model, the initial states would be in the form $\rho^{\otimes \zeta}$
with $1-\bra{CCZ}\rho\ket{CCZ}\leq p_\mathrm{ph}$,
and then by performing twirling~\cite{bravyi2005universal}, we can convert each $\rho$ into 
$(1-p_\mathrm{ph})\ket{CCZ}\bra{CCZ}+p_\mathrm{ph}\rho_\perp$, where $\rho_\perp$ is a state in the orthogonal complement of $\ket{CCZ}$; in the same way, even if the errors are correlated, twirling transforms coherent errors into stochastic errors, which are in the scope of our error model.
With twirling, we can deal with the IID error model as a special case of the local stochastic error model at physical error rate $p_\mathrm{ph}$.

Within this error model, the task of magic state distillation is defined as follows.
Given any target error rate $\epsilon>0$ and $\zeta$ noisy $CCZ$ states (i.e. a mixed state $\rho_\zeta$ of $3\zeta$ qubits prepared by $\zeta$ $CCZ$-state preparation operations undergoing the local stochastic error model at physical error rate $p_\mathrm{ph}$), the task of magic state distillation is to transform $\rho_\zeta$ by stabiliser operations $\mathcal{E}_\zeta$ into a $3\zeta$-qubit state $\mathcal{E}_\zeta(\rho_\zeta)$ with its infidelity to the $\zeta$ copies of $CCZ$ states suppressed below $\epsilon$, i.e.,
\begin{align}
\label{eq:target_infidelty}
    1-\bra{CCZ}^{\otimes \xi(\zeta)}\left(\mathcal{E}_\zeta(\rho_\zeta)\right)\ket{CCZ}^{\otimes \xi(\zeta)}\leq\epsilon.
\end{align}
The overhead of a family $\{\mathcal{E}_\zeta\}_\zeta$ of protocols for magic state distillation is defined as the required number of initial magic states per final magic state at target error rate $\epsilon$, i.e.,
\begin{align}
    \inf\left\{\frac{\zeta}{\xi(\zeta)}:\text{A protocol $\mathcal{E}_\zeta$ achieves~\eqref{eq:target_infidelty} for some $\zeta$}\right\}.
\end{align}
In particular, the family of protocols achieves a polylogarithmic overhead if there is $\gamma>0$ such that
\begin{align}
    \inf\left\{\frac{\zeta}{\xi(\zeta)}:\text{A protocol $\mathcal{E}_\zeta$ achieves~\eqref{eq:target_infidelty} for some $\zeta$}\right\}=\mathcal{O}\left(\log^\gamma\left(\frac{1}{\epsilon}\right)\right)
\end{align}
as $\epsilon\to 0$.
The family of protocols achieves a constant overhead if
\begin{align}
    \inf\left\{\frac{\zeta}{\xi(\zeta)}:\text{A protocol $\mathcal{E}_\zeta$ achieves~\eqref{eq:target_infidelty} for some $\zeta$}\right\}=\mathcal{O}\left(1\right)
\end{align}
as $\epsilon\to 0$.
Note that $\ket{T}$ and $\ket{CCZ}$ may be exactly converted into each other at finite rates~\cite{beverland2020lower}; that is, the particular choice of the target magic state among these magic states does not affect the asymptotic scaling of overhead up to constant factors.

\section{Triorthogonal Codes on Prime-Power Qudits}\label{mainCodeSection}

This section is devoted to constructing and analysing the properties of the quantum code that will go on to be the central element of our magic state distillation protocols addressed in Section \ref{MSDProtocolMainSection}. These codes are so-called triorthogonal codes on $q$-dimensional qudits. Throughout, $q=p^s$ where $p=2$ and $s=10$, but we often write simply $q$ for brevity, and several of the results of this section will hold for other values of $s$, as we will indicate.

The triorthogonality framework for qubits was initially introduced by Bravyi and Haah in \cite{bravyi2012magic} and extended to prime-dimensional qudits by Krishna and Tillich in \cite{krishna2019towards}. The idea is to construct a matrix with certain nice properties, from which a quantum code may be built. As a result, the quantum code has certain nice properties that make it amenable to use for magic state distillation. The definition that we present for our triorthogonal matrix over $\mathbb{F}_q$ is one possible definition of a triorthogonal matrix over $\mathbb{F}_q$, but this definition is carefully chosen for our purposes.

\begin{definition}[Triorthogonal Matrix]\label{secondTriorthogDef}
    For us, a matrix $G \in \mathbb{F}_{q}^{m \times n}$, where $q=2^{s}$, with rows $(g^a)_{a=1}^m$, is called triorthogonal if for all $a,b,c \in \{1, ..., m\}$,
    \begin{equation}\label{triorthogonalityDef1}
        \sum_{i=1}^n(g^a_i)^4(g^b_i)^2(g^c_i) = \begin{cases}
            1 &\text{ if } 1 \leq a=b=c \leq k\\
            0 &\text{ otherwise}
        \end{cases}
    \end{equation}
    and
    \begin{equation}\label{triorthogonalityDef2}
        \sum_{i=1}^n \sigma_ig_i^ag_i^b = \begin{cases}
            \tau_a &\text{ if } 1 \leq a=b \leq k\\
            0 &\text{ otherwise}
        \end{cases}
    \end{equation}
    for some integer $k \in \{1, ..., m\}$ and for some $\sigma_i, \tau_a \in \mathbb{F}_{q}$, where $\sigma_i, \tau_a \neq 0$. In these expressions, all arithmetic takes place over $\mathbb{F}_{q}$.
\end{definition}

These properties are designed to ensure both that the quantum code we go on to define forms a valid quantum code on qudits (of dimension $q$), and also that it supports a non-Clifford transversal gate. This means that there is a one qubit non-Clifford gate $U$ such that acting on the $n$ physical qudits with $U^{\otimes n}$ executes the encoded gate $\overline{U^{\otimes k}}$ (on the $k$ encoded qudits). In Section \ref{nonCliffordGate}, we will analyse certain diagonal gates, one of which will go on to be our transversal gate $U$. In Section \ref{quantumCodefromMatrix}, we will show how to construct a triorthogonal quantum code from a triorthogonal matrix, and prove the properties of it that we require. This code will go on to be used to construct the magic state distillation protocols in Section \ref{MSDProtocolMainSection}. Note that the triorthogonal matrix that is used to build the quantum code will itself be developed in Section \ref{mainConstruction}.

As a final comment, we note that certain parts of the upcoming section are more original, and other parts are only minor generalisations of the triorthogonality formalism of \cite{bravyi2012magic} and \cite{krishna2019towards}. When the latter is true, we will comment on it.

\subsection{The Non-Clifford Transversal Gate}\label{nonCliffordGate}
 
Throughout, we primarily think of $q=p^s$ where $p=2$ and $s=10$, but the results of Section \ref{nonCliffordGate} in fact hold for $p=2$ and $s \geq 5$. Let us study gates acting on a single qudit of dimension $q$ of the following form.
\begin{equation}\label{initialgeneralGateDef}
    U_\beta^{(n)} \coloneq \sum_{\gamma \in \mathbb{F}_{q}}\exp\left[i\pi\tr(\beta\gamma^n)\right]\ket{\gamma}\bra{\gamma}\text{ where } \beta \in \mathbb{F}_{q} \text{ and } n \in \mathbb{Z}_{> 0}.
\end{equation}
Throughout, we will use the fact that $(U_\beta^{(n)})^\dagger = U_\beta^{(n)}$, and that $U_\beta^{(n)}$ is diagonal. The gate of interest to us will be $U_1^{(7)}$, and this section's analysis will prove the following.

\begin{lemma}\label{gatelemma}
    With $q=p^s$, where $p=2$ and $s \geq 5$, the gate $U_1^{(7)}$, as defined in Equation \eqref{initialgeneralGateDef}, is non-Clifford and in the third level of the Clifford hierarchy.
\end{lemma}

Naturally, one first notes that $U_\beta^{(1)} = Z^\beta$. Next, we see that
\begin{equation}
    U_\beta^{(2)} = \sum_{\gamma \in \mathbb{F}_{q}}\exp\left[i\pi\tr(\beta\gamma^2)\right]\ket{\gamma}\bra{\gamma} = \sum_{\gamma \in \mathbb{F}_{q}}\exp\left[i\pi\tr(\sqrt{\beta}\gamma)\right]\ket{\gamma}\bra{\gamma} = Z^{\sqrt{\beta}}
\end{equation}
where we have used Proposition \ref{traceProps} in the second equality. Also, we note that for all $\beta \in \mathbb{F}_{q}$, $\sqrt{\beta}$ is well-defined in $\mathbb{F}_{q}$\footnote{This is true for all finite fields $\mathbb{F}_{2^s}$, which may be seen from the well-known fact that the multiplicative group of $\mathbb{F}_q$ is exactly the cyclic group of order $q-1$. In fact, since $\beta^q = \beta$ for all $\beta \in \mathbb{F}_q$, we have $\sqrt{\beta} = \beta^{\frac{q}{2}}$.}. We have therefore shown that $U_\beta^{(1)}$ and $U_\beta^{(2)}$ are Paulis for all $\beta \in \mathbb{F}_{q}$. The same is not true for $n=3$. Consider
\begin{equation}
    U_\beta^{(3)} = \sum_{\gamma \in \mathbb{F}_{q}}\exp\left[i\pi\tr(\beta\gamma^3)\right]\ket{\gamma}\bra{\gamma}.
\end{equation}
Because this is a diagonal gate, it will certainly commute with $Z^\gamma$ for all $\gamma \in \mathbb{F}_{q}$. However, for $\eta \in \mathbb{F}_{q}$,
\begin{align}
    U_\beta^{(3)}X^\gamma U_\beta^{(3)}\ket{\eta} &= \exp\left[i\pi\tr(\beta\eta^3)\right]U_\beta^{(3)}X^\gamma\ket{\eta}\\
    &=\exp\left[i\pi\tr(\beta\eta^3)\right]U_\beta^{(3)}\ket{\eta+\gamma}\\
    &=\exp\left[i\pi\tr(\beta\eta^3)\right]\exp\left[i\pi\tr(\beta(\eta+\gamma)^3)\right]\ket{\eta+\gamma}\\
    &=\exp\left[i\pi\tr(\beta\eta^3 + \beta\eta^3 + \beta\eta^2\gamma + \beta\eta\gamma^2+\beta\gamma^3)\right]\ket{\eta+\gamma}\label{expansion3Clifford}\\
    &= \exp\left[i\pi\tr(\beta\gamma^3)\right]X^\gamma Z^{\sqrt{\beta\gamma}+\beta\gamma^2}\ket{\eta}.
\end{align}
Equation \eqref{expansion3Clifford} makes use of the fact that $\tr:\mathbb{F}_{q} \to \mathbb{F}_2$ is a linear map, and so we have $\exp\left[i\pi(\tr(x) + \tr(y))\right] = \exp\left[i\pi\tr(x+y)\right]$, because the prefactor $i\pi$ and the presence of the exponent means we may evaluate the addition $\tr(x)+\tr(y)$ modulo 2. Equation \eqref{expansion3Clifford} also makes use of the fact that $(\eta+\gamma)^3 = \eta^3 + \eta^2\gamma + \eta\gamma^2 + \gamma^3$ in finite fields $\mathbb{F}_{2^s}$. We have shown that
\begin{equation}
    U_\beta^{(3)}X^\gamma U_{\beta}^{(3)} = \exp\left[i\pi\tr(\beta\gamma^3)\right]X^\gamma Z^{\sqrt{\beta\gamma}+\beta\gamma^2}
\end{equation}
from which we deduce that $U_\beta^{(3)}$ is Clifford for all $\beta \in \mathbb{F}_{q}$ (again, because it is diagonal it will automatically commute with all $Z^\gamma$). We would also like to show that $U_\beta^{(3)}$ is not a Pauli for all $\beta \neq 0$. Indeed, if $U_\beta^{(3)}$ were a Pauli, then $U_\beta^{(3)}X^\gamma U_\beta^{(3)}$ would be equal to $\pm X^\gamma$ for all $\gamma \in \mathbb{F}_{q}$. Therefore, if $\sqrt{\beta\gamma} + \beta\gamma^2 \neq 0$ for some $\gamma \in \mathbb{F}_{q}$, then $U_\beta^{(3)}$ is not a Pauli. Suppose that, for some $\beta$, $\sqrt{\beta\gamma} = \beta\gamma^2$ for every $\gamma \in \mathbb{F}_{q}$. If $\beta \neq 0$, this implies $\beta = \gamma^{-3}$ for every $\gamma \in \mathbb{F}_{q}^*$\footnote{Here, as is often the case, $\mathbb{F}_q^*$ denotes the multiplicative group of $\mathbb{F}_q$, i.e. its non-zero elements.}, which is impossible as $\gamma^{-3}$ takes multiple values in $\mathbb{F}_q$. Thus, $U_\beta^{(3)}$ is a Clifford for every $\beta$, but not a Pauli, for every $\beta \neq 0$.

Next, we find
\begin{equation}
    U_\beta^{(4)} = \sum_{\gamma \in \mathbb{F}_{q}}\exp\left[i\pi\tr(\beta\gamma^4)\right]\ket{\gamma}\bra{\gamma} = \sum_{\gamma \in \mathbb{F}_{q}}\exp\left[i\pi\tr(\beta^{1/4}\gamma)\right]\ket{\gamma}\bra{\gamma} = Z^{\beta^{1/4}}
\end{equation}
again using Proposition \ref{traceProps} and the fact that fourth roots are again defined in $\mathbb{F}_{q}$. We then find that $U_\beta^{(5)}$ is, like $U_\beta^{(3)}$, a Clifford for all $\beta \in \mathbb{F}_{q}$, and not a Pauli when $\beta \neq 0$. Indeed,
\begin{align}
    U_\beta^{(5)}X^\gamma U_\beta^{(5)}\ket{\eta} &=\exp\left[i\pi\tr(\beta\eta^5)\right]\exp\left[i\pi\tr(\beta(\eta+\gamma)^5)\right]\ket{\eta+\gamma}\\
    &=\exp\left[i\pi\tr(\beta\eta^4\gamma + \beta\eta\gamma^4 + \beta\gamma^5\right]\ket{\eta+\gamma}\label{expansion5Clifford}\\
    &=\exp\left[i\pi\tr(\beta\gamma^5)\right]X^\gamma Z^{(\beta\gamma)^{1/4}+\beta\gamma^4}\ket{\eta}
\end{align}
where in Equation \eqref{expansion5Clifford} we have used that $(\eta+\gamma)^5 = \eta^5 + \eta^4\gamma + \eta\gamma^4 + \gamma^5$ in finite fields $\mathbb{F}_{2^s}$, because the binomial coefficients $\begin{pmatrix}5 \\ k \end{pmatrix}$ are odd for $k=0,1,4,5$ and even for $k=2,3$. By deducing that
\begin{equation}
    U_\beta^{(5)}X^\gamma U_\beta^{(5)} = \exp\left[i\pi\tr(\beta\gamma^5)\right]X^\gamma Z^{(\beta\gamma)^{1/4}+\beta\gamma^4},
\end{equation}
we again see that $U_\beta^{(5)}$ is Clifford for all $\beta \in \mathbb{F}_{q}$. Moreover, just as for the case $n=3$, if $(\beta\gamma)^{1/4} + \beta\gamma^4 \neq 0$ for some $\gamma \in \mathbb{F}_{q}$, then $U_\beta^{(5)}$ is not Pauli. Suppose that, for some $\beta$, $(\beta\gamma)^{1/4}+\beta\gamma^4 = 0$ for every $\gamma \in \mathbb{F}_{q}$. If $\beta \neq 0$, this implies that $\beta^3 = \gamma^{-15}$ for all $\gamma \in \mathbb{F}_{q}^*$. Because $\gamma^{-15}$ takes multiple values in $\mathbb{F}_{q}^*$\footnote{Again, this is seen by noting that the multiplicative group of $\mathbb{F}_{q}$ is a cyclic group of order $q-1$ where $q=2^s$ and $s \geq 5$.}, and $\beta$ is fixed, this is impossible. Thus, $U_\beta^{(5)}$ is a Clifford for all $\beta$, and not a Pauli for all $\beta \neq 0$.

The case of $n=6$ may be identified quickly as
\begin{equation}
    U_\beta^{(6)} = \sum_{\gamma \in \mathbb{F}_{q}}\exp\left[i\pi\tr(\beta\gamma^6)\right]\ket{\gamma}\bra{\gamma} = \sum_{\gamma \in \mathbb{F}_{q}}\exp\left[i\pi\tr(\sqrt{\beta}\gamma^3)\right]\ket{\gamma}\bra{\gamma} = U_{\sqrt{\beta}}^{(3)}
\end{equation}
which we know to be Clifford for all $\beta$ and non-Pauli for all $\beta \neq 0$.

Finally, we move to the most important case of $n=7$. This is
\begin{equation}
    U_\beta^{(7)} = \sum_{\gamma \in \mathbb{F}_{q}}\exp\left[i\pi\tr(\beta\gamma^7)\right]\ket{\gamma}\bra{\gamma}.
\end{equation}
Here, we have
\begin{align}
    U_\beta^{(7)}X^\gamma U_\beta^{(7)}\ket{\eta} &= \exp\left[i\pi\tr(\beta\eta^7)\right]\exp\left[i\pi\tr(\beta(\eta+\gamma)^7)\right]\ket{\eta+\gamma}\\
    &= \exp\left[i\pi\tr(\beta\eta^6\gamma + \beta\eta^5\gamma^2 + \beta\eta^4\gamma^3 + \beta\eta^3\gamma^4 + \beta\eta^2\gamma^5 + \beta\eta\gamma^6 + \beta\gamma^7\right]\ket{\eta+\gamma}\\
    &=\exp\left[i\pi\tr(\beta\gamma^7)\right]X^\gamma U^{(3)}_{\sqrt{\beta\gamma} + \beta\gamma^4}U^{(5)}_{\beta\gamma^2}Z^{\beta^{1/4}\gamma^{3/4} + \beta^{1/2}\gamma^{5/2} + \beta\gamma^6}\ket{\eta}
\end{align}
from which we deduce
\begin{equation}
    U_\beta^{(7)}X^\gamma U_\beta^{(7)} = \exp\left[i\pi\tr(\beta\gamma^7)\right]X^\gamma U^{(3)}_{\sqrt{\beta\gamma} + \beta\gamma^4}U^{(5)}_{\beta\gamma^2}Z^{\beta^{1/4}\gamma^{3/4}+\beta^{1/2}\gamma^{5/2} + \beta\gamma^6}.
\end{equation}
Most interesting to us is the case $\beta = 1$, where we have
\begin{equation}
    U_1^{(7)}X^\gamma U_1^{(7)} = \exp\left[i\pi\tr(\gamma^7)\right]X^\gamma U^{(3)}_{\sqrt{\gamma}+\gamma^4}U^{(5)}_{\gamma^2}Z^{\gamma^{3/4}+\gamma^{5/2}+\gamma^6}.
\end{equation}
We therefore find that $U_1^{(7)}$ is in the third level of the Clifford hierarchy, since we know the right-hand side is Clifford, and we know that $U_1^{(7)}$ will commute with $Z^\gamma$ for every $\gamma$ since it is diagonal. We would like to prove further that $U_1^{(7)}$ is non-Clifford, i.e. it is exactly in the third level of the Clifford hierarchy. For this, it suffices to consider $\gamma = 1$, for which
\begin{equation}
    U^{(7)}_1X^1U_1^{(7)} = X^1U^{(5)}_1Z^1.
\end{equation}
We have shown that $U_1^{(5)}$ is not a Pauli, and thus we find that $U_1^{(7)}$ is not Clifford, thus establishing Lemma \ref{gatelemma}.

\subsection{The Triorthogonal Code from the Triorthogonal Matrix}\label{quantumCodefromMatrix}

In this section, we will construct the quantum code from which our magic state distillation protocol will be built in Section \ref{MSDfromCode}, and we will prove the properties of it which we shall need. In turn, the quantum code built in this section will be constructed from the triorthogonal matrix constructed in Section \ref{mainConstruction}. Our definition of a triorthogonal matrix will be as in Definition \ref{secondTriorthogDef}. This definition is one example of a possible generalisation of the definition of triorthogonal matrices made for binary matrices in \cite{bravyi2012magic} and the definition made for matrices over $\mathbb{F}_p$ in \cite{krishna2019towards}, where $p$ is a prime. In all cases, there are polynomial expressions of the rows analogous to Equations \eqref{triorthogonalityDef1} and \eqref{triorthogonalityDef2}, and these equations are designed to ensure that the quantum code we construct is a valid code and supports a non-Clifford transversal gate.

We adopt similar notation to \cite{bravyi2012magic}, as we now detail. Given the triorthogonal matrix $G \in \mathbb{F}_q^{m \times n}$, the submatrix formed by its first $k$ rows is denoted $G_1$, whereas the submatrix formed by the latter $m-k$ rows is denoted $G_0$. Then, we denote the vector space over $\mathbb{F}_q$ generated by the rows of $G$ as $\mathcal{G}\subseteq \mathbb{F}_{q}^n$, the vector space generated by the rows of $G_1$ as $\mathcal{G}_1$, and the vector space generated by the rows of $G_0$ as $\mathcal{G}_0$. We emphasise that these spaces are all $\mathbb{F}_{q}$-linear vector spaces.

Now, the quantum code that we need for our construction is $CSS(X,\mathcal{G}_0;Z,\mathcal{G}^\perp)$, where
\begin{equation}
    \mathcal{G}^\perp = \{y \in \mathbb{F}_{q}^n: \sum_{i=1}^nx_iy_i = 0 \text{ for all } x \in \mathcal{G}\}
\end{equation}
and the arithmetic in this expression takes place over $\mathbb{F}_{q}$. This is a well-defined CSS code, since $\mathcal{G}_0 \subseteq \mathcal{G} = (\mathcal{G}^\perp)^\perp$. We note that the definition of the quantum code $CSS(X,\mathcal{G}_0;Z,\mathcal{G}^\perp)$ is exactly that which is employed in both \cite{bravyi2012magic} and \cite{krishna2019towards}, although note this comes with a different definition of the triorthogonal matrix. In particular, the presence of the elements $\sigma_i$ in Equation \eqref{triorthogonalityDef2} is the main cause of complication in proving the quantum code's properties relative to the proofs of \cite{bravyi2012magic} and \cite{krishna2019towards}.

To start proving these properties, the following will be useful.

\begin{lemma}\label{matrixProps}The following hold for the triorthogonal matrix $G$.
    \begin{enumerate}
        \item The $k$ rows of $G_1$ are linearly independent over $\mathbb{F}_{q}$.
        \item $\mathcal{G}_0\cap\mathcal{G}_1 = \{0\}$.
    \end{enumerate}
\end{lemma}
Note that this lemma is the analogue of the first two points of Lemma 1 in \cite{bravyi2012magic}, and the proof is only slightly different due to the presence of $\sigma_i$.
\begin{proof}[Proof of Lemma \ref{matrixProps}]
For the first point, consider elements $\mu_a \in \mathbb{F}_{q}$ and an expression
\begin{equation}
    \sum_{a=1}^k\mu_ag^a = 0
\end{equation}
with arithmetic over $\mathbb{F}_{q}$. This means that $\sum_{a=1}^k\mu_ag^a_i=0$ for each $i=1, ..., n$, and multiplying by $\sigma_ig_i^b$ leads us to
\begin{equation}
    \sum_{a=1}^k\mu_a\sigma_ig_i^ag_i^b = 0
\end{equation}
for every $i=1, ..., n$ and $b=1, ..., k$. Summing over $i \in \{1, ..., n\}$ and employing Equation \eqref{triorthogonalityDef2} leads us to $\mu_b\tau_b = 0$ for each $b=1, ..., k$ $\implies \mu_b = 0$ for each $b=1, ..., k$.

The second point of the lemma is proved very similarly. If we consider an element in $x \in \mathcal{G}_0\cap\mathcal{G}_1$, it may be written as
\begin{equation}
    x = \sum_{a=1}^k\mu_ag^a = \sum_{a=k+1}^m\mu_ag^a
\end{equation}
for some $\mu_a \in \mathbb{F}_{q}$. Via the exact same steps - taking components, multiplying by $\sigma_ig_i^b$ for any $b =1, ..., k$, summing over $i \in \{1, ..., n\}$ and using Equation \eqref{triorthogonalityDef2}, one arrives at $\mu_a = 0$ for all $a=1, ..., k \implies x=0$.
\end{proof}
Just as it is done in \cite{bravyi2012magic}, Lemma \ref{matrixProps} allows us to determine the number of logical qudits encoded by the quantum code in question. From Equation \eqref{numEncodedQudits}, this is
\begin{equation}
    n-\dim\mathcal{G}_0 - \dim\mathcal{G}^\perp,
\end{equation}
where $\dim$ refers to the dimension of an $\mathbb{F}_{q}$-vector space over $\mathbb{F}_{q}$. We have 
\begin{equation}\label{GperpDim}
    \dim\mathcal{G}^\perp = n-\dim\mathcal{G} = n- (k + \dim\mathcal{G}_0),
\end{equation}
where the second equality uses both points in Lemma \ref{matrixProps}. The quantum code at hand thus encodes exactly $k$ logical qudits (of dimension $q$).

Turning to the distance of the quantum code, we use the standard formula for the distance of a CSS code
\begin{equation}
    d = \min\{d_X,d_Z\} \text{  where  } d_X = \min_{f \in \mathcal{G}\setminus \mathcal{G}_0}|f|\text{ 
 and  } d_Z = \min_{f \in \mathcal{G}_0^\perp \setminus \mathcal{G}^\perp}|f|.
\end{equation}
Here, the Hamming weight is defined as $|f| = |\{i:f_i \neq 0\}|$ for $f \in \mathbb{F}_{q}^n$. Note that here we are counting the distance of the quantum code in terms of Paulis on $q$-dimensional qudits. It turns out that only the $Z$-distance, $d_Z$, is relevant for the magic state distillation constructed out of the quantum code, although for completeness it is nice to have a complete picture of the quantum code's distance (and in fact this is necessary to establish Theorem \ref{quantumCodeTheorem}). Indeed, in both \cite{bravyi2012magic} and \cite{krishna2019towards}, it turns out that $d_X \geq d_Z$, so that in fact $d_Z$ is equal to the distance of the quantum code $d$. This will be true in our case as well, although the proof will be a little more involved. 

\begin{claim}\label{distanceClaim}
    Our quantum code satisfies $d_X \geq d_Z$, so in particular $d_Z$ equals the distance of the code.
\end{claim}
In this direction, consider the following.

\begin{definition}
    Given the elements $(\sigma_i)_{i=1}^n$ from Equation \eqref{triorthogonalityDef2}, we define the map $m_\sigma : \mathbb{F}_{q}^n \to \mathbb{F}_{q}^n$ as
    \begin{equation}
        m_\sigma(v)_i = \sigma_iv_i
    \end{equation}
    i.e. $m_\sigma$ acts as componentwise multiplication by the $\sigma_i$. For any subset $C \subseteq \mathbb{F}_{q}^n$, $m_\sigma(C) \subseteq \mathbb{F}_{q}$ is defined as the set of images of members of $C$ under $m_\sigma$.
\end{definition}
Because $\sigma_i \neq 0$, $m_\sigma$ is invertible. In fact, it is easily seen that for any linear subspace $C \subseteq \mathbb{F}_{q}^n$, $m_\sigma(C)$ forms a linear subspace, and $m_\sigma : C \to m_\sigma(C)$ is an isomorphism of vector spaces. We further have that $|m_\sigma(v)| = |v|$ for all $v \in \mathbb{F}_{q}^n$ and, for this reason, if we consider any set of vectors $X \subseteq \mathbb{F}_{q}^n$, one has that $\min\{|f|: f \in X\} = \min\{|f|: f \in m_\sigma(X)\}$.

We may now rewrite $d_X$ and $d_Z$ into useful forms. Given the observation at the end of the previous paragraph, we have
\begin{equation}
    d_X = \min_{f \in m_\sigma(\mathcal{G}\setminus \mathcal{G}_0)}|f|.
\end{equation}
In addition, because $\mathcal{G} = \mathcal{G}_1 + \mathcal{G}_0$ and $\mathcal{G}_1 \cap \mathcal{G}_0 = \{0\}$ from Lemma \ref{matrixProps}, we have
\begin{equation}
    \mathcal{G}\setminus\mathcal{G}_0 = \{x+y: x \in \mathcal{G}_1\setminus\{0\}, y \in \mathcal{G}_0\}.
\end{equation}
Therefore,
\begin{align}
    m_\sigma(\mathcal{G}\setminus\mathcal{G}_0) &= \{m_\sigma(x) + m_\sigma(y): x \in \mathcal{G}_1\setminus\{0\}, y \in \mathcal{G}_0\}\\
    &= \{x+y: x \in m_\sigma(\mathcal{G}_1)\setminus\{0\}, y \in m_\sigma(\mathcal{G}_0)\}.\label{mwgg0}
\end{align}
We then treat $d_Z$, starting with the following claim.
\begin{claim}\label{g0pClaim}
    $\mathcal{G}_0^\perp = m_\sigma(\mathcal{G}_1)+\mathcal{G}^\perp$.
\end{claim}
Note that this is a generalisation of the fourth point of Lemma 1 of \cite{bravyi2012magic}.
\begin{proof}
    Equation \eqref{triorthogonalityDef2} implies that $m_\sigma(\mathcal{G}_1) \subseteq \mathcal{G}_0^\perp$. By definition, $\mathcal{G}^\perp \subseteq \mathcal{G}_0^\perp$, and so $m_\sigma(\mathcal{G}_1) + \mathcal{G}^\perp \subseteq \mathcal{G}_0^\perp$. A further application of Equation \eqref{triorthogonalityDef2} gives us that $m_\sigma(\mathcal{G}_1) \cap \mathcal{G}^\perp = \{0\}$. Therefore, we have
    \begin{alignat}{2}
        \dim(m_\sigma(\mathcal{G}_1)+\mathcal{G}^\perp) &= \dim m_\sigma(\mathcal{G}_1) + \dim \mathcal{G}^\perp \\
        &=\dim\mathcal{G}_1 + \dim\mathcal{G}^\perp\hspace{1cm}&&\text{ (since } m_\sigma : \mathcal{G}_1 \to m_\sigma(\mathcal{G}_1) \text{ is an isomorphism)}\\
        &= k + (n-k-\dim\mathcal{G}_0)\hspace{1cm}&&\text{ (using Equation \eqref{GperpDim})}\\
        &= \dim\mathcal{G}_0^\perp
    \end{alignat}
    which establishes the claim.
\end{proof}
Claim \ref{g0pClaim}, and the fact that $m_\sigma(\mathcal{G}_1)\cap\mathcal{G}^\perp = \{0\}$, allow us to write
\begin{equation}
    \mathcal{G}_0^\perp\setminus\mathcal{G}^\perp = \{x+y:x \in m_\sigma(\mathcal{G}_1)\setminus\{0\}, y \in \mathcal{G}^\perp\}.\label{g0pgp}
\end{equation}
Finally, because $m_\sigma(\mathcal{G}_0)\subseteq \mathcal{G}^\perp$, we have $m_\sigma(\mathcal{G}\setminus\mathcal{G}_0) \subseteq \mathcal{G}_0^\perp\setminus\mathcal{G}^\perp$ from Equations \eqref{mwgg0} and \eqref{g0pgp}, which implies $d_X \geq d_Z$, as claimed. In Section \ref{mainConstruction}, a lower bound will be proved on $d_Z$, and this, therefore, gives us a lower bound on the distance of the quantum code.

Now that we have handled the distance of the code, we would like to move towards demonstrating that the code supports a non-Clifford transversal gate. We will show the following.
\begin{lemma}\label{transversalGateLemma}
    With $q=p^s$, where $p=2$ and $s$ is any positive integer, we consider a triorthogonal matrix $G \in \mathbb{F}_q^{m \times n}$ as in Definition \ref{secondTriorthogDef} and the quantum code $CSS(X,\mathcal{G}_0;Z,\mathcal{G}^\perp)$ defined above. Writing $U=U_1^{(7)}$ as defined in Equation \eqref{initialgeneralGateDef}, $U^{\otimes n}$ and the encoded gate $\overline{U^{\otimes k}}$ act in the same way on code states, i.e. the quantum code supports the gate $U$ transversally.
\end{lemma}
To prove this, it is necessary for us to identify the natural choice of logical Pauli operators on the encoded qudits; this is again a small generalisation of \cite{bravyi2012magic} and \cite{krishna2019towards}.

\begin{claim}
    For $a=1 ..., k$, $\overline{X^\beta_a}$ and $\overline{Z^\gamma_a}$, which are respectively the encoded operators $X^\beta$ and $Z^\gamma$ acting on the $a$-th logical qudit, may be chosen to be
    \begin{align}
        \overline{X^\beta_a} &= X^{\beta g^a}\\
        \overline{Z^\gamma_a} &= Z^{\gamma \hat{g}^a}
    \end{align}
    where $g^a$ is the $a$-th row of the matrix $G$ and $\hat{g}^a$ is related to $g^a$ by
    \begin{equation}
        \hat{g}^a_i = \tau_a^{-1}\sigma_ig_i^a.
    \end{equation}
\end{claim}
We emphasise that again we have used the notation $X^v = \bigotimes_{i=1}^nX^{v_i}$, where $v \in \mathbb{F}_{q}^n$, and similarly for $Z$.
\begin{proof}
    $X$-stabilisers are of the form $X^v$ for $v \in \mathcal{G}_0$. Such a $v$ satisfies $\sum_{i=1}^nv_i\hat{g}^a_i = 0$ for all $a=1, ..., k$ by Equation \eqref{triorthogonalityDef2} - therefore $Z^{\gamma \hat{g}^a}$ commutes with the $X$-stabilisers for all $a$ and $\gamma$. Meanwhile, $Z$-stabilisers are of the form $Z^w$ for $w \in \mathcal{G}^\perp$, which makes $X^{\beta g^a}$ commute with all $Z$-stabilisers. Also, for $a,b = 1, ..., k$, we have $\sum_{i=1}^ng^a_i\hat{g}^b_i = \delta_{ab}$, again by Equation \eqref{triorthogonalityDef2}, which gives us
    \begin{equation}
        X^{\beta g^a}Z^{\gamma \hat{g}^b} = (-1)^{\tr(\beta\gamma\sum_{i=1}^n(g^a)_i(\hat{g}^b)_i)}Z^{\gamma \hat{g}^b}X^{\beta g^a} = (-1)^{\tr(\beta\gamma)\delta_{ab}}Z^{\gamma \hat{g}^b}X^{\beta g^a}.
    \end{equation}
    This proves the expected commutation relations for the logical operators:
    \begin{equation}
        \overline{X^\beta_a}\overline{Z^\gamma_b} = (-1)^{\tr(\beta\gamma)\delta_{ab}}\overline{Z^\gamma_b}\overline{X^\beta_a},
    \end{equation}
    thus establishing the claim.
\end{proof}
We are now in a position to be able to prove the transversality of the non-Clifford logical gate $U_1^{(7)}$ presented in the previous section. This is where this section departs most dramatically from \cite{bravyi2012magic} and \cite{krishna2019towards}. Indeed, let us start by considering the encoded all-zeros state
\begin{equation}
    \overline{\ket{0}^{\otimes k}} = \sum_{g \in \mathcal{G}_0}\ket{g}
\end{equation}
which we have left unnormalised. To see that this is the encoded all-zeros state, one may refer to the usual expression for logical computational basis states for a CSS code (see Chapter 8 of \cite{gottesman2016surviving}), or observe that this state is fixed by all stabilisers and all $\overline{Z^\gamma_a}$. Given $u \in \mathbb{F}_{q}^k$, one may determine the encoded computational basis state $\overline{\ket{u}}$ by acting with the appropriate $X$-logical operators, obtaining
\begin{equation}
    \overline{\ket{u}} = \sum_{g \in \mathcal{G}_0}\Ket{\sum_{a=1}^ku_ag^a + g}
\end{equation}
where all arithmetic in the ket takes place over $\mathbb{F}_{q}$. Abbreviating $U_1^{(7)}$ as $U$, what we would now like to do is show that acting on any state $\overline{\ket{u}}$ with $U$ transversally, i.e. acting with $U^{\otimes n}$, has the same effect as acting on the encoded qudits transversally with $U$, i.e. acting with the encoded gate $\overline{U^{\otimes k}}$. To this end, consider any element $f \in \mathcal{G}$. Write it as $f = \sum_{a=1}^mu_ag^a$ for some $u_a \in \mathbb{F}_{q}$, and where all arithmetic here takes place over $\mathbb{F}_{q}$. Let us consider
\begin{equation}
    U^{\otimes n}\ket{f} = \prod_{i=1}^n\exp\left[i\pi\tr(f_i^7)\right]\ket{f}.
\end{equation}
Now we have $f_i = \sum_{a=1}^mu_ag^a_i$, and so $f_i^7 = \left(\sum_{a=1}^mu_ag^a_i\right)^7$, which we may expand as
\begin{multline}
    f_i^7 = \sum_{a=1}^mu_a^7(g_i^a)^7 + \sum_{a \neq b} \left[u_a^6u_b(g_i^a)^6(g_i^b)+u_a^5u_b^2(g_i^a)^5(g_i^b)^2 + u_a^4u_b^3(g_i^a)^4(g_i^b)^3\right] + \\\sum_{\substack{a,b,c\\\text{pairwise distinct}}} u_a^4u_b^2u_c(g_i^a)^4(g_i^b)^2(g_i^c).\label{7thmultinomial}
\end{multline}
Note that usually one would expect more terms than this in Equation \eqref{7thmultinomial}. However, arithmetic here takes place over $\mathbb{F}_{q}$, so in particular addition is taken modulo 2, and one can see that the formula written is true by evaluating all possible multinomial coefficients and noting that the ones that are odd correspond exactly to the terms written. For completeness, we also include a derivation of this formula in Appendix \ref{multinomialDerivation}.

Using the triorthogonality property of our matrix given in Equation \eqref{triorthogonalityDef1}, it is then true that
\begin{equation}
    \sum_{i=1}^nf_i^7 = \sum_{a=1}^ku_a^7.
\end{equation}
Indeed, of the 5 terms on the right-hand side of Equation \eqref{7thmultinomial}, the second, third and fourth vanish by setting $a=b$, $a=c$ and $b=c$ respectively in Equation \eqref{triorthogonalityDef1}. The final term naturally vanishes with $a,b,c$ pairwise distinct in Equation \eqref{triorthogonalityDef1}. The first term is treated by setting $a=b=c$ in Equation \eqref{triorthogonalityDef1}.

With this, one finds that
\begin{align}
    U^{\otimes n}\ket{f} &= \exp\left[i\pi \sum_{i=1}^n\tr(f_i^7)\right]\ket{f}\label{modulo2SumTransversality}\\
    &=\exp\left[i\pi\tr\left(\sum_{i=1}^nf_i^7\right)\right]\ket{f}\\
    &=\exp\left[i\pi\tr\left(\sum_{a=1}^ku_a^7\right)\right]\ket{f}
\end{align}
where we emphasise that the sum can be taken inside the trace because the sum in Equation \eqref{modulo2SumTransversality} may be evaluated modulo 2, due to the prefactor $i\pi$ and the presence of $\exp$. One then calculates the desired expression as follows.
\begin{align}
    U^{\otimes n}\overline{\ket{u}} &= \sum_{g \in \mathcal{G}_0}U^{\otimes n}\Ket{\sum_{a=1}^ku_ag^a+g}\\
    &=\sum_{g \in \mathcal{G}_0}\exp\left[i\pi\tr\left(\sum_{a=1}^ku_a^7\right)\right]\Ket{\sum_{a=1}^ku_ag^a+g}\\
    &=\exp\left[i\pi\tr\left(\sum_{a=1}^ku_a^7\right)\right]\overline{\ket{u}}\\
    &=\overline{U^{\otimes k}}\;\overline{\ket{u}}.
\end{align}
Thus, as desired, it is shown that on code states, $U^{\otimes n}$ and the encoded gate $\overline{U^{\otimes k}}$ act in the same way, which establishes Lemma \ref{transversalGateLemma}.

\section{Magic State Distillation Protocols}\label{MSDProtocolMainSection}

In this section, we will present our magic state distillation protocols. The core of the construction is a subroutine for distilling a one-qudit magic state, for qudits of dimension $q=2^{10}$, but we show how this naturally corresponds to the distillation of a $10$-qubit magic state. In turn, this will enable distillation of the magic state for a $CCZ$ gate on qubits. The magic state distillation protocols are built from the quantum code for $q$-dimensional qudits constructed in Section \ref{mainCodeSection}, which in turn are built out of matrices constructed in Section \ref{mainConstruction}. The analysis of the protocols will take place in Section \ref{errorAnalysis}, where, in particular, we will show how constant-overhead magic state distillation is achieved.

We will start in Section \ref{MSDSummarySection} with a summary of our magic state distillation protocols. We then elaborate briefly on the importance of the twirling stage of the protocol in Section \ref{twirling}, before moving on to Section \ref{MSDfromCode}. In Section \ref{MSDfromCode}, we provide more detail on the bulk of our protocol, which is the distillation of the one-qudit magic state, and we go into detail on the translation to the distillation of a $10$-qubit magic state. Finally, in Section \ref{CCZConversion}, we go into detail on the conversion of the distillation of the $10$-qubit magic state to a more standard qubit magic state (the $\ket{CCZ}$ state which enables the execution of the qubit CCZ gate).

\subsection{Summary of Magic State Distillation Protocols}\label{MSDSummarySection}

To describe the protocol, we introduce a little notation based on our triorthogonal matrix. Let us write, as usual, $q=2^s$ where in general we think of $s=10$ (although some of our results hold for various other $s$). Recalling that $G_0$ denotes the submatrix formed by the latter $m-k$ rows of our triorthogonal matrix, and $\mathcal{G}_0$ denotes the linear subspace of $\mathbb{F}_q^n$ generated (over $\mathbb{F}_q$) by the rows of $G_0$, denote $d_{\mathrm{cl}}$ for the distance of the classical code $\mathcal{G}_0^\perp \subseteq \mathbb{F}_{q}^n$\footnote{It will go on to be shown in Section \ref{puncturingSection} that $d_{\mathrm{cl}}$ is linear, $d_{\mathrm{cl}} = \Theta(n)$. We show this because we want to show that the distance of the quantum code, $d_Z = \min\{|f|: f \in \mathcal{G}_0^\perp \setminus \mathcal{G}^\perp\}$ is $\Theta(n)$, but we in fact achieve this by using $d_Z \geq d_{\mathrm{cl}}$ and showing that $d_{\mathrm{cl}} = \Theta(n)$.},  which we note is the code with parity-check matrix $G_0$.

Further, consider a classical decoder for the code $\mathcal{G}_0^\perp$ with decoding radius $t = \Theta(d_{\mathrm{cl}})$. Then $t= \Theta(n)$ and $2t+1 \leq d_{\mathrm{cl}}$. Explicitly, suppose that some codeword $v$ of $\mathcal{G}_0^\perp$, is affected by an error $e \in \mathbb{F}_q^n$, where $|e| \leq t$, leaving $v+e$. The decoder is given the syndrome for the error, which is $G_0e$. The decoder will deterministically return the error $e$. If the decoder is given the syndrome of some error $e$ with $|e| > t$, it will return some vector $\tilde{e}$ that is not necessarily equal to $e$. Note that, under the assumption of instantaneous classical computation, an inefficient, exhaustive decoder achieving $2t+1 = d_{\mathrm{cl}}$ would suffice. Nevertheless, we do not require such an assumption since the code $\mathcal{G}_0^\perp$ may be expressed as the dual of an algebraic geometry code, for which efficient decoders have been extensively studied\footnote{The efficient classical decoding of $\mathcal{G}_0^\perp$ will be discussed further in Section \ref{concrete}.}.

Now, the main aim of our magic state distillation protocols is to distil the three-qubit state
\begin{equation}
    \ket{CCZ} = CCZ\left(\ket{+}\right)^{\otimes 3}.
\end{equation}
Via a standard gate teleportation procedure (see Section \ref{CCZConversion}), one copy of such a state may be consumed, using Clifford operations only, to execute a CCZ gate on an unknown three-qubit state, thus enabling universal quantum computation when combined with Clifford operations (note that CCZ is a non-Clifford gate, and quite a standard one at that). It is, therefore, desirable to be able to distil $\ket{CCZ}$ states, i.e. to generate a smaller number of lower error rate copies of them from a larger number of higher error rate copies.
As in the conventional protocols~\cite{bravyi2005universal,bravyi2012magic}, we apply a dephasing transformation to $\ket{CCZ}$ at the beginning of magic state distillation, which is also known as twirling.
The twirling for $\ket{CCZ}$ is performed by applying uniformly random combinations of all the stabilisers of $\ket{CCZ}$, which are exactly the elements of the group generated by $X_1\otimes CZ_{2,3}$, $X_2\otimes CZ_{1,3}$, and $X_3\otimes CZ_{1,2}$, where $X_i$ is an $X$ gate acting on the $i$th qubit, and $CZ_{i,j}$ is a $CZ$ gate acting on the pair $(i,j)$ of $i$th and $j$th qubits ($i,j\in\{1,2,3\}$).
The twirling map may be compactly written as
\begin{align}
\label{eq:twirling}
    \rho\to\frac{1}{8}\sum_{\mathbf{d} \in \mathbb{F}_2^3}S^{(\mathbf{d})}\rho S^{(\mathbf{d})}
\end{align}
where
\begin{equation}
    S^{(\mathbf{d})} \coloneq CCZ \; X^{\mathbf{d}} \; CCZ
\end{equation}
are the stabilisers of $\ket{CCZ}$. Because the $CCZ$ gate is in the third level of the Clifford hierarchy, the twirling map is implementable by stabiliser operations.
As we will see in Section~\ref{twirling}, the $CCZ$ states after twirling only have $Z$ errors.

The bulk of our protocol will not, however, be distilling the $\ket{CCZ}$ states; it will be distilling an $s$-qubit magic state, as described in Section~\ref{MSDfromCode} in detail.
We let $\ket{U_s}$ denote this $s$-qubit magic state to be distilled.
This may equally well be thought of as a one-qudit magic state, for a qudit of dimension $q$. To differentiate these, we denote the one-qudit state $\ket{M}$, but we note that $\ket{U_s}$ and $\ket{M}$ have the same coefficients in their computational bases --- the only difference is in different labelling of the computational basis states. For the qudit state $\ket{M}$, computational basis states are labelled by the elements of the field $\mathbb{F}_{q}$, and for the $s$-qubit state $\ket{U_{s}}$, computational basis states are labelled by bit strings of length $s$, i.e. elements of $\mathbb{F}_2^s$. The translation between elements of $\mathbb{F}_{q}$ and $\mathbb{F}_2^s$ goes via decomposition in a self-dual basis, as discussed in Section \ref{primePowerQuditsPrelims}. Where $U = U_1^{(7)}$ is the non-Clifford, diagonal gate in the third level of the Clifford hierarchy identified in Section \ref{nonCliffordGate},
\begin{equation}
    U = \sum_{\gamma \in \mathbb{F}_{q}}\exp\left[i\pi\tr(\gamma^7)\right]\ket{\gamma}\bra{\gamma},
\end{equation}
the state $\ket{M}$ is then
\begin{equation}
    \ket{M} = U\ket{+_q},
\end{equation}
where $\ket{+_q}$ is the one-qudit state satisfying $X^\gamma \ket{+_q} = \ket{+_q}$ for all $\gamma \in \mathbb{F}_{q}$:
\begin{equation}
    \ket{+_q} = \frac{1}{\sqrt{q}}\sum_{\gamma \in \mathbb{F}_{q}}\ket{\gamma}.
\end{equation}
This one-qudit gate $U$ is equivalent to the $s$-qubit gate $\theta(U)$, where $\theta$ is the map introduced in Section \ref{primePowerQuditsPrelims}, which does nothing other than to relabel computational basis states to change a one-qudit gate into an $s$-qubit gate. Imagined as matrices in terms of their computational bases, $U$ and $\theta(U)$ are exactly the same, and in fact we have
\begin{equation}
    \ket{U_s} = \theta(U)\left(\ket{+}^{\otimes s}\right).
\end{equation} 
We will go on to describe in Section \ref{CCZConversion} how the $s$-qubit gate $\theta(U)$ may be decomposed into a sequence of $Z$, $CZ$ and $CCZ$ gates, where the number of $CCZ$ gates used is some constant $C$. By using the gate teleportation procedure for the $CCZ$ gate described therein, $C$ copies of the state $\ket{CCZ}$ may therefore be used to generate one copy of the state $\ket{U_s}$ using only Clifford operations. It will also be described in Section \ref{CCZConversion} how the reverse conversion may be performed: how one (noisy) $\ket{U_s}$ state may be used to make one (noisy) $\ket{CCZ}$ state, with only Clifford operations.

Now, our protocol for distillation of $\ket{CCZ}$ states with constant overhead is summarised as follows.

\begin{enumerate}
    \item Perform twirling on each of the $Cn$ input $\ket{CCZ}$ states.
    \item Use the $Cn$ twirled $\ket{CCZ}$ states of higher error rate to make a supply of $n$ copies of the $s$-qubit $\ket{U_s}$ states (of higher error rate), which may be interpreted as single-qudit $\ket{M}$ states.
    \item Perform one round of the magic state distillation subroutine on the $\ket{M}$ states, as follows.
    \begin{enumerate}
        \item Prepare $k$ qudits in the state $\ket{+_q}^{\otimes k}$ and encode them into our quantum code.
        \item Using gate teleportation, apply (a noisy version of) $U^{\otimes n}$ to the $n$ physical qudits of the code by using Clifford operations, and by consuming the $n$ $\ket{M}$ states.
        \item Measure the $X$-stabilisers of the code, obtaining a syndrome $S \in \mathbb{F}_{q}^{m-k}$. Pass the syndrome to the classical decoder.
        \item If the decoder returns the error vector $v_S \in \mathbb{F}_{q}^n$ (for error correction, without post-selection), apply $Z^{v_S}$ to the state.
        \item Decode the state from the quantum code. The remaining $k$ qudits are treated as $k$ copies of the $\ket{M}$ state, of higher quality.
    \end{enumerate}
    \item The outputted $k$ $\ket{M}$ states, each of which is a state of one qudit, may equally well be thought of as $k$ copies of the $s$-qubit state $\ket{U_s}$.
    \item Convert the $k$ $\ket{U_s}$ states into $k$ $\ket{CCZ}$ states.
\end{enumerate}

Readers familiar with \cite{bravyi2012magic} and \cite{krishna2019towards} will note the similarity of step 3 to those procedures. Those works construct $[[n,k,d]]$ quantum codes (on qubits in the first case and on prime-dimensional qudits in the second case), which support transversal gates that are diagonal, non-Clifford and in the third level of the Clifford hierarchy. Calling this non-Clifford gate $U$ for convenience, they each make use of a distillation subroutine, which goes as follows, on either qubits or prime-dimensional qudits.
\begin{enumerate}
    \item Prepare $k$ qubits/qudits in the state $\ket{+}^{\otimes k}$, where $\ket{+}$ is an equal superposition of all computational basis states. Encode these qubits/qudits into the quantum code.
    \item Using a supply of $n$ copies of the state $\ket{M} = U\ket{+}$ (of higher error rate), and gate teleportation, apply (a noisy) $U^{\otimes n}$ to the state of the quantum code.
    \item Measure the $X$-stabilisers of the code.
    \item If any non-trivial syndrome is obtained, declare a failure of the protocol. If a trivial syndrome is obtained, decode the state from the quantum code. The result is treated as $k$ copies of the state $\ket{M}$ of a lower error rate than the initial supply.
\end{enumerate}
Then, \cite{bravyi2012magic} and \cite{krishna2019towards} concatenate this distillation subroutine to make a larger distillation routine that may arbitrarily suppress the error rate. In particular, the distillation routine will have multiple levels where the output of one level is used as the input to the next level. Where $[[n,k,d]]$ are the parameters of the quantum code at hand, both papers show that, in order to generate a supply of magic states with an error rate below $\epsilon$, the number of input magic states (of some (small) constant error rate) required per output state is
\begin{equation}\label{canonicalOverhead}
    \mathcal{O}(\log^\gamma(1/\epsilon)),
\end{equation}
which is called the overhead of the magic state distillation protocol, where the exponent $\gamma$ is\footnote{Conventionally, magic state distillation is performed by concatenating protocols using an $[[n,k,d]]$ code $L$ times. The error rate becomes $\epsilon=(p_\mathrm{ph}/p_\mathrm{th})^{\mathcal{O}(d^L)}$, i.e., $L\log d=\log(\mathcal{O}(\log(1/\epsilon)))$, where $p_\mathrm{th}$ is a threshold, and $p_\mathrm{ph}<p_\mathrm{th}$ is a physical error rate. The overhead is given by $(n/k)^L=\mathcal{O}(\log^\gamma(1/\epsilon))$ with $\gamma = \log(n/k)/\log d$. Note that this overhead scaling is not applicable to our protocol since we do not perform the concatenation but develop a way to suppress the error rate by a single-round protocol using asymptotically good triorthogonal codes.}
\begin{equation}
    \gamma = \frac{\log(n/k)}{\log d}.
\end{equation}
Now, the comparison of this conventional protocol and our protocol for constant-overhead magic state distillation may be made. The most obvious thing to say is that the bulk of our protocol, step 3, is preceded and succeeded by conversions from $\ket{CCZ}$ to $\ket{U_s}$ states, and then vice versa. The next thing to note is that our step 3 strongly resembles the magic state distillation subroutine of \cite{bravyi2012magic} and \cite{krishna2019towards}; however, whereas their magic state distillation subroutine declares a failure upon the measurement of any non-trivial syndrome, we make an attempt to correct the error in this case, without post-selection. The last difference to note is that we use only one round of magic state distillation of the $\ket{U_s}$ states, whereas, in the more conventional approach, the subroutine is concatenated several times; importantly, the one-round protocol can arbitrarily suppress error rate in our case because of the linear distance $d=\Theta(n)$ of our $[[n,k,d]]$ code, which was not achieved in the existing protocols.

Note that one could also perform the concatenated magic state distillation subroutines using our code, by simply performing the initial $\ket{CCZ} \mapsto \ket{U_s}$ conversion, concatenating several rounds of the magic state distillation subroutine for $\ket{U_s}$ states, and then converting $\ket{U_s} \mapsto \ket{CCZ}$ back. Since the first and final conversions only come with a constant-factor loss only, the overhead of $\ket{CCZ}$ distillation would be achieved again with the overhead of Equation \eqref{canonicalOverhead}.
Since our quantum code has asymptotically good parameters $(k,d = \Theta(n))$,
one would then be able to achieve arbitrarily small yet nonzero $\gamma \to 0$ in Equation \eqref{canonicalOverhead}, ``arbitrarily-low-polylogarithmic overhead magic state distillation''.
However, we use the asymptotically good parameters of our quantum code to show an even stronger result, i.e., constant overhead $\gamma = 0$, as we will argue in Section \ref{errorAnalysis}.

\subsection{Twirling the Magic States}\label{twirling}

Before going into more detail on the bulk of our protocol, we make a brief comment on the twirling~\eqref{eq:twirling} of the input $\ket{CCZ}$ magic states. This is performed before the conversion to $\ket{M}$ states, in step 1 of the protocol described in Section~\ref{MSDSummarySection}.
We start by explaining the twirling of a single $CCZ$ state and then proceed by explaining that of multiple $CCZ$ states, which may suffer from correlated error under the local stochastic error model in Section~\ref{sec:task}.

In the case of a single $CCZ$ state, a general noisy input magic state may be written as the three qubit state $\rho_{CCZ}$, where
\begin{equation}
    1-\bra{CCZ}\rho_{CCZ}\ket{CCZ} \leq p_\mathrm{ph},
\end{equation}
and $p_\mathrm{ph}$ is the input error rate.
For $\mathbf{b}=(b_1,b_2,b_3) \in \mathbb{F}_2^3$, let us define
\begin{equation}
    \ket{M_{\mathbf{b}}^{(3)}} \coloneqq \left(Z^{b_1}\otimes Z^{b_2}\otimes Z^{b_3}\right)\ket{CCZ}.
\end{equation}
It holds that $\braket{M_{\mathbf{b}}^{(3)}|M_{\mathbf{c}}^{(3)}} = \delta_{\mathbf{b}\mathbf{c}}$, so that these states form a basis for the space of three qubits.
For the stabiliser of $\ket{CCZ}$ generated by $X_1\otimes CZ_{2,3}$, $X_2\otimes CZ_{1,3}$, and $X_3\otimes CZ_{1,2}$ used in the twirling~\eqref{eq:twirling}, $\ket{M_{\mathbf{b}}^{(3)}}$ is one of the $\pm 1$ eigenstates.
Then, due to a straightforward generalisation of the corresponding argument for the $T$ gate~\cite{bravyi2005universal,bravyi2012magic}, the twirling transforms $\rho_{CCZ}$ into a dephased state
\begin{gather}\label{diagonalrhoCCZ}
    \left(1-\sum_{\mathbf{b} \in \mathbb{F}_2^3, \mathbf{b} \neq 0}p_{\mathbf{b}}\right)\ket{CCZ}\bra{CCZ} + \sum_{\mathbf{b}\in \mathbb{F}_2^3, \mathbf{b} \neq \mathbf{0}}p_{\mathbf{b}}\ket{M_\mathbf{b}^{(3)}}\bra{M_\mathbf{b}^{(3)}},,
\end{gather}
for some $p_{\mathbf{b}}\geq 0$ with $\sum_{\mathbf{b} \in \mathbb{F}_2^3, \mathbf{b} = 0}p_{\mathbf{b}} \leq p_\mathrm{ph}$.
That is, the state~\eqref{diagonalrhoCCZ} after the twirling becomes diagonal in the $\{\ket{M_\mathbf{b}^{(3)}}\}$ basis, only having $Z$-type three-qubit errors $Z^{b_1}\otimes Z^{b_2}\otimes Z^{b_3}$ on $\ket{CCZ}$.

In the case of multiple $CCZ$ states, the same argument implies that the states after the twirling only have $Z$-type errors.
The difference is that we may have correlated errors in the local stochastic error model, which is more general than the IID error model and used for the conventional analysis of fault-tolerant quantum computation~\cite{gottesman2010,yamasaki2024time}.
Given $\zeta$ initial $CCZ$ states prepared by the $CCZ$-state preparation operations undergoing the local stochastic error model at the physical error rate $p_\mathrm{ph}$ in~\eqref{eq:physical_error_rate}, the state after the twirling can then be represented as
\begin{align}
    \sum_{F}p(F)E_Z^{(F)} \left(\ket{CCZ}\bra{CCZ}^{\otimes \zeta}\right)E_Z^{(F)\dag},
\end{align}
where $F$ is the randomly chosen set of faulty locations among the $\zeta$ $CCZ$-state preparation operations, $E_Z^{(F)}$ is a $Z$-type error with its support included in the $3|F|$ qubits of $F$, and $p(F)$ is the probability distribution used in~\eqref{eq:physical_error_rate}.
In particular, by definition of the local stochastic error model, for any set $S$ of states among the $\zeta$ three qubit states after the twirling,~\eqref{eq:physical_error_rate} shows that the probability of having some $Z$-type errors at every state in $S$ is at most
\begin{align}
\Pr\{F\supset S\}\leq p_\mathrm{ph}^{|S|}.
\end{align}

It is for the reason of twirling that we may assume stochastic $Z$ errors only on our input magic states $\ket{CCZ}$, and subsequently on our $s$-qubit magic states $\ket{U_s}$. Correspondingly, the same statement may be made about our input single qudit magic states $\ket{M}$, and this is why only the $X$-stabilisers need to be measured in our main magic state distillation protocol. Accordingly, in our magic state distillation protocol, we only need a decoder for the code $\mathcal{G}_0^\perp$, to correct $Z$-type errors.

\subsection{Magic State Distillation from the Quantum Code}\label{MSDfromCode}

In this section, we go into more detail about the bulk of our distillation protocol, which is the distillation of the $s$-qubit magic state $\ket{U_s}$, which can equivalently be thought of as the distillation of the one-qudit state $\ket{M}$. This occurred in steps 3a to 3e of our constant-overhead magic state distillation protocol. In this section, as usual, we write $q=2^s$ and think of $s=10$, although other values work in certain places.

By using Clifford operations, exactly one copy of $\ket{M}$ may be consumed to execute $U_1^{(7)}$ on an unknown single qudit state $\ket{\psi}$ via a standard gate teleportation procedure. The gate teleportation protocol goes as follows.

\begin{enumerate}
    \item Begin with the ancillary state $\ket{M}$ in the first (single qudit) register and the target state $\ket{\psi}$ to which we wish to apply $U_1^{(7)}$ in the second (single qudit) register.
    \item Apply
    \begin{equation}
        SUM = \sum_{\gamma \in \mathbb{F}_{q}}\ket{\gamma}\bra{\gamma}\otimes X^\gamma.
    \end{equation}
    \item Measure the second register in the computational basis $\{\ket{\eta}\}_{\eta \in \mathbb{F}_{q}}$. Retain the outcome $\beta$.
    \item Apply $U_1^{(7)}X^\beta U_1^{(7)}$ to the first register.
    \item The state $U_1^{(7)}\ket{\psi}$ is found in the first register.
\end{enumerate}
We emphasise that all the operations performed here are Clifford operations. Indeed, one can conjugate $X^\gamma$ and $Z^\gamma$ by $SUM$ to verify that $SUM$ is Clifford. Moreover, because we showed that $U_1^{(7)}$ is in the third level of the Clifford hierarchy in Section \ref{nonCliffordGate}, the gate $U_1^{(7)}X^\beta (U_1^{(7)})^\dagger = U_1^{(7)} X^\beta U_1^{(7)}$ is Clifford. We note that if one applies the above gate teleportation procedure by using a noisy magic state $\ket{M}$, rather than the ideal state, a noisy version of $U_1^{(7)}$ will be implemented, rather than the ideal $U_1^{(7)}$. The relevant error analysis will be performed in Section \ref{errorAnalysis}.

Let us note that this gate teleportation procedure for the $U_1^{(7)}$ is used in the bulk of our constant-overhead magic state distillation protocol, in particular in Step 3b in Section \ref{MSDSummarySection}.

Now, although it is a slightly tedious rewording, for the sake of clarity we would like to elucidate the translation between the distillation protocol, expressed in terms of $q$-dimensional qudits, and that on qubits. The corresponding qubit protocol distils an $s$-qubit magic state corresponding to $\ket{M}$, which we have denoted $\ket{U_s}$.

We recall the correspondence between a $q$-dimensional qudit, and $s$ qubits presented in Section \ref{primePowerQuditsPrelims}. Indeed, we fixed a self-dual basis for $\mathbb{F}_{q}$ over $\mathbb{F}_2$, call it $(\alpha_i)_{i=0}^{s-1}$, and constructed a linear map $\psi : \mathbb{C}^{q} \to (\mathbb{C}^2)^{\otimes s}$. $\psi$ mapped the computational basis states of $\mathbb{C}^{q}$, denoted $\ket{\beta}$ ($\beta \in \mathbb{F}_{2^s}$), to $\ket{\mathbf{b}}$, where $\mathbf{b} \in \mathbb{F}_{2}^s$ is the bit string of coefficients of $\beta$ in the basis $(\alpha_i)$. Similarly, we constructed a map $\theta$ which maps unitaries acting on $\mathbb{C}^{q}$ to unitaries acting on $(\mathbb{C}^2)^{\otimes s}$. Given a unitary $U$ acting on $\mathbb{C}^{q}$ and a matrix for $U$ corresponding to the computational basis of $\mathbb{C}^{q}$, $\theta(U)$ had an identical matrix corresponding to the computational basis of $(\mathbb{C}^2)^{\otimes s}$, where the two computational bases are identified via $\psi$. We further showed that $\theta$ specialises to an isomorphism between the Pauli groups, the Clifford groups, and the diagonal gates in the third level of the Clifford hierarchy.

When phrased in terms of qubits, our magic state distillation protocol distils the $s$-qubit state $\psi(\ket{M})$. This is
\begin{equation}
    \ket{U_s} \coloneq \psi(\ket{M}) = \theta(U)\left(\ket{+}^{\otimes s}\right)
\end{equation}
where again we abbreviate $U = U_1^{(7)}$ and $\ket{+}$ is the usual single-qubit +1-eigenstate of the $X$-operator. It is important to note that because $\theta$ specialises to an isomorphism between the Clifford groups, and between diagonal gates in the third level of the Clifford hierarchy, the fact that $U$ is in the third level, but not the second level of the Clifford hierarchy, implies the same statement about $\theta(U)$: that it is not Clifford, but in the third level of the Clifford hierarchy.

The quantum code may be phrased as a stabiliser code on qubits, equivalently by either applying $\theta^{\otimes n}$ to every stabiliser of the qudit code, thus defining a group of a qubit stabiliser, or by applying the map $\psi^{\otimes n}$ to every code state, thus defining a qubit code space. 

The logical operators $\overline{X^\beta_a}$ and $\overline{Z^\gamma_a}$ may be carried over to the qubit code via $\theta^{\otimes n}$ to give the expected $ks$ independent logical operators (of both $X$ and $Z$ type) acting on $ns$ qubits, and these give the expected commutation relations between each other, and with the stabilisers. It is most natural to think of the $ks$ logical qubits of the code in $k$ sets of $s$ qubits, where we label each set of $s$ qubits via the indices $a,b \in \{1, ..., k\}$. The encoded operators $X^{\mathbf{b}}$ and $Z^{\mathbf{c}}$ ($\mathbf{b}, \mathbf{c} \in \mathbb{F}_{2}^s$) acting on the $a$-th set of logical qubits are denoted
\begin{equation}
    \overline{X^\mathbf{b}_a} \text{ and }\overline{Z^{\mathbf{c}}_a}.
\end{equation}
Indeed, it is most natural to let $\overline{X^\mathbf{b}_a} \text{ and }\overline{Z^{\mathbf{c}}_b}$ be respectively $\overline{X^\beta_a}$ and $\overline{Z^\gamma_b}$ carried over to the qubit code via $\theta^{\otimes n}$, where $\mathbf{b}$ and $\mathbf{c}$ are the bit strings of coefficients of $\beta$ and $\gamma$ respectively in the fixed self-dual basis. This then gives us the required commutation relations for $\overline{X^{\mathbf{b}}_a}$ and $\overline{Z^{\mathbf{c}}_b}$, because $(-1)^{\tr(\beta\gamma)\delta_{ab}} = (-1)^{\mathbf{b}\cdot\mathbf{c}\delta_{ab}}$.

At the same time, the encoded computational basis states $\overline{\ket{u}}$ may be carried over to the qubit code via $\psi^{\otimes n}$. In particular, considering an encoded computational basis state $\overline{\ket{u}}$ of the qudit code, where $u \in \mathbb{F}_{q}^k$, we may write the elements $(u_a)_{a=1}^k$ in terms of the self-dual basis with sets of coefficients $(\mathbf{b}_a)_{a=1}^k$, where $\mathbf{b}_a \in \mathbb{F}_2^s$, and then we identify the encoded computational basis states of the qubit code as
\begin{equation}
    \overline{\ket{\mathbf{b}_1...\mathbf{b}_k}} = \psi^{\otimes n}\left(\overline{\ket{u}}\right)
\end{equation}
and these behave in the required way under our choice of logical operators.

Ultimately, we have seen that from our code on $q$-dimensional qudits, which has $n$ physical qudits and $k$ logical qudits, we can define a qubit code on $ns$ physical qudits, encoding $ks$ logical qudits, which is ``the same'' in all respects. All code states can be taken to be the same states, with different labels for their bases, and all stabilisers and logical operators are the same matrices (in the respective computational bases).

The qudit code had the property that there was a one qudit non-Clifford (third-level) diagonal gate $U$ such that acting with $U^{\otimes n}$ on a code state had the same effect as acting with the encoded gate $\overline{U^{\otimes k}}$. One may consider the same gate acting on $s$ qubits, $\theta(U)$, which we know to be diagonal and in exactly the third level of the Clifford hierarchy. Treating the $ns$ physical qubits of the code in $n$ sets of $s$ qubits, if we act with $\theta(U)^{\otimes n}$ on a code state, where each copy of $\theta(U)$ acts on a separate set of $s$ qubits, then it is quite easy to see that this will have the same effect as acting with the encoded gate $\overline{\theta(U)^{\otimes k}}$. Indeed, for $u \in \mathbb{F}_{q}^k$, we had
\begin{equation}
    U^{\otimes n}\overline{\ket{u}} = \overline{U^{\otimes k}}\; \overline{\ket{u}},
\end{equation}
meaning that when $U^{\otimes n}$ acts on $\overline{\ket{u}}$, it adds the same phase as if $U^{\otimes k}$ had acted on $\ket{u}$. If we then consider the qubit encoded computational basis state $\overline{\ket{\mathbf{b}_1...\mathbf{b}_k}}$ corresponding to $\overline{\ket{u}}$, where $\mathbf{b}_a \in \mathbb{F}_2^s$, 
\begin{equation}
    \theta(U)^{\otimes n}\overline{\ket{\mathbf{b}_1...\mathbf{b}_k}} = \overline{\theta(U)^{\otimes k}}\;\overline{\ket{\mathbf{b}_1...\mathbf{b}_k}}
\end{equation}
because $\theta(U)$ has the same coefficients as $U$ in its computational basis, and $\overline{\ket{\mathbf{b}_1...\mathbf{b}_k}}$ has the same coefficients as $\overline{\ket{u}}$ in its computational basis.

For convenience, we will go on to rephrase the bulk of our magic state distillation protocol, steps 3a to 3e, in terms of qubits. 

First, the gate teleportation protocol for $\theta(U)$ goes as follows.
\begin{enumerate}
    \item Begin with the $s$-qubit state $\ket{U_s}$ in the first ``$s$-qubit register'' and the target $s$-qubit state $\ket{\psi}$ in the second ``$s$-qubit register''.
    \item Apply $\theta(SUM)$ from the first register to the second.
    \item Measure the second register in the computational basis $\{\ket{\mathbf{e}}\}_{\mathbf{e} \in \mathbb{F}_{2}^s}$. Retain the outcome $\mathbf{b}$.
    \item Apply $\theta(U)X^{\mathbf{b}}\theta(U)$ to the first register.
    \item The state $\theta(U)\ket{\psi}$ is found in the first register.
\end{enumerate}
Note that all these operations are Clifford operations; in particular $\theta(SUM)$ and $\theta(U)X^{\mathbf{b}}\theta(U)$ are Clifford operations because $\theta$ maps Cliffords to Cliffords and third-level gates to third-level gates. Again, applying the gate teleportation protocol with noisy $\ket{U_s}$, rather than the ideal state, will implement a noisy version of $\theta(U)$.

With this, the bulk of our constant-overhead magic state distillation protocol, steps 3a to 3e, may be rephrased on qubits as:
\begin{enumerate}\setcounter{enumi}{2}
\item \begin{enumerate}
    \item Prepare $k$ sets of $s$ qubits, where each set is in the state $\ket{+}^{\otimes s}$, and $\ket{+}$ is the usual single-qubit $+1$-eigenstate of the $X$-operator. Encode these qubits into the quantum code.
    \item Use $n$ noisy input magic states $\rho$, which are each noisy versions of the $s$-qubit gate $\ket{U_s}$, and the qubit gate teleportation protocol to apply a noisy $\theta(U)$ gate to each of the $n$ sets of $s$ physical qudits of the code.
    \item Measure the $X$-stabilisers of the code, obtaining a syndrome $S_{bin} \in \left(\mathbb{F}_2^s\right)^{m-k}$. This may be used to deduce a syndrome for the equivalent qudit code $S \in \mathbb{F}_{q}^{m-k}$. Pass the syndrome to the classical decoder.
    \item If the decoder returns an error vector $v_S \in \mathbb{F}_{q}^n$, apply the Pauli $\theta^{\otimes n}(Z^{v_S})$ to the state.
    \item Decode the state from the quantum code. The remaining $ks$ qubits are treated as $k$ copies of the $\ket{U_s}$ state, of higher quality.
\end{enumerate}
\end{enumerate}

\subsection{CCZ Conversion}\label{CCZConversion}

In this section, we demonstrate how the described distillation of the $s$-qubit state $\ket{U_s}$ can be converted, up to some constant loss, into the distillation of $\ket{CCZ}$ states. Again, we write $q=2^s$ and think of $s=10$ for our main construction.

The $\ket{CCZ}$ state enables the execution of the qubit $CCZ$ gate via Clifford operations only according to the following gate teleportation protocol.
\begin{enumerate}
    \item Begin with $\ket{CCZ}$ in the first ``3-qubit register'' and the target 3-qubit state $\ket{\psi}$ in the second ``3-qubit register''.
    \item Apply $CNOT^{\otimes 3}$ with the three qubits of the first register as the controls and the three qubits of the second register as the targets.
    \item Measure the second register in the computational basis $\{\ket{\mathbf{e}}\}_{\mathbf{e} \in \mathbb{F}_2^3}$. Retain the outcome $\mathbf{b}$.
    \item Apply (the Clifford) $CCZ\; X^{\mathbf{b}}\; CCZ$ to the first register.
    \item The state $CCZ\ket{\psi}$ is found in the first register.
\end{enumerate}

Let us recall that the gate from the last section, $\theta(U)$, is an $s$-qubit diagonal gate with phases $\pm 1$ only on the diagonal and that it is in exactly the third level of the Clifford hierarchy. It is known that any diagonal $s$-qubit gate with $\pm 1$ only on the diagonal may be decomposed into a sequence of gates $C^kZ$ gates \cite{houshmand2014decomposition} for $0 \leq k \leq s-1$, where $C^kZ$ is the $Z$-gate acting on one target qubit, controlled on $k$ other qubits being in the state $\ket{1}$. Note that all these gates commute, and square to the identity - therefore each $C^kZ$ gate appears at most once in the decomposition. Moreover, because \cite{houshmand2014decomposition} shows that these gates are independent (see Lemma 1 in that paper), this decomposition is in fact unique\footnote{Note that this decomposition is unique, once one has fixed an isomorphism between qudits of dimension $q$ and sets of $s$ qubits. This is done by fixing a self-dual basis for $\mathbb{F}_q$ over $\mathbb{F}_2$. A different self-dual basis may lead to a different isomorphism, and therefore $\theta(U)$ will be different (although it is always a gate in exactly the third level of the Clifford hierarchy), leading to a different decomposition.}.

Now, $CCZ$ is a diagonal gate in the third level of the Clifford hierarchy and thus enables universal quantum computation when combined with Clifford operations. It is also considered quite a natural non-Clifford gate to implement, whereas the $s$-qubit gate $\theta(U)$ discussed previously is non-standard, so it is advantageous to demonstrate this conversion.

There, therefore, exists a unique decomposition of the $s$-qubit gate $\theta(U)$ into a sequence of $Z$, $CZ$, and $CCZ$ gates. This is so because it is in the third level of the Clifford hierarchy, and $C^kZ$ is in the $(k+1)$-th level. Note that $\theta(U)$ is not a Clifford gate, and so this decomposition includes at least one $CCZ$ gate. Suppose that this decomposition uses $C$ $CCZ$ gates.  Here, as throughout the paper, we are only interested in asymptotic performance and do not attempt to optimise constant factors. Nevertheless, we mention that $1 \leq C \leq 120$ (for our main construction with $s=10$) because the decomposition of $\theta(U)$ into $Z$, $CZ$ and $CCZ$ gates may contain at most $\binom{s}{3}=\binom{10}{3} = 120$ $CCZ$ gates.
The required number $C$ of $CCZ$ gates in such gate decomposition depends on the choice of the self-dual basis.
In Appendix \ref{exampleDecomposition}, we provide an example of decomposition of the gate $\theta(U)$ into $Z$, $CZ$, and $CCZ$ gates which uses $C=70$ $CCZ$ gates for an example of the self-dual basis.

To perform $CCZ$ distillation, we imagine that we start with some supply of input, lower quality (higher error rate) $\ket{CCZ}$ states. By consuming $\ket{CCZ}$ states in their gate teleportation protocols, as well as Cliffords $Z$ and $CZ$, lower quality $\ket{U_s}$ state may be prepared, by acting according to the decomposition of the gate $\theta(U)$ on the state $\ket{+}^{\otimes s}$. Note that $C$ $\ket{CCZ}$ states are required to make one $\ket{U_s}$ state, so generating n $\ket{U_s}$ states requires $Cn$ initial $\ket{CCZ}$ states. One then performs the distillation protocol described in the previous section to the collection of $\ket{U_s}$ states to produce a smaller number, $k$, of higher quality $\ket{U_s}$ states.

Given this supply of higher quality $\ket{U_s}$ states, we would like to finish with a supply of higher quality $\ket{CCZ}$ states. One makes this final conversion as follows. Given a $\ket{U_s}$ state, consider any three out of the $s$ qubits such that in the decomposition of $\theta(U)$, there is a $CCZ$ state acting on those three qubits. We measure the remaining $s-3$ qubits in the $Z$ basis and can perform Clifford gates conditioned on the result to obtain a $\ket{CCZ}$ state. Explicitly, we do the following (for the sake of convenience assume that the decomposition of $\theta(U)$ includes a gate $CCZ_{1,2,3}$, and the qubits that are measured are therefore $4, \ldots, s$; we will end with a $\ket{CCZ}$ state on qubits 1,2,3).
\begin{enumerate}
    \item Suppose that the decomposition of $\theta(U)$ contained a $Z$ or $CZ$ gate supported only on the qubits 1,2,3. We act with this same gate again.
    \item Suppose that the decomposition of $\theta(U)$ contained a $CZ$ gate acting on qubits $a,b$ where $a \in \{1, 2, 3\}$ and $b \notin \{1, 2, 3\}$. If the $b$-th qubit was measured to be 1, then perform $Z$ on the $a$-th qubit.
    \item Suppose that the decomposition of $\theta(U)$ contained a $CCZ$ gate acting on qubits $a,b,c$ where $a,b \in \{1,2,3\}$ and $c \notin \{1,2,3\}$. If the $c$-th qubit was measured to be 1, then perform $CZ$ across the $a$-th and $b$-th qubits.
    \item Suppose that the decomposition of $\theta(U)$ contained a $CCZ$ gate acting on qubits $a,b,c$ where $a \in \{1,2,3\}$ and $b,c \notin \{1,2,3\}$. If the $b$-th and $c$-th qubits were both measured to be 1, then perform $Z$ on the $a$-th qubit.
\end{enumerate}
It can be checked that this procedure reproduces a $\ket{CCZ}$ state upon all measurement outcomes. The error analysis for the whole distillation procedure will be performed in Section \ref{errorAnalysis}.

\section{Triorthogonal Matrices from Algebraic Geometry Codes}\label{mainConstruction}

In this section, we will provide an explicit construction of the triorthogonal matrices over $\mathbb{F}_q$ which can be used to make the quantum codes in Section \ref{quantumCodefromMatrix}, and thus the distillation protocols constructed in Section \ref{MSDProtocolMainSection}. We will prove all the required properties of the matrix constructed, namely that it is triorthogonal, and that the quantum code that arises from it is asymptotically good, i.e. $k=\Theta(n)$ and $d=\Theta(n)$. When combined with the results of Section \ref{mainCodeSection}, this establishes Theorem \ref{quantumCodeTheorem}.

It was shown in Claim \ref{distanceClaim} that the distance of the code is, in fact, equal to its $Z$-distance, so we will be satisfied for these purposes with a lower bound on the $Z$-distance. The construction itself will take place in Sections \ref{AGCodeConstruction} and \ref{puncturingSection}, while in Section \ref{concrete}, we will make things concrete, meaning that we will give values to all the parameters of our construction. This will show that all hypotheses of our construction can be simultaneously satisfied, and that the parameters of the resulting code are asymptotically good. In Sections \ref{AGCodeConstruction} and \ref{puncturingSection}, we will denote $q=2^s$ for brevity. We will reintroduce $q=2^{10}=1024$ when we make things concrete in Section \ref{concrete}. We will go on to show, however, that this construction of asymptotically good triorthogonal codes for qudits of dimension $2^s$ works for all $s$ satisfying $s \geq 10$ and $s \not\equiv 0 \pmod 3$. Therefore, the following theorem is true, which is a more technical version of Theorem \ref{quantumCodeTheorem}.
\begin{theorem}\label{technicalCodeTheorem}
    Fix any constant $s \geq 10$ with $s \not\equiv 0 \pmod 3$, and take $q=2^s$.
    If we have a family of algebraic function fields of one variable $(F_i/\mathbb{F}_q)_{i=1,2,\ldots}$ with genera $(g_i)_{i=1,2,\ldots}$, where each has a set $\{P_1,\ldots,P_{n_i^\prime}\}$ of $n_i^\prime$ rational places satisfying
    \begin{align}
    n_i^\prime-4+g_i\geq 7(3g_i+2),
    \end{align}
    then for any sequences $\{a_i\}_i$ and $\{k_i\}_i$ of integers satisfying
    \begin{align}
        &a_i\geq 3g_i+2,\\
        &n_i^\prime-4+g_i-7a_i\geq 0,\\
        &a_i-3g_i-1\geq k_i>0,
    \end{align}
    there exists a sequence $\{N_i\}_i$ of integers satisfying
    \begin{align}
        n^{\prime}_i-2\geq N_i\geq n^{\prime}_i-2-g_i,
    \end{align}
    and a family of $[[n_i,k_i,d_i]]$ quantum codes on $2^s$-dimensional qudits such that the code parameters are
    \begin{align}
    n_i&=N_i-k_i,\\
    d_i&\geq a_i-k_i-(2g_i-2),
    \end{align}
    and any code in this family supports a non-Clifford transversal gate that is in the third level of the Clifford hierarchy (i.e., there is a single-qudit gate in exactly the third level of the Clifford hierarchy, $U$, such that the physical transversal gate $U^{\otimes n_i}$ acts in the same way as the logical transversal gate $\overline{U^{\otimes k_i}}$ on code states).
    In addition, for any sequence $\{t_i\}_i$ of parameters satisfying
    \begin{align}
        0<t_i\leq\frac{[a_i-k_i-(2g_i-2)]-g_i-1}{2},
    \end{align}
    this family of quantum codes have an $O(\mathrm{poly}(n_i))$-time decoder with the decoding radius $t_i$ for correcting $Z$-type errors (as is relevant to magic state distillation).
    There exists an explicit construction of an infinite family $\{F_i/\mathbb{F}_q\}_{i=1,2,\ldots}$ of algebraic functions fields that enable an explicit construction of this code family achieving
    \begin{align}
        k_i&=\Theta(n_i),\\
        d_i&=\Theta(n_i),\\
        t_i&=\Theta(n_i),
    \end{align}
    as $i\to\infty$.
\end{theorem}

Before beginning, let us recall the definition of triorthogonality that we aim for.
\begin{definition}[Triorthogonal Matrix]\label{triorthogMatDefinition}
    For us, a matrix $G \in \mathbb{F}_{q}^{m \times n}$ with rows $(g^a)_{a=1}^m$ is called triorthogonal if for all $a,b,c \in \{1, ..., m\}$,
    \begin{equation}\label{triorthogSecondDef1}
        \sum_{i=1}^n(g^a_i)^4(g^b_i)^2(g^c_i) = \begin{cases}
            1 &\text{ if } 1 \leq a=b=c \leq k\\
            0 &\text{ otherwise}
        \end{cases}
    \end{equation}
    and
    \begin{equation}\label{triorthogSecondDef2}
        \sum_{i=1}^n \sigma_ig_i^ag_i^b = \begin{cases}
            \tau_a &\text{ if } 1 \leq a=b \leq k\\
            0 &\text{ otherwise}
        \end{cases}
    \end{equation}
    for some integer $k \in \{1, ..., m\}$ and for some $\sigma_i, \tau_a \in \mathbb{F}_{q}$, where $\sigma_i, \tau_a \neq 0$. In these expressions, all arithmetic takes place over $\mathbb{F}_{q}$.
\end{definition}

\subsection{The Algebraic Geometry Code}\label{AGCodeConstruction}

We rely heavily on the details of algebraic geometry codes for our construction. The necessary preliminary material is presented in Section \ref{AGPrelims}. Our argument also uses several ideas from the quantum algebraic geometry codes of \cite{ashikhmin2001asymptotically}, along with some key differences, but what follows is self-contained within this paper.

Let us start by considering an algebraic function field of one variable $F/\mathbb{F}_q$ with genus $g>0$, and let us denote a set of its rational places as $\mathcal{P}' = \{P_1, ..., P_{n'}\}$, where we assume a condition
\begin{align}
\label{hypothesis1}
    n^\prime-4+g>7(3g+2).
\end{align}
Now consider in addition two positive divisors $A \in \mathrm{Div}(F)$ and $E \in \mathrm{Div}(F)$ with
\begin{align}
    \label{positivityCondition_a}
    a&\coloneqq \deg(A) \geq 3g+2,\\
    \label{positivityCondition}
    \deg(E) &= n'-4+g-7a \geq 0,
\end{align}
where under the condition~\eqref{hypothesis1}, we can always find such $a$ and the corresponding divisors, e.g., $A = aP_{n'-1}$ and $E = (n'-4+g-7a)P_{n'}$.
For later use in Section~\ref{puncturingSection}, pick an integer $k$ with
\begin{align}
\label{eq:condition_k}
     a-3g-1\geq k>0,
\end{align}
which exists due to the condition~\eqref{positivityCondition_a}.
The condition~\eqref{hypothesis1}, which yields all the other conditions~\eqref{positivityCondition_a},~\eqref{positivityCondition}, and~\eqref{eq:condition_k}, will turn out to be satisfied when particular values are chosen as in Section \ref{concrete}.

Let $\mathcal{P}'' = \{P_1, ..., P_{n'-2}\}$, where without loss of generality $\supp(A)\cap \mathcal{P}'' = \supp(E) \cap \mathcal{P}'' = \emptyset$\footnote{Recall that given any divisor $D$, its support $\supp(D)$ is the set of places $P$ for which $\nu_P(D) \neq 0$.}, because we can pick $A$ and $E$ to be supported at only one rational place. For convenience, let
\begin{align}
\label{eq:n_prime_prime}
    n'' \coloneq n'-2,
\end{align}
so that $\mathcal{P}'' = \{P_1, ..., P_{n''}\}$.
Now define new divisors
\begin{align}
    D'' &\coloneq P_1 + ... + P_{n''}\\
    B &\coloneq D''-7A-E
\end{align}
and we note that $\deg(B) = 2-g$. Given any canonical divisor $W$, Riemann's theorem (Equation~\eqref{eq:riemannTheorem}) tells us that
\begin{equation}
    l(W+B) \geq (2g-2)+(2-g)+1-g = 1
\end{equation}
(since we know that $\deg(W) = 2g-2$ due to Equation~\eqref{eq:degree_canonical_divisor}) and so by the Duality Theorem (Theorem \ref{dualityThm}), the space $\Omega_F(-B)$ contains a non-zero element. Let $0 \neq \omega_0 \in \Omega_F(-B)$ and, as in Definition~\ref{def:canonical_divisor}, we write $(\omega_0)$ for the canonical divisor corresponding to $\omega_0$. By Remark \ref{canonicalDivisors}, $(\omega_0) \geq -B$, i.e.,
\begin{equation}
    (\omega_0) \geq 7A+E-D''.
\end{equation}
For any $P_i \in \mathcal{P}''$, by Definition~\ref{def:divisor}, we calculate
\begin{align}
\nu_{P_i}\left((\omega_0)\right) &\geq 7\nu_{P_i}(A)+\nu_{P_i}(E)-\nu_{P_i}(D'') \\
&=-\nu_{P_i}(D'')\\
&=-1\label{omega0Fact}
\end{align}
where the second line uses the fact that $\supp(A)\cap \mathcal{P}'' = \supp(E)\cap \mathcal{P}'' = \emptyset$, so $\nu_{P_i}(A) = \nu_{P_i}(E)=0$. With this in mind, let us define the following subset $\mathcal{P}_0 \subseteq \mathcal{P}''$:
\begin{equation}
    \mathcal{P}'' \supseteq \mathcal{P}_0 \coloneq \{P_{i_1}, ..., P_{i_M}\} = \{P_{i_j} \in \mathcal{P}'' : \nu_{P_{i_j}}\left((\omega_0)\right) \geq 0\}.
\end{equation}
In words, $\mathcal{P}_0$ is the set of $P_{i_j}$ in $\mathcal{P}''$ for which we not only have $\nu_{P_{i_j}}((\omega_0)) \geq -1$, but also $\nu_{P_{i_j}}((\omega_0)) \geq 0$. What is the significance of the set $\mathcal{P}_0$? Using a claim showed in the proof of Theorem 2.2.7 in \cite{stichtenoth2009algebraic}, we have that for $P_{i_j} \in \mathcal{P}''$, $\nu_{P_{i_j}}((\omega_0)) \geq 0 \iff \omega_{0P_{i_j}}(1) = 0$, where $\omega_{0P_{i_j}}$ denotes the local component of $\omega_0$ at $P_{i_j}$. To construct our triorthogonal matrix, it will be of great importance to us to only use the $P_{i_j}$ for which $\omega_{0P_{i_j}}(1) \neq 0$. To this end, we set
\begin{equation}
    \mathcal{P} \coloneq \mathcal{P}'' \setminus \mathcal{P}_0
\end{equation}
and after relabelling we can write
\begin{equation}
    \mathcal{P} = \{P_1, ..., P_{N}\},
\end{equation}
but how do we know how large $\mathcal{P}$ is? Do we even know that it is non-empty? We show the following:
\begin{claim}\label{boundP0}
$|\mathcal{P}_0| = M \leq g$.
\end{claim}

\begin{proof}
    Because $(\omega_0) \geq 7A+E-D''$ and $\nu_{P_{i_j}}(\omega_0) \geq 0$ for all $P_{i_j} \in \mathcal{P}_0$, we have that
\begin{equation}\label{omega0geq}
    (\omega_0) \geq 7A+E-D'' + \sum_{j=1}^MP_{i_j}.
\end{equation}
Taking degrees (and recalling that all canonical divisors have degree $2g-2$) yields
\begin{align}
    2g-2&\geq 7a+(n''-2+g-7a)-n''+M\\
    \implies g&\geq M.
\end{align}
\end{proof}
The claim then establishes that $n'' \geq N \geq n''-g$, meaning that $\mathcal{P}$ is non-empty if $n'' > g$; due to Equation~\eqref{eq:n_prime_prime}, we have
\begin{equation}\label{boundN}
    n'-2\geq N \geq n'-2-g,
\end{equation}
where these values are always positive under the condition~\eqref{hypothesis1}.
By setting
\begin{equation}
    D \coloneq P_1 + ... + P_N,
\end{equation}
we see from the proof of Claim \ref{boundP0}, in particular Equation \eqref{omega0geq}, that
\begin{equation}
    (\omega_0) \geq 7A+E-D
\end{equation}
and so
\begin{equation}\label{omega0Membership}
    \omega_0 \in \Omega_F(7A+E-D)
\end{equation}
by Remark \ref{canonicalDivisors}.

We may now consider the algebraic geometry code $C = C_{\mathcal{L}}(D,A)$ as in Definition~\ref{def:AG_code}, which we write out explicitly for the sake of completeness as
\begin{equation}\label{exlicitCode}
    C \coloneq C_{\mathcal{L}}(D,A) = \{(x(P_1), ..., x(P_N)): x \in \mathcal{L}(A)\} \subseteq \mathbb{F}_q^N.
\end{equation}
Now consider any $x_1, x_2, x_3, x_4, x_5, x_6, x_7 \in \mathcal{L}(A)$, set $X \coloneq x_1x_2x_3x_4x_5x_6x_7$, and consider the principal adele $\mathrm{princ}(X)$ defined in Definition~\ref{def:principal_adele}. Recalling from Definition \ref{differentialDef} that Weil differentials always vanish on principal adeles, we have
\begin{equation}
    0 = \omega_0\left(\mathrm{princ}(X)\right).
\end{equation}
By further using Proposition \ref{localComponentFormula}, we have that
\begin{equation}
    0 = \sum_{P \in \mathbb{P}_F}\omega_{0P}(X)
\end{equation}
where $\omega_{0P}$ is the local component of $\omega_0$ at the place $P$, recalling from Definition~\ref{def:local_components_of_differentials} that the value of the principal adele $\mathrm{princ}(X)$ is $X$ at all places $P \in \mathbb{P}_F$. Then, by Definition~\ref{def:local_components_of_differentials}, we have
\begin{equation}
    0 = \sum_{P \in \mathbb{P}_F}\omega_0(\iota_P(X)).
\end{equation}
Now let us show:
\begin{claim}\label{iotaMembership}
    For any place $P \notin \mathcal{P}$, we have $\iota_P(X) \in \mathcal{A}_F(7A+E-D)$.
\end{claim}
\begin{proof}
    Let $P$ be a place such that $P \notin \mathcal{P}$. For any place $Q \neq P$, $\iota_P(X)$ takes a value $0$ at $Q$, and so $\nu_Q(\iota_P(X)) = \infty$ due to Lemma~\ref{PlaceValuation}; then, we have
    \begin{equation}
        \nu_Q(\iota_P(X)) + \nu_Q(7A+E-D)=\infty \geq 0.
    \end{equation}
    
    For the place $P$ itself, because $x_1, ..., x_7 \in \mathcal{L}(A)$, and using the second point of Definition \ref{discreteValuations}, we have $X \in \mathcal{L}(7A)$. We claim that, because $E \geq 0$ and $P \notin \mathcal{P} = \supp(D)$, we have
    \begin{equation}\label{adeleMembership}
        \nu_P(\iota_P(X)) + \nu_P(7A+E-D)\geq 0.
    \end{equation}
    Indeed, we recall from Definition \ref{def:local_components_of_differentials} that the adele $\iota_P(X)$ takes the value $X$ at the place $P$, so that $\nu_P(\iota_P(X)) = \nu_P(X)$. Then $X \in \mathcal{L}(7A)$ tells us that $\nu_P(\iota_P(X)) + \nu_P(7A) \geq 0$ by Definition~\ref{def:riemann_space}, and since $E \geq 0$, we have $\nu_P(E) \geq 0$ by Definition~\ref{def:positive_divisor}, and $P \notin \mathcal{P} = \supp(D) \implies \nu_P(-D)=0$ by Definition~\ref{def:divisor}, thus establishing Equation \eqref{adeleMembership}.
    
    We have thus shown that for every place $Q \in \mathbb{P}_F$, we have $\nu_Q(\iota_P(X)) + \nu_Q(7A+E-D) \geq 0$, and so, by Definition~\ref{def:adeles}, $\iota_P(X) \in \mathcal{A}_F(7A+E-D)$.
\end{proof}
From Equation \eqref{omega0Membership}, Claim \ref{iotaMembership}, and Definition \ref{differentialDef}, $\omega_0(\iota_P(X)) = 0$ for all $P \notin \mathcal{P}$, and we deduce that
\begin{align}
    0&=\sum_{i=1}^N\omega_0(\iota_{P_i}(X))\\
    &=\sum_{i=1}^N\omega_{0P_i}(X)
\end{align}
where $\omega_{0P_i}$ is the local component of $\omega_0$ at the place $P_i \in \mathcal{P}$. To proceed from here, we need the following fact from \cite{stichtenoth2009algebraic}.

\begin{proposition}[From the proof of Theorem 2.2.8 of \cite{stichtenoth2009algebraic}]
    Let $P \in \mathbb{P}_F$ be a rational place and $\omega$ a Weil differential with $\nu_P((\omega)) \geq -1$. If $x \in F$ has $\nu_P(x) \geq 0$, then
    \begin{equation}
        \omega_P(x) = x(P) \omega_P(1).
    \end{equation}
\end{proposition}
This is applicable in our case to places $P_i \in \mathcal{P}$ and $\omega_0$ because we established earlier that $\nu_P((\omega_0)) \geq -1$ for all $P \in \mathcal{P}''$ (see Equation \eqref{omega0Fact}). We note also that $x_1, ..., x_7 \in \mathcal{L}(A)$ and $\supp(A) \cap \mathcal{P}'' = \emptyset$ together imply $\nu_P(x_i) \geq 0$ for all $i=1, ..., 7$ and all $P \in \mathcal{P}''$ due to Definition~\ref{def:riemann_space}, which means that $\nu_P(X) \geq 0$ for all $P \in \mathcal{P}''$. These statements thus hold in particular for $P_i \in \mathcal{P}$. We therefore find
\begin{equation}
    0 = \sum_{i=1}^NX(P_i)\omega_{0P_i}(1).
\end{equation}
Now using the fact (see Definition~\ref{residueClassFieldAndMapDef}) that the residue class map at the place $P_i$ forms a ring homomorphism from $\mathcal{O}_{P_i}$ into the residue class field at $P_i$ (which is $\mathbb{F}_q$ in this case), we get
\begin{equation}
    0=\sum_{i=1}^Nx_1(P_i)x_2(P_i)x_3(P_i)x_4(P_i)x_5(P_i)x_6(P_i)x_7(P_i)\omega_{0P_i}(1)\text{ for all } x_1, ..., x_7 \in \mathcal{L}(A).\label{nearHomogeneous}
\end{equation}

From here, we make use of the following claim.

\begin{claim}\label{seventhRoots}
    For each $i=1, ..., N$, there exist $w_i \in \mathbb{F}_q\setminus\{0\}$ such that $w_i^7 = \omega_{0P_i}(1)$.
\end{claim}
\begin{proof}
    The multiplicative group of a finite field is cyclic. Therefore the multiplicative group of $\mathbb{F}_q$ is isomorphic to the cyclic group of order $q-1$, denoted $C_{q-1}$, which we will write multiplicatively. Recall that we have $q-1 = 2^s-1$ which is coprime to 7 for $s \not\equiv 0 \pmod 3$\footnote{This claim is exactly why we take $s \not\equiv 0 \pmod 3$ in Theorem \ref{technicalCodeTheorem}.}, and so in particular for $s=10$. It follows that the map \begin{equation}
        \phi:\begin{cases}
            C_{q-1} &\to C_{q-1}\\
            g &\mapsto g^7
        \end{cases}
    \end{equation}
    is a group automorphism. Therefore every element in $\mathbb{F}_q$ has a seventh-root. Note that because $\omega_{0P_i}(1) \neq 0$ for $i=1, ..., N$, we have $w_i \neq 0$ as well.
\end{proof}

We will use this to take seventh-roots of $\omega_{0P_i}(1)$ and absorb these roots into the components of the codespace in Equation \eqref{nearHomogeneous}, thus defining a new code $\tilde{C}$. For this, the following definition is useful.

\begin{definition}\label{componentwiseMult}
    Define an $\mathbb{F}_q$-linear map $m_w:\mathbb{F}_q^N \to \mathbb{F}_q^N$ by componentwise multiplication by the $w_i$, where the $w_i \in \mathbb{F}_q\setminus \{0\}$ are defined as in Claim \ref{seventhRoots}. Explicitly,
    \begin{equation}
        m_w(v)_i = w_iv_i.
    \end{equation}
    Given an $\mathbb{F}_q$-linear subspace $C \subseteq\mathbb{F}_q^N$, define $m_w(C)$ to be the image of $C$ under $m_w$. The space $m_w(C)$ is also an $\mathbb{F}_q$-linear subspace. Because $w_i\neq 0$ for all $i$, it is easy to check that $m_w : C \to m_w(C)$ forms an isomorphism of vector spaces.
\end{definition}

With this, consider the following.
\begin{definition}
\label{def:C_tilde}
    We define the code $\tilde{C}$ as
    \begin{equation}
        \tilde{C}\coloneq m_w(C) \subseteq \mathbb{F}_q^N,
    \end{equation}
    where $C$ is defined as in Equation~\eqref{exlicitCode}, and $m_w$ is as in Definition~\ref{componentwiseMult}.
\end{definition}
From Equation \eqref{exlicitCode}, Equation \eqref{nearHomogeneous}, Claim \ref{seventhRoots} and Definition \ref{componentwiseMult}, $\tilde{C}$, it follows that
\begin{equation}\label{seventhPoly}
    0 = \sum_{i=1}^N(y_1)_i(y_2)_i(y_3)_i(y_4)_i(y_5)_i(y_6)_i(y_7)_i \text{ for all } y_1, ..., y_7 \in \tilde{C}.
\end{equation}
Let us emphasise that all arithmetic in this expression has taken place over $\mathbb{F}_q$. While this expression will ultimately lead to us obtaining the first point in the triorthogonality definition, Equation \eqref{triorthogSecondDef1}, let us make some preparation for satisfying the second point, Equation \eqref{triorthogSecondDef2}. Because the divisor $A$ is positive, the element $1 \in \mathbb{F}_q$ is in the Riemann-Roch space $\mathcal{L}(A)$; $1 \in \mathcal{L}(A)$. The element $1 \in \mathcal{L}(A)$ takes the value $1 \in \mathbb{F}_q$ under the residue class map at every place\footnote{As is discussed before Definition 1.1.14 of \cite{stichtenoth2009algebraic}, given any place $P$, one has $\mathbb{F}_q \subseteq \mathcal{O}_P$ and $\mathbb{F}_q \cap P = \{0\}$. This means that the residue class map embeds a copy of $\mathbb{F}_q$ into the residue class field. The degree of the place is defined as the degree of the extension of the residue class field with respect to this embedding of $\mathbb{F}_q$. Therefore, if the place $P$ is rational, this embedding is exactly the residue class field. Thus, in particular, at every rational place on which $C$ is defined, $1 \in \mathcal{L}(A)$ takes the value 1 under the residue class map.}. This, therefore, means that $1^N$, i.e. all 1's vector of length $N$, is in the code $C$. In turn, this means that the vector $m_w(1^N)$ is in the code $\tilde{C}$. 

Note that $m_w(1^N)_i = w_i$. Applying this to Equation \eqref{seventhPoly} gives us
\begin{equation}\label{preRowTwo}
    0 = \sum_{i=1}^Nw_i^5(y_1)_i(y_2)_i\text{ for all } y_1, y_2 \in \tilde{C}.
\end{equation}

\subsection{Puncturing a Generator Matrix for the Code}\label{puncturingSection}

Let us now consider a generator matrix $\tilde{G}$ for $\tilde{C}$ in Definition~\ref{def:C_tilde}. Explicitly, $\tilde{G}$ is a matrix over $\mathbb{F}_q$ with dimensions $m \times N$ whose rows form a basis over $\mathbb{F}_q$ for $\tilde{C}$.
Here, $m$ is the dimension of the code $\tilde{C}$, which is also the dimension of the code $C$ since $m_w$ forms an isomorphism.
Under the conditions presented in Section~\ref{AGCodeConstruction}, in particular, under the condition~\eqref{hypothesis1}, it always follows from~\eqref{positivityCondition} that
\begin{align}
    \label{positivityCondition_a2}
    a<n'-2-g.
\end{align}
Due to Equations \eqref{positivityCondition_a}, \eqref{boundN}, and~\eqref{positivityCondition_a2}, it holds that
\begin{align}
    1 &\leq [a-(2g-2)]-g-3,\label{hypothesis2}\\
    a&<N.\label{hypothesis3}
\end{align}
Then, the condition~\eqref{positivityCondition_a} implies $g \leq a < N$ and $a \geq 2g-1$, and so from Theorem \ref{AGCodesParams}, we have
\begin{equation}
    m=a+1-g.
\end{equation}

Now, $g \leq a < N$ and $a \geq 2g-1$ just ensured that we had a definite value for $m = \dim C = \dim\tilde{C}$, but by Theorem \ref{AGCodesParams} these also ensure that the dual code to $C$, $C^\perp$, has distance
\begin{equation}
    d^\perp \geq a-(2g-2).
\end{equation}
We wish to show that the dual code to $\tilde{C}$, $(\tilde{C})^\perp$, has the same distance as $C^\perp$. We start by showing the following.
\begin{claim}
    $(\tilde{C})^\perp = m_{w^{-1}}(C^\perp)$ where $w^{-1} \in \mathbb{F}_q^N$ is the vector defined as $(w^{-1})_i = w_i^{-1}$.
\end{claim}
\begin{proof}
    Consider any $x \in m_{w^{-1}}(C^\perp)$ and $y \in \tilde{C}$. Write $x_i = w_i^{-1}X_i$ and $y_i = w_iY_i$ where $(X_i)_{i=1}^N$ and $(Y_i)_{i=1}^N$ are the components of vectors in $C^\perp$ and $C$ respectively. Then $\sum_{i=1}^Nx_iy_i = 0$ and so $m_{w^{-1}}(C^\perp) \subseteq (\tilde{C})^\perp$. The conclusion then follows from
    \begin{align}
        \dim (\tilde{C})^\perp &= N-\dim\tilde{C}\\
        &=N-\dim C\\
        &= \dim C^\perp\\
        &= \dim m_{w^{-1}}(C^\perp).
    \end{align}
\end{proof}
We note that the map $m_{w^{-1}}$, as well as being an isomorphism of vector spaces, preserves the Hamming weight of vectors, i.e., $|m_{w^{-1}}(v)| = |v|$ for any $v \in \mathbb{F}_q^N$ because $w^{-1}_i \neq 0$ for each $i=1, ..., N$. As such, the distance of $(\tilde{C})^\perp$ is the same as the distance of $C^\perp$, which is then at least $a-(2g-2)$.

Recall that in Equation~\eqref{eq:condition_k}, we have picked an integer $k$ satisfying
\begin{align}
\label{eq:condition_k_2}
    0 < k \leq [a-(2g-2)]-g-3,
\end{align}
so that in particular $k$ is less than the distance of $(\tilde{C})^\perp$. Because $\tilde{G}$ is a generator matrix for $\tilde{C}$, it is also a parity-check matrix for $(\tilde{C})^\perp$. Non-zero codewords of $(\tilde{C})^\perp$ therefore correspond exactly to non-trivial linear combinations of columns of $\tilde{G}$ that equate to zero. In particular, any $k$ columns of $\tilde{G}$ are linearly independent. By considering the submatrix of $\tilde{G}$ formed by its first $k$ columns, there must be some $k$ rows of the submatrix that are linearly independent. As such, via row operations over $\mathbb{F}_q$, $\tilde{G}$ may be rewritten into the following form.

\begin{figure}[h]
\begin{center}
\begin{tikzpicture}
\node at (0,0.5) {$\tilde{G} = $};
\node at (1.3, 0.6) {$1$};
\node at (3.7, -1.7) {$1$};
\node [rotate = -12] at (2.5,-0.5) {$\ddots$};
\node at (6.5,-0.5) {$G_1$};
\node at (6.5,-2.5) {$G_0$};
\node at (2.5,-2.5) {$\mathbf{0}$};
\draw (1,1) rectangle (9,-3);
\draw [-](4,1) -- (4,-3);
\draw [-](1,-2) -- (4,-2);
\draw [dashed](4,-2) -- (9,-2);
\draw [stealth-stealth](0.8,1)--(0.8,-2);
\draw [stealth-stealth](0.8,-2)--(0.8,-3);
\draw [stealth-stealth](1,1.2)--(4,1.2);
\draw [stealth-stealth](4,1.2)--(9,1.2);
\node at (0.5,-0.5) {$k$};
\node at (0.15, -2.5) {$m-k$};
\node at (2.5,1.5) {$k$};
\node at (6.5,1.5) {$n= N-k$};

\end{tikzpicture}
\end{center}
\end{figure}
In words, the upper-left $k \times k$ submatrix of $\tilde{G}$ is an identity matrix, and the lower-left $(m-k)\times k$ submatrix is all zeros. We denote the resulting upper-right $k \times (N-k)$ submatrix as $G_1$ and the lower-right $(m-k)\times (N-k)$ submatrix as $G_0$. We also define $n \coloneq N-k$. Finally, we define $G$ as the right-hand $m \times n$ matrix (formed of $G_1$ and $G_0$).

Let us denote the rows of $\tilde{G}$ as $(\tilde{g}^a)_{a=1}^m$. These rows form a basis for $\tilde{C}$, and so in particular lie in $\tilde{C}$. From Equations \eqref{seventhPoly} and \eqref{preRowTwo}, we therefore have
\begin{align}
    0 &= \sum_{i=1}^N(\tilde{g}^a_i)^4(\tilde{g}^b_i)^2(\tilde{g}^c_i)\\
    0 &= \sum_{i=1}^Nw_i^5\tilde{g}^a_i\tilde{g}^b_i
\end{align}
for all $a,b,c \in \{1, ..., m\}$. With this, it is immediate that $G$ is a triorthogonal matrix in the sense of Definition \ref{triorthogMatDefinition}. For each $a=1, ..., k$, let $\tau_a = w_a^5$, and for each $i=1, ..., n$, let $\sigma_i = w_{i+k}^5$. Letting the rows of $G$ be $(g^a)_{a=1}^m$, we have, for all $a,b,c \in \{1, ..., m\}$,
\begin{align}
    \sum_{i=1}^n(g_i^a)^4(g_i^b)^2(g_i^c) &= \begin{cases}
        1 &\text{ if } 1 \leq a=b=c \leq k\\
        0 &\text{ otherwise}
    \end{cases}\\
    \sum_{i=1}^n\sigma_ig_i^ag_i^b &= \begin{cases}
        \tau_a &\text{ if } 1 \leq a=b \leq k\\
        0 &\text{ otherwise}
    \end{cases}
\end{align}
and because $\sigma_i, \tau_a \neq 0$, we find that $G$ is a triorthogonal matrix.

Having constructed a triorthogonal matrix, we are concerned with the parameters of the quantum code constructed from it according to the quantum code definition of Section~\ref{quantumCodefromMatrix}.
From the argument in Section~\ref{quantumCodefromMatrix}, the dimension of the quantum code is manifest --- it is the value $k$ in the above equations, and showing $k = \Theta(n)$ only comes down to the values for the parameters that will be chosen in Section~\ref{concrete}.

The distance of the code, on the other hand, requires analysis. It was shown in Section \ref{quantumCodefromMatrix} that $d_X \geq d_Z$, so the distance of the quantum code is
\begin{equation}
    d = d_Z = \min_{f \in \mathcal{G}_0^\perp\setminus \mathcal{G}^\perp}|f|.
\end{equation}
We use the simple lower bound for this:
\begin{equation}
    d_Z \geq \min_{f \in \mathcal{G}_0^\perp \setminus \{0\}}|f|
\end{equation}
which is the distance of $\mathcal{G}_0^\perp$ as a classical code\footnote{We denoted this classical distance as $d_{\mathrm{cl}}$ in Section \ref{MSDSummarySection}.}. To demonstrate a lower bound on this, we prove the following.
\begin{lemma}\label{G0CShort}
    Recall that $(P_1, ..., P_N)$ are the set of places on which the algebraic geometry code is defined and these correspond, in order, to the columns of $\tilde{G}$. As in Definition~\ref{def:AG_code}, consider the algebraic geometry code
    \begin{equation}
    \label{eq:C_short}
        C_\mathrm{short} = C_{\mathcal{L}}(P_{k+1}+...+P_N,A-P_1-...-P_k).
    \end{equation}
    Further, let $\tilde{m}_w: \mathbb{F}_q^n \to \mathbb{F}_q^n$ be the map defined by componentwise multiplication by $(w_i)_{k+1}^N$, i.e. for $x \in \mathbb{F}_q^n$, $\tilde{m}_w(x)_i = w_{i+k}x_i$. Then $G_0$ is a generator matrix for the code $\tilde{m}_w(C_\mathrm{short})$.
\end{lemma}
\begin{proof}
    We first start by noting that if $x$ is in the Riemann-Roch space $\mathcal{L}(A)$, then $x(P_1) = ... = x(P_k) = 0$ if and only if $x \in \mathcal{L}(A-P_1-...-P_k)$, which is a subspace of $\mathcal{L}(A)$.

    Let us consider the evaluation map that defines the algebraic geometry code $C$:
    \begin{equation}
        \mathrm{ev}:\begin{cases}
            \mathcal{L}(A) &\to C = C_{\mathcal{L}}(D,A)\\
            x &\mapsto (x(P_1), ..., x(P_N)).
        \end{cases}
    \end{equation}
    Due to $N > a$ shown in Equation~\eqref{hypothesis3}, this map is an isomorphism (see the proofs of Theorem 2.2.2 and Corollary 2.2.3 of \cite{stichtenoth2009algebraic}). Moreover, if we restrict $\mathrm{ev}$ to the subspace $\mathcal{L}(A-P_1-....-P_k)$, it gives an isomorphism
    \begin{equation}\label{restrictedEvaluation}
        \mathrm{ev}|_{\mathcal{L}(A-P_1-...-P_k)}:\mathcal{L}(A-P_1-...-P_k) \to \{x \in C: x_1 = ... = x_k = 0\}.
    \end{equation}
    Indeed, the fact that $\mathrm{ev}|_{\mathcal{L}(A-P_1-...-P_k)}$ maps into this space, and surjectivity, follow from the first sentence of the proof, and injectivity is true because a restriction of an injective map is injective.

    Recall that $k$ has been chosen to be less than $a-(2g-2)$ as shown in Equation~\eqref{eq:condition_k_2}. Therefore, $a-k \geq 2g-1$, and so by the Riemann-Roch theorem (Theorem~\ref{riemannRochTheorem}),
    \begin{equation}
        l(A-P_1-...-P_k) = a-k+1-g = m-k.
    \end{equation}
    Thus, both spaces in Equation \eqref{restrictedEvaluation} have dimension $m-k$. The last $m-k$ rows of $\tilde{G}$ are linearly independent and in the space $\tilde{m}_w\left(\{x \in C: x_1 = ... = x_k = 0\}\right)$, and thus form a basis for this space. By deleting the first $k$ zeros, this space is isomorphic to $\tilde{m}_w(C_\mathrm{short})$, and so the rows of $G_0$ form a basis for this space, thus establishing the lemma.
\end{proof}
Given that $G_0$ is a generator matrix for the classical code $\tilde{m}_w(C_\mathrm{short})$, we in fact have $\mathcal{G}_0 = \tilde{m}_w(C_\mathrm{short})$ and therefore
\begin{equation}
    \mathcal{G}_0^\perp = (\tilde{m}_w(C_\mathrm{short}))^\perp = \tilde{m}_{w^{-1}}(C_\mathrm{short}^\perp)
\end{equation}
where $\tilde{m}_{w^{-1}}:\mathbb{F}_q^n\to\mathbb{F}_q^n$ is the linear map enacting componentwise multiplication by $(w_{i+k}^{-1})_{i=1}^n$. The distance of the classical code $\mathcal{G}_0^\perp$ is therefore the same as the distance of $C_\mathrm{short}^\perp$, since $w_{i+k} \neq 0$ for $i=1, ..., n$. We can show that all the hypotheses of Theorem \ref{AGCodesParams} apply to $C_\mathrm{short}$. Indeed, we had $a < N$ and so $a-k < N-k$. Moreover, we chose $k \leq a-(2g-2)-g-3$ and so $a-k \geq 2g-1$ and $a-k \geq g$. Therefore Theorem \ref{AGCodesParams} applies in full, and we deduce the distance of $C_\mathrm{short}^\perp$, and therefore that of $\mathcal{G}_0^\perp$, as at least $a-k-(2g-2)$. This means that the distance of the quantum code is
\begin{equation}
\label{eq:distance_bound}
    d = d_Z \geq a-k-(2g-2)
\end{equation}
which will turn out to be $\Theta(n)$ in the following section.

Lastly, let us discuss the efficient classical decoding of the code $\mathcal{G}_0^\perp$, which is useful in the constant-overhead magic state distillation protocol discussed in Section \ref{MSDSummarySection}. Specifically, it is best, for the practical implementation of the protocols, to be able to efficiently decode the classical code $\mathcal{G}_0^\perp$ from a number of errors $t = \Theta(d_{\mathrm{cl}})$, where $d_{\mathrm{cl}}$ is the (classical) distance of the code $\mathcal{G}_0^\perp$, which we recall was for the correction of $Z$-errors. From Lemma \ref{G0CShort}, we know that the code $\mathcal{G}_0$ is (up to componentwise multiplication by known non-zero field elements) the algebraic geometry code $C_\mathrm{short}$ in Equation~\eqref{eq:C_short}.
We, therefore, know that $\mathcal{G}_0^\perp$ is (up to componentwise multiplication by known non-zero field elements) $C_\mathrm{short}^\perp$.
We then refer to the details of efficient decoding of such dual algebraic geometry codes in Section \ref{AGPrelims3}.
We let
\begin{align}
\label{eq:A_1}
    A_1 \in \mathrm{Div}(F)
\end{align}
be a divisor supported at only one place --- a rational place not in the set $\{P_{k+1}, ..., P_N\}$, so that $\supp A_1 \cap \supp (P_{k+1}+...+P_N) = \emptyset$.
The decoding radius $t$ for the efficient decoder is determined by whether $t$ and $A_1$ meet the conditions for efficient classical decoding in Equations \eqref{decodingConditionFirst} to \eqref{decodingConditionLast} in Section \ref{AGPrelims3}
\begin{align}
\label{eq:condition_A_1_1}
     \deg(A_1) &< \deg(A-P_1-...-P_k)-(2g-2) - t=a-k-(2g-2)-t
\end{align}
and
\begin{equation}
\label{eq:condition_A_1_2}
    l(A_1) \geq \deg(A_1)+1-g > t
\end{equation}
where the first inequality is true by Riemann's theorem --- see the paragraph after Theorem \ref{riemannRochTheorem}.
From these inequalities, it follows that there exists such an integer $\deg(A_1)$ if $t$ is chosen as
\begin{align}
\label{eq:t_bound}
    0<t\leq\frac{[a-k-(2g-2)]-g-1}{2}.
\end{align}
Note that under the condition~\eqref{eq:condition_k}, i.e., $[a-k-(2g-2)]-g\geq 3$, such an integer $t$ always exists. Moreover, in the following section, it will turn out that it is possible to achieve $t=\Theta(n)$.

\subsection{A Concrete Construction}\label{concrete}

The aim of this section is to make concrete the construction of the triorthogonal matrix by assigning values to each parameter. With these values, each of the hypotheses used in the above construction will be satisfied, and the final parameters of the quantum code will be asymptotically good, i.e. $k,d = \Theta(n)$. Our aim for now is only to show that this can be achieved, and not to optimise the parameters up to constants.

Our function fields were built over the field $\mathbb{F}_q$ with $q=2^s$, where our primary case is $s=10$, i.e., $q=2^{10} = 1024$, and we know in this case that the Ihara constant is
\begin{equation}
    A(1024) = 31.
\end{equation}
Therefore, there is a sequence of function fields over $\mathbb{F}_{2^{10}}$ whose genera $g$ diverge to infinity, and whose number of rational places is at least $(31-\delta)g$ for any $\delta>0$.
We may thus take the initial number of rational places employed to be
\begin{align}
\label{eq:n_prime_concrete}
    n' = \left\lfloor\frac{123}{4}g\right\rfloor. 
\end{align}
For any such sequence of function fields, we define a sequence of codes, and in each, we pick the following values:
\begin{align}
\label{eq:a_concrete}
    a &= \left\lfloor\frac{9}{2}g\right\rfloor;\\
\label{eq:k_concrete}
    k &= \left\lfloor\frac{5}{4}g\right\rfloor.
\end{align}

Let us show how these satisfy all the hypotheses of our construction and lead to asymptotically good quantum codes.
It is clear that the choices of $n^\prime$, $a$, and $k$ in Equations~\eqref{eq:n_prime_concrete},~\eqref{eq:a_concrete}, and~\eqref{eq:k_concrete} satisfy the conditions~\eqref{hypothesis1},~\eqref{positivityCondition_a},~\eqref{positivityCondition}, and~\eqref{eq:condition_k} for every sufficiently large $g$.
We then set $n'' = n'-2$ and find a set of places defining the algebraic geometry code $C$ of size $N$, where $n'' \geq N \geq n'' - g$, i.e.,
\begin{equation}
    \left\lfloor\frac{123}{4}g\right\rfloor -2 \geq N \geq \left\lfloor\frac{119}{4}g\right\rfloor -2.
\end{equation}
Because $n=N-k$, we have
\begin{equation}
    \left\lfloor\frac{118}{4}g\right\rfloor -1 \geq n \geq \left\lfloor\frac{114}{4}g\right\rfloor-3,
\end{equation}
and so in particular
\begin{align}
    n = \Theta(g)
\end{align}
as $g \to \infty$.
Due to Equation~\eqref{eq:distance_bound}, the distance of our quantum code is then
\begin{align}
    d=d_Z &\geq \left\lfloor\frac{9}{2}g\right\rfloor - \left\lfloor\frac{5}{4}g\right\rfloor -2g+2\\
    &\geq \left\lfloor\frac{5}{4}g\right\rfloor+1.
\end{align}
Therefore, as $g\to\infty$, we have
\begin{align}
    k=\Theta(g)=\Theta(n),\\
    d=\Theta(g)=\Theta(n),
\end{align}
as claimed.

The decoding radius $t$ for the efficient decoder can be determined as follows.
For the divisor $A_1$ discussed in Equation~\eqref{eq:A_1},
we set
\begin{align}
    \deg(A_1) &= \left\lfloor \frac{9}{8}g\right\rfloor,\\
    t &=\left\lfloor\frac{1}{8}g\right\rfloor.
\end{align}
Then, the conditions shown in Equations~\eqref{eq:condition_A_1_1},~\eqref{eq:condition_A_1_2}, and~\eqref{eq:t_bound} are satisfied.
Thus, we have
\begin{align}
    t=\Theta(g)=\Theta(n),
\end{align}
as $g\to\infty$.

In summary, given any sequence of function fields over $\mathbb{F}_{1024}$ with diverging genera, and whose number of rational places exceeds 30 times their genera, we can construct asymptotically good triorthogonal quantum codes over $\mathbb{F}_{1024}$, where the bounds on the parameters follow from Theorem \ref{AGCodesParams}.
Let us give an explicit instantiation of our asymptotically good triorthogonal quantum codes.

\begin{example}
    In \cite{garcia1995tower}, a sequence of function fields over $\mathbb{F}_q$, where $q=l^2$ is a square prime power, is presented as $F_1\coloneq\mathbb{F}_q(x_1)$ and $F_{n+1}\coloneq F_n(z_{n+1})$ for $n\geq 1$, where $z_{n+1}$ satisfies the equation $z_{n+1}^l+z_{n+1}=x_{n}^{l+1}$ with $x_n\coloneqq z_n/x_{n-1}\in F_n$ for $n\geq 2$. The sequence of function fields is optimal in the sense that its number of rational places to genera tends to $l-1$, thus meeting the Ihara constant $A(q) = l-1$. For each $i \geq 3$, the sequence of function fields $F_i$ have genera $g_i$ (see Theorem 2.10 of \cite{garcia1995tower}) and number of rational places $N^{(1)}_i$ (Proposition 3.1 of \cite{garcia1995tower}) satisfying
    \begin{equation}
        g_i = \begin{cases}
            l^i + l^{i-1} - l^{\frac{i+1}{2}}-2l^{\frac{i-1}{2}}+1, &\text{ if } i \equiv 1 \pmod 2\\
            l^i + l^{i-1} - \frac{1}{2}l^{\frac{i}{2}+1}-\frac{3}{2}l^{\frac{i}{2}} - l^{\frac{i}{2}-1}+1, &\text{ if } i \equiv 0 \pmod 2
        \end{cases}
    \end{equation}
    and
    \begin{equation}
        N^{(1)}_i \geq (l^2-1)l^{i-1}+2l.
    \end{equation}
    For every $i \geq 3$, we have $N_i^{(1)} > 31g_i$, and $g_i$ is large enough that the conditions~\eqref{hypothesis1},~\eqref{positivityCondition_a},~\eqref{positivityCondition}, and~\eqref{eq:condition_k} are satisfied. Thus, every member of this family of function fields may be used in our construction. We apply this sequence of function fields to the above construction of $[[n_i,k_i,d_i]]$ triorthogonal quantum codes, where the number of physical qudits $n_i$ satisfies $\left\lfloor \frac{123}{4}g_i\right\rfloor - k_i - 2 \geq n_i \geq \left\lfloor \frac{119}{4}g_i\right\rfloor - k_i - 2$, the number of logical qudits is $k_i=\lfloor \frac{5}{4}g_i\rfloor$, and the distance satisfies $d_i\geq \lfloor \frac{5}{4}g_i\rfloor+1$; we also take the decoding radius as $t_i=\lfloor \frac{1}{8}g_i\rfloor$.
    With $q=1024$, i.e, $l=32$, we then obtain a family of $[[n_i,k_i,d_i]]$ triorthogonal quantum codes on $1024$-dimensional qudits with decoding radius $t_i$ satisfying, for all $i\in\{3,5,7,\ldots\}$,
    \begin{align}
        \left\lfloor \frac{118}{4}\left(33\times 32^{i-1}-34\times 32^{\frac{i-1}{2}}+1\right)\right\rfloor -1 \geq \; &n_i \geq \left\lfloor \frac{114}{4}\left(33\times 32^{i-1}-34\times 32^{\frac{i-1}{2}}+1\right)\right\rfloor-3,\\
        &k_i=\left\lfloor\frac{5}{4}\left(33\times 32^{i-1}-34\times 32^{\frac{i-1}{2}}+1\right)\right\rfloor,\\
        &d_i\geq\left\lfloor\frac{5}{4}\left(33\times 32^{i-1}-34\times 32^{\frac{i-1}{2}}+1\right)\right\rfloor+1,\\
        &t_i=\left\lfloor\frac{1}{8}\left(33\times 32^{i-1}-34\times 32^{\frac{i-1}{2}}+1\right)\right\rfloor,
    \end{align}
    and for all $i\in\{4,6,8,\ldots\}$,
    \begin{align}
         \left\lfloor \frac{118}{4}\left(33 \times 32^{i-1}-561\times 32^{\frac{i}{2}-1}+1\right)\right\rfloor -1 \geq \; &n_i \geq \left\lfloor \frac{114}{4}\left(33 \times 32^{i-1}-561\times 32^{\frac{i}{2}-1}+1\right)\right\rfloor -3 ,\\
        &k_i=\left\lfloor \frac{5}{4}\left(33 \times 32^{i-1}-561\times 32^{\frac{i}{2}-1}+1\right)\right\rfloor,\\
        &d_i\geq\left\lfloor \frac{5}{4}\left(33 \times 32^{i-1}-561\times 32^{\frac{i}{2}-1}+1\right)\right\rfloor+1,\\
        &t_i=\left\lfloor \frac{1}{8}\left(33 \times 32^{i-1}-561\times 32^{\frac{i}{2}-1}+1\right)\right\rfloor.
    \end{align}
    We emphasise that the code parameters in this example are not optimised, and moreover this is only one example of a possible function field --- many other examples may be considered. Our contribution is to provide a general recipe to construct the asymptotically good triorthogonal codes on prime-power dimensional qudits to make such optimisation possible in future work, and importantly, our magic state distillation protocol can arbitrarily suppress the error rate of magic states with constant overhead using such families of asymptotically good triorthogonal codes.
\end{example}

We can finally comment on why we have picked $q=2^{10}$ for our main construction throughout the paper. It is clear that we wish to pick $q=2^s$ for some $s$ to ultimately obtain a qubit code. For each chosen value of $s$, the inequality~\eqref{positivityCondition}, i.e., $n'-4+g-7a \geq 0$, places the most restrictive lower bound on the number of rational places $n'$ in our function fields. It was necessary for us to pick the parameter $a$ with (at the very least) $a/g$ being a constant exceeding 3 so that the decoding radius $t$ in Equation~\eqref{eq:t_bound} may equal $\Theta(g)$ as $g$ increases\footnote{It is interesting to note that if one does not require the efficient decoding --- just the construction of the asymptotically good triorthogonal code --- one needs only $a/g > 2$ for the distance $d$ in Equation~\eqref{eq:distance_bound} to grow as $g$ increases, which in fact allows a construction over the field $\mathbb{F}_{2^8} = \mathbb{F}_{256}$.}. Putting these together, we therefore require an $s$ for which the Ihara constant is $A(2^s) > 20$. From Theorem 1.1 of \cite{bassa2012towers}, it in fact follows that $A(2^9) = 20.217\ldots$; however, $s=9$ is not usable because of the requirement that $s \not\equiv 0 \pmod 3$, which came from the need to be able to take seventh-roots in our field $\mathbb{F}_q$ in Claim \ref{seventhRoots}. 

There were no further restrictions on the value of $s$ in the paper --- we note that Lemma \ref{gatelemma} was proved for $s \geq 5$ which is certainly implied by the restriction $s \geq 10$ of Theorem \ref{technicalCodeTheorem}. We know that $A(q) = \sqrt{q}-1$ for square prime powers $q$, and that $A(2^{2a+1}) > 2^a-1$ for integers $a \geq 1$ from Theorem 1.1 of \cite{bassa2012towers}. Thus, for all values of $s$ which we consider in Theorem \ref{technicalCodeTheorem}, we in fact have $A(2^s) \geq 31$. This establishes Theorem \ref{technicalCodeTheorem}, and so in particular Theorem \ref{quantumCodeTheorem}.

\section{Analysis of Error Suppression and Overhead}\label{errorAnalysis}

In this section, we analyse the error and the overhead of the distillation protocol presented in this paper.
We consider the task of magic state distillation under the local stochastic error model described in Section~\ref{sec:task}.
As described in Sections~\ref{MSDProtocolMainSection}, our distillation protocol, using $[[n,k,d]]$ triorthogonal codes on $q$-dimensional qudits with $q=2^s$, starts with $\zeta=Cn$ noisy $CCZ$ states $\ket{CCZ}$ at physical error rate $p_\mathrm{ph}$, converts them into $n$ noisy magic states $\ket{M}$ on qudits, distils $k$ magic states $\ket{M}$, and converts them into $l=k$ $CCZ$ states.
In particular, as presented in Section~\ref{MSDProtocolMainSection}, our protocol is designed for $s=10$, $C=O(1)$, $k=\Theta(n)$, and $d=\Theta(n)$.
The efficient decoder of the algebraic geometry codes (i.e., of our triorthogonal codes) has the decoding radius of $t=\Theta(n)$ and thus can correct any $Z$-type errors (which are relevant for magic state distillation) up to weight $t$ on the $n$ qudits.
Note that our protocol does not involve post-selection since we use the decoder for error correction rather than error detection (see also Section~\ref{MSDProtocolMainSection}).

In this case, we show the following bounds on errors for the task of magic state distillation defined in Section~\ref{sec:task}.
\begin{theorem}
\label{thm:error_rate}
Consider a family of protocols for magic state distillation using $[[n,k,d]]$ codes for $q$-dimensional qudits with decoding radius $t$, which starts with $\zeta=Cn$ noisy $CCZ$ states $\ket{CCZ}$ at physical error rate $p_\mathrm{ph}$, converts them into $n$ noisy magic states $\ket{M}$ on qudits, distils $k$ magic states $\ket{M}$ without post-selection, and converts them into at least $\xi(\zeta)=k$ $CCZ$ states.
Under the local stochastic error model,
if the codes used in such a protocol satisfy
\begin{align}
    k&=\Theta(n),\\
    t&=\Theta(n),
\end{align}
then the protocol has a constant threshold $p_\mathrm{th}>0$ such that if the physical error rate of the $\zeta$ initial $CCZ$ states is below the threshold, i.e., $p_\mathrm{ph}<p_\mathrm{th}$, then the protocol outputs $\xi(\zeta)$ final $CCZ$ states achieving a target error rate $\epsilon$ exponentially suppressed as
\begin{align}
\label{eq:epsilon}
\epsilon=\left(\frac{p_\mathrm{ph}}{p_\mathrm{th}}\right)^{O(\zeta)}=\exp(-O(\zeta))
\end{align}
within a constant overhead
\begin{align}
\label{eq:overhead}
\frac{\zeta}{\xi(\zeta)}=O(1)
\end{align}
as $\zeta\to\infty$.
\end{theorem}

\begin{remark}
Even if we choose $\ket{T}$ as the target magic state to be distilled, the same bounds as~\eqref{eq:epsilon} and~\eqref{eq:overhead} for $\ket{CCZ}$ also hold for $\ket{T}$ up to modification of the constant factors.
To see this, we notice that $\ket{T}$ and $\ket{CCZ}$ can be transformed into each other exactly by stabiliser operations at finite conversion rates~\cite{beverland2020lower}.
In particular, we can convert four $\ket{T}$ into $\ket{CCZ}$ using the protocol in~\cite{PhysRevA.87.022328,PhysRevA.87.042302}; conversely, we can use a protocol in~\cite{Gidney2019efficientmagicstate} to transform $\ket{CCZ}\otimes\ket{T}$ into $\ket{T}^{\otimes 3}$, using a single copy of $\ket{T}$ as a catalyst.
For the latter catalytic transformation, the initial single copy of $\ket{T}$ needs to be distilled by conventional protocols for magic state distillation with a polylogarithmic overhead, e.g., by the protocols in~\cite{bravyi2005universal,bravyi2012magic,hastings2018distillation}.
However, this additional overhead of distilling the single $\ket{T}$ can be made negligible if we distil many $\ket{T}$s using our constant-overhead protocol for distilling $\ket{CCZ}$.
Hence, the choice of target magic states, $\ket{CCZ}$ or $\ket{T}$, only affects the constant factors in Theorem~\ref{thm:error_rate}.
\end{remark}

To prove Theorem~\ref{thm:error_rate}, we analyze the error rate at each step of our protocol.
Our protocol works on $\zeta=Cn$. Using $C$ copies of initial $\ket{CCZ}$ at physical error rate $p_{\mathrm{ph}}$, our protocol prepares a single copy of $\ket{M}$.
At this step, due to the union bound, the error rate of the qudit preparation operation to prepare a single-qudit state $\ket{M}$ is upper bounded by $Cp_\mathrm{ph}$, where the preparation operations of $\ket{M}$ undergo the local stochastic error model since the preparation operations of $\ket{CCZ}$ do.
Then, using the $[[n,k,d]]$ triorthogonal code, from $n$ copies of $\ket{M}$ in the previous step, our protocol prepares $\ket{M}^{\otimes k}$.
At this step, any errors with weight at most $t$ on $n$ qudits can be corrected; i.e., due to the union bound, the error rate of the preparation operation of $\ket{M}^{\otimes k}$ is upper bounded by the sum of probabilities of having $t+1$ (or more) errors, i.e., $\binom{n}{t+1}(Cp_\mathrm{ph})^{t+1}$.
Finally, our protocol prepares $\ket{CCZ}^{\otimes k}$ from $\ket{M}^{\otimes k}$, with the error rate upper bounded by $\binom{n}{t+1}(Cp_\mathrm{ph})^{t+1}$ due to the previous step.

To summarize, if we have $\zeta=Cn$ initial $CCZ$ states, our protocol outputs $k$ final $CCZ$ states and achieves the target error rate
\begin{align}
    \epsilon&\leq\binom{n}{t+1}(Cp_\mathrm{ph})^{t+1}\\
    &\leq 2^{nh\left(\frac{t+1}{n}\right)}(Cp_\mathrm{ph})^{t+1}\\
    \label{eq:error_bound}
    &=\left(\frac{p_\mathrm{ph}}{1/C2^{\frac{n}{t+1}h\left(\frac{t+1}{n}\right)}}\right)^{t+1},
\end{align}
where $h(x)\coloneqq-x\log_2(x)-(1-x)\log_2(1-x)$ is the binary entropy function, and the second inequality follows from $\binom{n}{k}\leq 2^{nh(k/n)}$.
Then, due to $t=\Theta(n)$, there exists a threshold
\begin{align}
    p_\mathrm{th}\geq\liminf_{n\to\infty}\frac{1}{C2^{\frac{n}{t+1}h\left(\frac{t+1}{n}\right)}}>0
\end{align}
such that if $p_{\mathrm{ph}}<p_\mathrm{th}$, then it holds for sufficiently large $n$ that
\begin{align}
\label{eq:epsilon_Cn}
    \epsilon\leq\left(\frac{p_\mathrm{ph}}{p_\mathrm{th}}\right)^{t+1}=\exp(-O(\zeta)),
\end{align}
where the right-hand side follows from $t=\Theta(n)$, $\zeta=Cn$, and $C=O(1)$.
The constant overhead follows from
\begin{align}
\label{eq:overhead_Cn}
    \frac{\zeta}{\xi(\zeta)}=\frac{Cn}{k}=O(1),
\end{align}
where we use $k=\Theta(n)$ and $C=O(1)$ on the right-hand side.

Lastly, we remark that the linear decoding radius $t=\Theta(n)$ (hence, the linear distance) is essential for achieving the constant-overhead magic state distillation with a constant threshold.
If $t/n$ vanishes as $n$ increases (even if $k/n$ is non-vanishing and $t$ grows), we would need code concatenation (i.e., repeatedly using a fixed code in many rounds) as in the conventional protocols for magic state distillation~\cite{bravyi2005universal,bravyi2012magic,hastings2018distillation}, which incurs polylogarithmic overhead.
With the code concatenation, if we had a code family with $k/n\to 1$ like the quantum Hamming codes~\cite{yamasaki2024time}, then one might be able to achieve the constant overhead with error suppression; however, the rate $k/n$ of triorthogonal codes is upper bounded by $1/2$ and thus can never approach to $1$ at least in the qubit case~\cite{nezami2022classification}, and we do not know how to achieve $k/n\to 1$ also in the qudit case.
By contrast, our error bound~\eqref{eq:error_bound} shows that the linearly growing $t$ makes it possible to prove the existence of a constant threshold for our single-round protocol using codes of growing sizes.
Thus, with this single-round protocol design, our construction of asymptotically good triorthogonal codes has led to the constant-overhead magic state distillation.
We emphasise that we do not aim to optimise the constant factors of the theoretical bounds, so one would need numerical simulation to estimate and optimise the overhead and the threshold in practice as we have discussed in Section~\ref{sec:discussion}, but our contribution is to provide the general recipe and rigorous proofs to make such optimisation of constant-overhead magic state distillation possible.

\section*{Acknowledgements}

All authors extend their thanks to N.\ Rengaswamy for discussions on CSS-T codes. A.\ Wills thanks The University of Tokyo, in particular the whole group of H.\ Yamasaki, for their funding of and hospitality during his visit to the institution, which led to this work.
H.\ Yamasaki was supported by JST PRESTO Grant Number JPMJPR201A, JPMJPR23FC, JSPS KAKENHI Grant Number JP23K19970, and MEXT Quantum Leap Flagship Program (MEXT QLEAP) JPMXS0118069605, JPMXS0120351339\@.

\bibliographystyle{unsrt}
\bibliography{references}

\pagebreak

\appendix

\section{Derivation of the Multinomial Formula}\label{multinomialDerivation}

In this section, we provide a proof of a formula used in Section \ref{quantumCodefromMatrix} that was crucial to demonstrate the transversality of the non-Clifford gate $U=U_1^{(7)}$. This formula is as given in the following lemma.
\begin{lemma}\label{multinomialLemma}
    Let $y_1, ..., y_m \in \mathbb{F}_q$, where $q$ is a power of two. Then, with all arithmetic taking place in $\mathbb{F}_q$, we have
    \begin{equation}
        \left(\sum_{a=1}^my_a\right)^7 = \sum_{a=1}^my_a^7 + \sum_{a \neq b}\left[y_a^6y_b + y_a^5y_b^2 + y_a^4y_b^3\right] + \sum_{\substack{a,b,c\\\text{pairwise distinct}}}y_a^4y_b^2y_c.
    \end{equation}
\end{lemma}
Note that this formula can be seen quickly by calculating all multinomial coefficients relevant to the present case and observing that the ones that are odd are exactly the ones corresponding to the terms shown. However, let us provide an inductive proof now for completeness.
\begin{proof}
    We will prove formulae for similar expressions $\left(\sum_{a=1}^my_a\right)^i$ for $i=2, \ldots, 7$ one-by-one to prove the above. For $i=2$, it is well known that
    \begin{equation}
        \left(\sum_{a=1}^my_a\right)^2 = \sum_{a=1}^my_a^2
    \end{equation}
    because the field has characteristic 2, and indeed this follows from Equation \eqref{char2Binomial}. Then, for $i=3$, we have
    \begin{equation}
        \left(\sum_{a=1}^my_a\right)^3 = \sum_{a=1}^my_a^3 + \sum_{a \neq b}y_a^2y_b.
    \end{equation}
    One may prove this by induction on $m$; first one notes the validity of the statement for $m=1$ and $m=2$ (note for $m=2$ that all binomial coefficients $\begin{pmatrix} 3\\k\end{pmatrix}$ are odd). For the general case, we have
    \begin{align}
        \left(\sum_{a=1}^my_a+y_{m+1}\right)^3 &= \left(\sum_{a=1}^my_a\right)^3 + \left(\sum_{a=1}^my_a\right)^2y_{m+1} + \sum_{a=1}^my_ay_{m+1}^2 + y_{m+1}^3\\
        &= \sum_{a=1}^my_a^3 + \sum_{\substack{a,b=1\\a \neq b}}^my_a^2y_b + \sum_{a=1}^my_a^2y_{m+1}+\sum_{a=1}^my_ay_{m+1}^2 + y_{m+1}^3\\
        &= \sum_{a=1}^{m+1}y_a^3 + \sum_{\substack{a,b=1\\a \neq b}}^{m+1}y_a^2y_b
    \end{align}
    where in the second equality we use the inductive hypothesis and the statement for $i=2$. The case of $i=4$,
    \begin{equation}
        \left(\sum_{a=1}^my_a\right)^4 = \sum_{a=1}^my_a^4,
    \end{equation}
    follows readily from the statement for $i=2$. Then, for $i=5$, we have
    \begin{equation}
        \left(\sum_{a=1}^my_a\right)^5 = \sum_{a=1}^my_a^5 + \sum_{a \neq b}y_a^4y_b.
    \end{equation}
    Again, this is proved by induction on $m$. The statement is noted to be true for $m=1$ and $m=2$, where the case of $m=2$ follows from the fact that $\begin{pmatrix}
        5\\k
    \end{pmatrix}$ is odd for $k=0,1,4,5$ and even for $k=2,3$. For the general case, we have
    \begin{align}
        \left(\sum_{a=1}^my_a+y_{m+1}\right)^5 &= \left(\sum_{a=1}^my_a\right)^5 + \left(\sum_{a=1}^my_a\right)^4y_{m+1} + \sum_{a=1}^my_ay_{m+1}^4 + y_{m+1}^5\\
        &=\sum_{a=1}^my_a^5 + \sum_{\substack{a,b=1\\a \neq b}}^my_a^4y_b + \sum_{a=1}^my_a^4y_{m+1}+\sum_{a=1}^my_ay_{m+1}^4 + y_{m+1}^5\\
        &= \sum_{a=1}^{m+1}y_a^5 + \sum_{\substack{a,b=1\\a \neq b}}^{m+1}y_a^4y_b
    \end{align}
    where in the second equality we apply the inductive hypothesis and the statement for $i=4$. The case of $i=6$ is then
    \begin{equation}
        \left(\sum_{a=1}^my_a\right)^6 = \sum_{a=1}^my_a^6 + \sum_{a \neq b}y_a^4y_b^2
    \end{equation}
    which then follows from the cases of $i=2$ and $i=3$. The lemma is then proved by showing the statement ($i=7$) by induction on $m$. First, one notes that the statement is true for $m=1$, $m=2$, and $m=3$ by direct calculation. Then,
    \begin{alignat}{2}
        \left(\sum_{a=1}^my_a+y_{m+1}\right)^7 &=&& \left(\sum_{a=1}^my_a\right)^7 + \left(\sum_{a=1}^my_a\right)^6y_{m+1}+\left(\sum_{a=1}^my_a\right)^5y_{m+1}^2\;+\\ &&&\left(\sum_{a=1}^my_a\right)^4y_{m+1}^3 + \left(\sum_{a=1}^my_a\right)^3y_{m+1}^4 + \left(\sum_{a=1}^my_a\right)^2y_{m+1}^5 + \sum_{a=1}^my_ay_{m+1}^6 + y_{m+1}^7\\
        &=&&\sum_{a=1}^my_a^7 + \sum_{\substack{a,b=1\\a \neq b}}^m\left(y_a^6y_b + y_a^5y_b^2 + y_a^4y_b^3\right) + \sum_{\substack{a,b,c=1\\\text{pairwise distinct}}}^my_a^4y_b^2y_c \;+ \\&&&\left(\sum_{a=1}^my_a^6 + \sum_{\substack{a,b=1\\a\neq b}}^my_a^4y_b^2\right)y_{m+1} + \left(\sum_{a=1}^my_a^5 + \sum_{\substack{a,b=1\\a\neq b}}^my_a^4y_b\right)y_{m+1}^2 + \sum_{a=1}^my_a^4y_{m+1}^3\;+\\
        &&&\left(\sum_{a=1}^my_a^3 + \sum_{\substack{a,b=1\\a\neq b}}^my_a^2y_b\right)y_{m+1}^4 + \sum_{a=1}^my_a^2y_{m+1}^5 + \sum_{a=1}^my_ay_{m+1}^6 + y_{m+1}^7\\
        &=&&\sum_{a=1}^{m+1}y_a^7 + \sum_{\substack{a,b=1\\a \neq b}}^{m+1}\left[y_a^6y_b+y_a^5y_b^2+y_a^4y_b^3\right] + \sum_{\substack{a,b,c=1\\\text{pairwise distinct}}}^{m+1}y_a^4y_b^2y_c
    \end{alignat}
    where in the second equality we have used the inductive hypothesis for the present case $i=7$ as well as the cases of $i=2,3,4,5$ and 6. By induction, we have thus proved Lemma \ref{multinomialLemma}.
\end{proof}

\section{An Example Decomposition of the 10-Qubit Gate}\label{exampleDecomposition}

We give an example decomposition of our $10$-qubit gate $\theta(U)$ from Section \ref{MSDfromCode} into $Z$, $CZ$ and $CCZ$ gates as discussed in Section \ref{CCZConversion}, where
\begin{equation}
    U = U_1^{(7)} = \sum_{\gamma \in \mathbb{F}_{1024}}\exp\left[i\pi\tr(\gamma^7)\right]\ket{\gamma}\bra{\gamma}.
\end{equation}
We consider the self-dual basis given in Equations \eqref{sdbFirst} to \eqref{sdbLast}, which was
\begin{align}
\alpha_0 &= 1+\alpha^2+\alpha^4+\alpha^5+\alpha^7+\alpha^8,\\
\alpha_1 &= \alpha^3+\alpha^6+\alpha^7+\alpha^8+\alpha^9,\\
\alpha_2 &= \alpha+\alpha^2+\alpha^5+\alpha^7 + \alpha^8+\alpha^9,\\
\alpha_3 &= 1+\alpha+\alpha^2+\alpha^3+\alpha^4+\alpha^6+\alpha^7+\alpha^8+\alpha^9,\\
\alpha_4 &= 1+\alpha+\alpha^4+\alpha^5+\alpha^7+\alpha^9,\\
\alpha_5 &= \alpha + \alpha^2+\alpha^3 + \alpha^7,\\
\alpha_6 &= \alpha^2+\alpha^6+\alpha^7,\\
\alpha_7 &= \alpha^2 + \alpha^5 + \alpha^7,\\
\alpha_8 &= 1+\alpha^3 + \alpha^7,\\
\alpha_9 &= 1 + \alpha^4 + \alpha^6 + \alpha^7,
\end{align}
where we recall $\alpha \in \mathbb{F}_{1024}$ satisfies $\alpha^{10} = 1+\alpha^3$. With respect to this basis, $\theta(U)$ is a $10$-qubit gate that decomposes into a sequence of $C=70$ $CCZ$ gates, as well as the Cliffords $Z$ and $CZ$. Explicitly,
\begin{gather*}
    \theta(U) = Z_1Z_2Z_3Z_4Z_5Z_6Z_7Z_8Z_9Z_{10}CZ_{1,2}CZ_{1,5}CZ_{1,7}CZ_{1,10}CZ_{2,3}CZ_{2,6}CZ_{2,8}\\CZ_{3,4}CZ_{3,7}CZ_{3,9}CZ_{4,5}CZ_{4,8}CZ_{4,10}CZ_{5,6}CZ_{5,9}CZ_{6,7}CZ_{6,10}CZ_{7,8}CZ_{8,9}CZ_{9,10}CCZ_{1,2,3}\\CCZ_{1,2,5}CCZ_{1,2,6}CCZ_{1,2,8}CCZ_{1,2,9}CCZ_{1,2,10}CCZ_{1,3,4}CCZ_{1,3,5}CCZ_{1,3,8}CCZ_{1,3,9}\\CCZ_{1,4,5}CCZ_{1,4,6}CCZ_{1,4,10}CCZ_{1,5,10}CCZ_{1,6,7}CCZ_{1,6,9}CCZ_{1,7,8}CCZ_{1,7,9}CCZ_{1,7,10}CCZ_{1,8,10}\\CCZ_{1,9,10}CCZ_{2,3,4}CCZ_{2,3,6}CCZ_{2,3,7}CCZ_{2,3,9}CCZ_{2,3,10}CCZ_{2,4,5}CCZ_{2,4,6}CCZ_{2,4,9}CCZ_{2,4,10}CCZ_{2,5,6}CCZ_{2,5,7}\\CCZ_{2,7,8}CCZ_{2,7,10}CCZ_{2,8,9}CCZ_{2,8,10}CCZ_{3,4,5}CCZ_{3,4,7}CCZ_{3,4,8}CCZ_{3,4,10}CCZ_{3,5,6}CCZ_{3,5,7}\\CCZ_{3,5,10}CCZ_{3,6,7}CCZ_{3,6,8}CCZ_{3,8,9}CCZ_{3,9,10}CCZ_{4,5,6}CCZ_{4,5,8}CCZ_{4,5,9}CCZ_{4,6,7}CCZ_{4,6,8}\\CCZ_{4,7,8}CCZ_{4,7,9}CCZ_{4,9,10}CCZ_{5,6,7}CCZ_{5,6,9}CCZ_{5,6,10}CCZ_{5,7,8}CCZ_{5,7,9}CCZ_{5,8,9}CCZ_{5,8,10}\\CCZ_{6,7,8}CCZ_{6,7,10}CCZ_{6,8,9}CCZ_{6,8,10}CCZ_{6,9,10}CCZ_{7,8,9}CCZ_{7,9,10}CCZ_{8,9,10}.
\end{gather*}

\end{document}